\documentclass[sigconf, nonacm]{acmart}
\usepackage{hyperref}

\usepackage{cleveref}
\usepackage{etoolbox}
\newtoggle{fullversion}
\toggletrue{fullversion}
\usepackage{adjustbox}
\usepackage{caption}
\usepackage{scalerel}
\usepackage{graphicx}
\usepackage{xcolor}
\usepackage{tcolorbox}
\usepackage{amsmath,amsfonts}
\usepackage{ifsym}
\usepackage{mdframed}
\usepackage{listings}
\usepackage{wrapfig}

\renewcommand{\baselinestretch}{0.985}

\makeatletter
\let\c@lofdepth\relax
\let\c@lotdepth\relax
\makeatother
\usepackage{subfigure}
\tcbset{before={\par\pagebreak[0]\smallskip\parindent=0pt},after={\par\noindent}}
\usepackage{tree-dvips}
\usepackage{multirow}

\usepackage{qtree}
\newcommand{\sys}{\textsc{BLIP}\xspace}
\usepackage[ruled,linesnumbered,vlined]{algorithm2e}
\usepackage{algorithmic}
\usepackage{textcomp}
\usepackage{tikz}

\usetikzlibrary{positioning, shapes.geometric, arrows}
\usetikzlibrary{shapes, arrows.meta, fit, shapes.geometric, fit, backgrounds}
\definecolor{codegreen}{rgb}{0,0.6,0} 
\definecolor{darkblue}{rgb}{0.0, 0.0, 0.5}
\definecolor{darkblueilp}{RGB}{0, 0, 175}
\definecolor{darkred}{rgb}{0.5, 0.0, 0.0}
\definecolor{darkgreen}{rgb}{0.0, 0.5, 0.0}
\definecolor{lightgray}{gray}{0.95}
\newcommand{\appendixtext}[1]{#1}

\tcbuselibrary{listings, skins}
\usepackage[disable]{todonotes}
\usepackage[T1]{fontenc}
\usepackage{listings}

\newcommand{\topic}[1]{\vspace{.5pt} \noindent{\bf #1.}}

\newif\ifshowblock

\newcommand{\vldb}[1]{#1}
\newcommand{\sigmod}[1]{}

\newcommand{\techreport}[1]{#1}
\newcommand{\trs}[1]{#1\ignorespaces}
\newcommand{\papertext}[1]{}
\newcommand{\rone}[1]{{\color{black}#1}}
\newcommand{\rtwo}[1]{{\color{black}#1}}
\newcommand{\rthree}[1]{{\color{black}#1}}
\newcommand{\rmix}[1]{{\color{black}#1}}
\showblocktrue

\newcommand{\cover}[1]{}

\definecolor{myForestGreen}{RGB}{34,139,34}

\usepackage[normalem]{ulem}
\newcommand{\delete}[1]{\sout{}}
\newcommand{\cradd}[1]{\textcolor{black}{#1}}

\newcommand{\sep}[1]{\textcolor{red}{Sep: #1}}

\newcommand{\edit}[1]{\textcolor{black}{#1}}
 \newcommand{\anonymous}[2]{#2} 
\newcommand{\revised}[1]{\textcolor{black}{#1}}

\usepackage{xcolor}
\usepackage{tcolorbox}

\usepackage[normalem]{ulem}

\usepackage{tikz}

\newcounter{definition}
\newenvironment{definition}[1][]{\refstepcounter{definition}\par\smallskip\textsc{Definition~\thedefinition.\ #1}}{\smallskip}

\newcounter{problem}
\newenvironment{problem}[1][]{\refstepcounter{problem}\par\smallskip\textsc{Problem~\theproblem.\ #1}}{\smallskip}

\newcounter{assumption}

\newcommand{\code}[1]{\texttt{\small #1}}

\newcounter{example}
\newenvironment{example}[1][]{\refstepcounter{example}\par\smallskip\textsc{Example~\theexample.\ #1}}{$\square$\smallskip}

\newcounter{theorem}
\newenvironment{theorem}[1][]{\refstepcounter{theorem}\par\smallskip\textsc{Theorem~\thetheorem.\ #1}}{\smallskip}

\newcounter{proposition}

\usepackage{enumitem}
\newcommand{\squishlist}{
	\begin{list}{$\bullet$}
		{
			\setlength{\itemsep}{0pt}
			\setlength{\parsep}{1pt}
			\setlength{\topsep}{1pt}
			\setlength{\partopsep}{0pt}
			\setlength{\leftmargin}{1.5em}
			\setlength{\labelwidth}{1em}
			\setlength{\labelsep}{0.5em} } }
	
\newcommand{\squishend}{\end{list}}

\newcommand{\squishlistnum}{
	\begin{enumerate}
		{
			\setlength{\itemsep}{0pt}
			\setlength{\parsep}{1pt}
			\setlength{\topsep}{1pt}
			\setlength{\partopsep}{0pt}
			\setlength{\leftmargin}{1.5em}
			\setlength{\labelwidth}{1em}
			\setlength{\labelsep}{0.5em}
		}
}

\newcommand{\squishendnum}{\end{enumerate}}
\usepackage{multirow}

\usepackage{amsthm}
\usepackage{mathrsfs}

\newcommand\vldbdoi{10.14778/3836663.3836668}
\newcommand\vldbpages{2992 - 3005}
\newcommand\vldbvolume{19}
\newcommand\vldbissue{11}
\newcommand\vldbyear{2026}
\newcommand\vldbauthors{\authors}
\newcommand\vldbtitle{\shorttitle} 
\newcommand\vldbavailabilityurl{https://github.com/yiminl18/BLIP.git}
\newcommand\vldbpagestyle{empty}

\lstdefinestyle{SQLStyle}{
  language=SQL,
  showspaces=false,
  basicstyle=\ttfamily\footnotesize,
  commentstyle=\color{gray},
  mathescape=true,
  numbers=none,
  escapeinside={^}{^},
  captionpos=b,
  mathescape=false,
}
\newcommand{\mypython}[1]{%
  \begin{tcolorbox}[
    colback=gray!10,
    colframe=black,
    boxrule=0.5pt,
    arc=2pt,
    left=5pt,
    right=5pt,
    top=2pt,
    bottom=2pt
  ]
  \small\ttfamily #1
  \end{tcolorbox}
}

\NewDocumentCommand{\sidequote}{m +m O{black}}{%
 \begin{tcolorbox}[
enhanced,
breakable,
frame hidden,
interior hidden,
sharp corners,
borderline west={2.5pt}{0pt}{gray!50},
left=15pt,
right=0pt,
top=0pt,
bottom=0pt,
coltext=#3 
]
{\normalcolor\small\sffamily\bfseries #2}\par\medskip 
#1 
\end{tcolorbox}
}
 
\AtBeginDocument{%
  }

\setcopyright{none}
\renewcommand\footnotetextcopyrightpermission[1]{}

\begin{document}

\cover{\papertext{\title{Cover Letter for Revised Submission: Bolt-on, Verifiable Provenance for LLM-Powered Data Processing}}}

\cover{\papertext{\maketitle}}

\cover{\papertext{We thank the meta-reviewer and reviewers for their helpful feedback. We have carefully revised our paper as suggested. We first respond to meta-reviewer comments by providing a summary of changes to the paper,  followed by our responses to individual comments by reviewers. In the revised paper, we indicate changes made in response to reviewer R1 using \rone{blue}, to reviewer R2 using \rtwo{green}, and to reviewer R3 using \rthree{orange}. Changes addressing common comments raised from multiple reviewers are marked with \rmix{pink}. }}

\cover{\papertext{\input{metareviewer}}
\papertext{\input{R1}}
\papertext{\input{R2}}
\papertext{\input{R3}}}

\cover{\newpage
\setcounter{section}{0}
\setcounter{figure}{0}}

\title{Bolt-on, Verifiable Provenance for LLM-Powered Data Processing}



\author{Yiming Lin}
\affiliation{%
  \institution{University of California, Berkeley}
}
\email{yiminglin@berkeley.edu}

\author{Sepanta Zeighami}
\affiliation{%
  \institution{University of California, Berkeley}
}
\email{zeighami@berkeley.edu}

\author{Aditya G. Parameswaran}
\affiliation{%
  \institution{University of California, Berkeley}
}
\email{adityagp@berkeley.edu}


\begin{abstract}
Large Language Models (LLMs) are powerful tools for processing data. However, LLMs are also complex black-boxes, returning answers to queries on data, without any indication for where the answer came from or whether it is trustworthy. We introduce the notion of {\em provenance} for data processing with LLMs. While existing heuristics (such as embedding similarity or directly asking an LLM) could provide some hints for where the answer was derived, they provide no guarantees that the answer can be derived using the identified provenance, and indeed, are often incorrect. Instead, we propose the notion of {\em verifiable provenance} wherein we identify a subset of the input text that reproduces the same (or equivalent) answer as that on the complete text, and introduce the notion of {\em minimality}, where the verifiable provenance is as small as possible. To identify such a provenance, a naive solution would require checking all possible subsets of the source data with the LLM, which is prohibitively expensive. We present \sys, a bolt-on framework for efficiently inferring a small-sized verifiable provenance for any LLM-powered data processing task, with any LLM. As part of \sys, we introduce eight strategies, each guaranteed to find a minimal verifiable provenance, as well as an adaptive strategy that combines their strengths to reduce cost further. We further extend \sys to produce multiple minimal verifiable provenances. \revised{Experiments on \rthree{seven} datasets show that the provenance generated by \sys is always guaranteed to reproduce the answer—achieving over 30\% higher accuracy than the best-performing baseline with a comparable provenance size. Moreover, \sys incurs a low cost, comparable to the original query on the original data.   }
\end{abstract}

\maketitle

\pagestyle{\vldbpagestyle}
\begingroup\small\noindent\raggedright\textbf{PVLDB Reference Format:}\\
\vldbauthors. \vldbtitle. PVLDB, \vldbvolume(\vldbissue): \vldbpages, \vldbyear.\\
\href{https://doi.org/\vldbdoi}{doi:\vldbdoi}
\endgroup
\begingroup
\renewcommand\thefootnote{}\footnote{\noindent
This work is licensed under the Creative Commons BY-NC-ND 4.0 International License. Visit \url{https://creativecommons.org/licenses/by-nc-nd/4.0/} to view a copy of this license. For any use beyond those covered by this license, obtain permission by emailing \href{mailto:info@vldb.org}{info@vldb.org}. Copyright is held by the owner/author(s). Publication rights licensed to the VLDB Endowment. \\
\raggedright Proceedings of the VLDB Endowment, Vol. \vldbvolume, No. \vldbissue\ %
ISSN 2150-8097. \\
\href{https://doi.org/\vldbdoi}{doi:\vldbdoi} \\
}\addtocounter{footnote}{-1}\endgroup

\ifdefempty{\vldbavailabilityurl}{}{
\vspace{.3cm}
\begingroup\small\noindent\raggedright\textbf{PVLDB Artifact Availability:}\\
The source code, data, and/or other artifacts have been made available at \url{\vldbavailabilityurl}.
\endgroup
}

\vspace{-2mm}
\section{Introduction}
\label{sec:introduction}
Large Language Models (LLMs) have shown impressive capabilities
in understanding, processing, and generating data, 
driving a wide range of data processing tasks, including data extraction, 
transformation, cleaning, and enrichment, as well as question-answering~\cite{zhao2024chat2data, huang2024transform, naeem2024retclean, buss2023generating, zhu2024autotqa} over structured~\cite{jin2022survey,nan2022fetaqa} and unstructured data~\cite{zhuang2023toolqa}. 
\rone{Recent declarative systems~\cite{shankar2024docetl,lin2024towards,patel2024lotus,liu2024declarative,wang2025unify,jo2024thalamusdb,satriani2025logical,sun2025quest,patel2024semantic,russo2025abacus}}, 
as well as SQL extensions to
data warehouses~\cite{databricks-llm,snowflake-llm, duckdb-llm,alloydb-llm} 
enable LLM-powered operators over textual columns, 
supporting similar tasks such as extraction, cleaning, and synthesis.
Most such LLM-powered data processing tasks or operators 
boil down to using an LLM 
to synthesize an answer for a given question 
from source data provided in-context, 
which could include 
structured or unstructured data. 
For example, data extraction can be framed as 
a natural language question, with the source data (documents or tables) provided in context. Or,  
in tabular question answering (TableQA)~\cite{jin2022survey,nan2022fetaqa}, 
the question is placed in the prompt, with
the relevant tables provided in context. 
In the following,
we refer to this context or source data as {\em text}, 
since both structured and unstructured data are converted 
into a textual form for use as context, comprising
one or more {\em sentences},
and to the data processing task or operator, expressed in natural language, as a {\em question}. 


When using LLMs for data processing, 
 it is {\em not clear how the LLM
response may have been derived, and if
the response can be trusted}.
Although LLM capabilities continue to improve, they are far from perfect~\cite{zhang2024benchmarking,liu2024lost}, often with accuracy less than 75\% in domains such as law, medicine, and finance~\cite{liang2023holistic, guha2023legalbench, islam2023financebench, arora2025healthbench, liu2024large}. 
Gaining human trust in how answers
are produced and whether those answers are correct can be critical, 
especially when a wrong answer may lead to serious consequences. 
Consider a real-world use case: 
analysis of police use-of-force records, 
\anonymous{provided by our collaborator X\footnote{Name omitted for anonymity.}}{from the Police Records Access project, a collaboration that we are part of, along with various journalism and legal  partners~\cite{clean} to build the first-ever statewide
use-of-force database\footnote{Available at: \url{https://clean.sfchronicle.com} or \url{https://clean.latimes.com}}}.  
As in Figure~\ref{fig:example}, journalists aim to extract
names of officers involved in the use of force, from a 35-page document, and repeat this for 
over 3500 cases in the project, spanning over 1.5M pages~\cite{clean}.    
The LLM \code{gpt-4o-mini} returns the answer {\em Officer Anderson, Maya}. 
Even if the LLM has high overall accuracy,
that does not suffice to determine correctness for this specific query, 
where errors are especially likely
with large context sizes ~\cite{liu2024lost,bai2023longbench}. 
To ensure correctness,
our journalism collaborators 
often manually review each text document (of dozens of pages).
Here, an engineer on the project
developed a custom keyword search approach to identify potentially relevant
text portions for the answer (e.g., mentions of ``Officer'', ``Captain''), 
reducing review time by 7$\times$, but introducing many false positives and negatives. 
Overall, 
as the text size increases, 
verifying LLM answers becomes increasingly challenging. 
Ensuring human trust in answers is critical across domains, including analysis over medical, legal, financial, scientific, governmental, and journalistic documents. 
Effective mechanisms for verifying answer correctness 
are key to ensuring  usability of LLM-powered applications.



\begin{figure}[tb]
    \centering
    \includegraphics[width=0.9\linewidth]{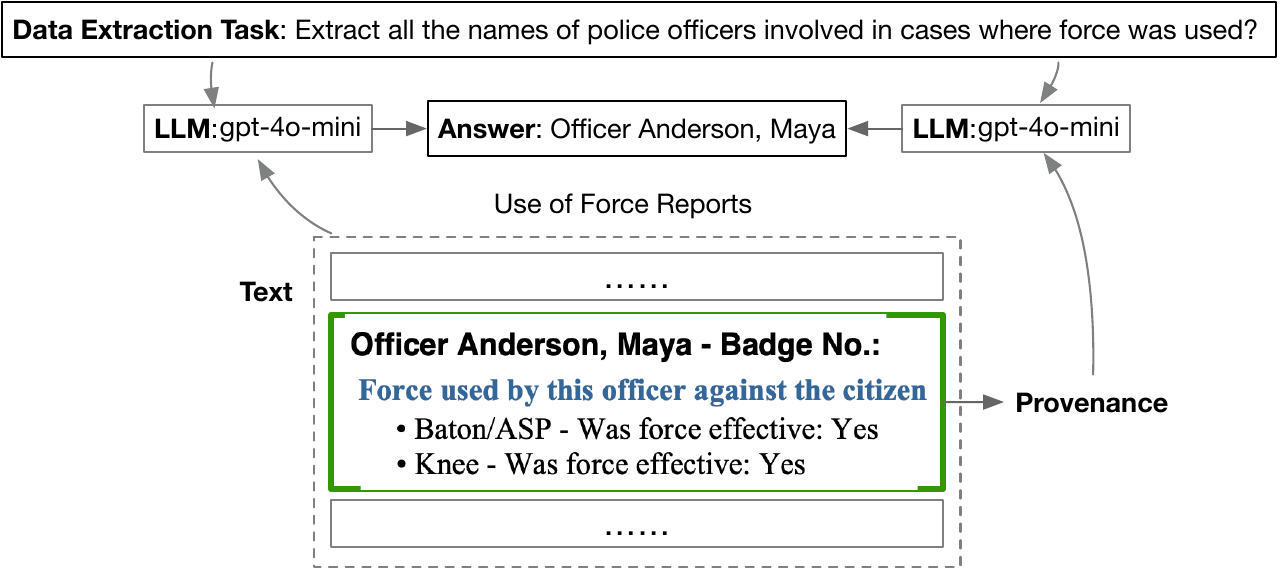}
    \vspace{-1em}
    \caption{\small Verifiable Provenance in Police Use of Force Records. (Actual values in police records have been replaced for privacy.)}  
    \vspace{-1em}
    \label{fig:example}
\end{figure}


In this paper, we revisit 
a traditional database concept in the context of verification of LLM-generated answers: {\em provenance},
which, for now, we will informally refer to as a text portion 
from which the answer may have been derived. 
In the use of force analysis in Figure~\ref{fig:example}, 
if supporting {\em evidence} is provided, 
such as a fragment of the document that mentions the use of force types and 
their associations with the officer’s name, 
users can understand how the answer may have been generated 
and better assess its correctness.

\topic{Current approaches to infer provenance are not reliable}  
To infer the provenance of an answer, one could leverage 
\revised{retrieval approaches by} identifying the top-$k$ sentences in the given text 
with the highest relevance scores to the question~\cite{sankararaman2024provenance,gao2023enabling,cohen2024contextcite}. 
These scores can be computed post hoc,
using embedding similarities~\cite{gao2023enabling} or next-token probabilities~\cite{cohen2024contextcite}.
However, these approaches do not guarantee
that the provenance is sufficient to reproduce the answer. 
\rthree{Table~\ref{tab:agreement} shows the accuracy of one such retrieval method based on embedding similarity, i.e., via RAG~\cite{lewis2020retrieval},  across five workloads: \code{Qasper}~\cite{dasigi2021dataset}, \code{NL\_DEV}~\cite{kwiatkowski2019natural},  \code{HotpotQA}~\cite{yang2018hotpotqa},  \code{CUAD}~\cite{hendrycks2021cuad}, and \code{PubMedQA}~\cite{jin2019pubmedqa}, each with 500 question-document pairs.} We vary the size of retrieved context (i.e., the inferred provenance) between 1\%, 5\%, and 10\% of the full text. 
To assess whether the retrieved context can reproduce the answer, we ask the same question to the same LLM (gpt-4o) with this context, and report accuracy as the fraction of cases where the resulting answer agrees with the one generated from the full text.\footnote{If the answers aren't identical, we use LLM-as-a-judge~\cite{gu2024survey,zheng2023judging} to determine accuracy, which is known to be effective at identifying semantically equivalent answers.} Unfortunately, {\em even when the retrieved context includes 10\% of the full text, accuracy is below 0.8}, making it neither reliable nor effective for human review.



\begin{table}[bt]
\small
\centering
\scalebox{0.76}{
\begin{tabular}{l|l|c|c|c|c|c}
\hline
\textbf{Method} & \textbf{Configuration} & \textbf{Qasper} & \textbf{NL\_DEV} & \textbf{HotpotQA} & \textbf{\rthree{CUAD}} & \textbf{\rthree{PubMedQA}} \\ \hline
\multirow{3}{*}{\textbf{Retrieval}} 
& \textbf{RAG-1\%} & 0.19 & 0.17 & 0.12 & \rthree{0.23} & \rthree{0.09} \\ 
& \textbf{RAG-5\%} & 0.46 & 0.41 & 0.35 & \rthree{0.39} & \rthree{0.24} \\ 
& \textbf{RAG-10\%} & 0.78 & 0.69 & 0.55  & \rthree{0.62} & \rthree{0.49}\\ \hline
\multirow{3}{*}{\textbf{LLM}} 
& \textbf{gpt-4o-mini} & 0.62 & 0.68 & 0.65 & \rthree{0.59} & \rthree{0.71}\\ 
& \textbf{gpt-4o} & 0.65 & 0.77 & 0.79 & \rthree{0.71} & \rthree{0.76} \\ 
& \textbf{gemini-2-flash} & 0.41 & 0.37 & 0.61 & \rthree{0.55} & \rthree{0.68}\\ \hline
\end{tabular}}
\caption{\small \revised{Accuracy of retrieval-based and LLM-generated provenance.}}
\label{tab:agreement} 
\vspace{-3.5em}
\end{table}




Another method is to directly
ask the LLM to identify the provenance of its answer
from the text, 
by prompting it to provide the line numbers
of the relevant sentences (by prefixing each sentence  with line numbers), 
as listed \papertext{in \vldb{~\cite{blip}.}} \techreport{in \code{LLM-provenance-prompt} below.} 
\revised{Unfortunately, as shown in Table~\ref{tab:agreement},} \code{gpt-4o-mini}, \code{gpt-4o}, and \code{gemini-2-flash} return provenance that reproduces the answer in only 65\%, 74\%, and 53\% of the cases, respectively---showing that LLMs are unable to identify the provenance 
that can reproduce the same (or equivalent) answer. \revised{When a provenance fails to reproduce the answer (e.g., those generated by retrieval or LLMs), human verification defaults to reading the entire document, significantly increasing human effort. 
} 

\techreport{
\noindent\underline{\code{LLM-provenance-prompt}}: \textcolor{blue}{[Question]}, \textcolor{blue}{[Answer]}, and \textcolor{blue}{[Context]} are placeholders for question $Q$, answer $A$, and the full text $T$.  }

\techreport{
\mypython{
\small 
    LLM-Provenance-Prompt: Given the following question: \textcolor{blue}{[Question]}, the corresponding answers are: \textcolor{blue}{[Answer]}. Your task is to extract the set of sentences from the provided context that contribute to generating these answers. Identify the most relevant sentences that support the given answers. Do not add explanations. Only return a list of sentence IDs. Do not return any words. The context is as follows: \textcolor{blue}{[Context]} }
}

Finally, other approaches~\cite{lee2025llm} 
require access to the model to compute attribution scores
for tokens based on training data but rely on access to training data, which 
is unavailable for black-box LLMs (e.g., OpenAI models) and frequently withheld for open-source models (e.g., Llama3). 
They also don't apply when operating on user-provided text, which is not part of the training data. 


\topic{Verifiable provenance}
We use insights from the previous experiment
to introduce the concept of {\em verifiable provenance} over text  
for LLM-powered data processing tasks.
Instead of arbitrarily selecting a set of sentences 
based on their relevance to the question and answer, 
we require the provenance to be {\em verifiable}, 
i.e., the provenance is a subsequence of the text (e.g., green box in Figure~\ref{fig:example}) 
that can {\em reproduce the same (or equivalent) 
answer as the one obtained from the full text when queried using the same LLM}. 
In the example in Figure~\ref{fig:example}, 
the provenance in the green box 
is a subsequence of the original text, 
and asking the same question on just this provenance with the same LLM \code{gpt-4o-mini} 
leads to the same answer, {\em Officer Anderson, Maya}. 
Instead, the results in Table~\ref{tab:agreement} 
show that roughly {\em half of the LLM-generated provenance is not verifiable}. 

\topic{Conciseness of verifiable provenance and its challenges}
To help users easily assess answer correctness, 
the verifiable provenance must additionally be {\em small-sized}. 
Otherwise, the original text could itself serve as a provenance, but reviewing this is difficult and time-consuming.  
Identifying a verifiable  provenance that is also concise is nontrivial. 
One naive approach would involve considering all subsets of $k \in [0, n]$ sentences 
of the text with $n$ sentences \revised{and verifying whether the enumerated subset produces an equivalent answer with the LLM, which would scale as $2^n$. }
Another approach would 
start from the given text as the initial provenance, and 
iteratively remove a sentence at a time  
if the remaining still produces the same answer 
when using the same LLM---until no further sentences can be removed.
However, this approach incurs significant cost and latency, 
as it requires a number of LLM calls 
proportional to the number of sentences in the text. 
This baseline, on the Qasper dataset  with 9,000 tokens on average, 
takes $>$2 hours, incurring  $>$100$\times$ the 
cost of using an LLM for answering the question on the original text.

\topic{Merely a \sys to identify verifiable, concise provenance}  
We introduce \sys\footnote{Short for {\bf B}olt-on {\bf L}LM ver{\bf I}fiable {\bf P}rovenance.},
{\em a provenance discovery framework comprising 
a family of intuitive strategies to efficiently infer small-sized verifiable provenance}. 
We define the notion of {\em minimality} of verifiable provenance, 
where removing any sentence from the provenance 
results in a different answer.  
We further define
the notions of {\em strong} and {\em weak monotonicity} to
characterize the types of tasks for which
minimal provenance
provides additional guarantees, while
empirically showing that those tasks
dominate our workloads---and contrast
these notions  to those in the provenance
literature~\cite{cheney2009provenance,buneman2001and,green2007provenance,glavic2021data}.

As part of \sys, 
we present a two-phase approach that is {\em guaranteed to find a minimal verifiable provenance}. 
The first phase aggressively returns a small-sized verifiable provenance 
that may not be minimal.  
In this phase, we explore four strategies, 
including ranking text portions
based on their relevance to the question and answer, 
as well as different enumeration orders. 
The second phase takes the provenance returned from the first phase and  guarantees a minimal provenance. This two-phase approach, though intuitive, exhaustively explores 10 strategy combinations, each producing a minimal provenance. 
Using a cost model that we introduce, 
we present a theoretical cost comparison, and propose an adaptive strategy  
that integrates their best features. 
We further employ KV-caching~\cite{cai2024pyramidkv,liu2024minicache,adnan2024keyformer} to optimize each individual strategy to minimize inference cost. 
Finally, we extend \sys to support the discovery of multiple minimal
provenances rather than just one, to support use cases
where multiple pieces of evidence are helpful. 



Overall, \sys can be applied in a bolt-on manner, 
is compatible with any LLM model, 
and offers a theoretical guarantee 
of minimal provenance at a low cost.  Our contributions include:
\begin{itemize}[leftmargin=*]
    \item We {\em \revised{are the first} to introduce the notion of verifiable provenance for LLM-powered data processing}, outline key properties, including verifiability, minimality, and monotonicity, and 
    characterize the family of tasks where verifiable provenance applies.  
    \item We conduct an exhaustive exploration of intuitive  but effective strategies to efficiently infer minimal provenance at low cost.   
    \item We present a theoretical analysis of correctness and cost of the strategies, and use it to develop an adaptive strategy. 
    \item We perform comprehensive evaluations across \rthree{seven} datasets, demonstrating that it is possible to infer minimal provenance with  {\bf guaranteed 100\% accuracy}, achieving \textbf{over 30\%} higher accuracy than the best-performing baseline, at a low cost—comparable to that of performing the original task. 
\end{itemize}
\noindent We define the notion of provenance and present an  overview of \sys in Section~\ref{sec:def}. We then present our approaches for finding a single minimal provenance and top-$k$ minimal provenances in Section~\ref{sec:1-provenance} and  Section~\ref{sec:k-provenance}, respectively. We  evaluate \sys in Section~\ref{sec:exp}. 



\section{Problem Definition and Overview} 
\label{sec:def}
\revised{We begin by formally defining the concept of provenance (Section~\ref{sec:prov-def}) and outlining our cost model (Section~\ref{subsec:cost-model}), followed by a description of our problem statement (Section~\ref{subsec:problem-def}). We further define the applicable tasks (Section~\ref{subsec:strong-weak}) and questions (Section~\ref{subsec:database-prov-relationship}) for which our provenance is useful, and discuss the impact of LLM behavior (e.g., non-determinism) on provenance (Section~\ref{subsec:llm-uncertainty}). 
We finally present an overview of \sys (Section~\ref{subsec:overview}).  }

\vspace{-1mm}
\subsection{Provenance Definition} \label{sec:prov-def}
Consider a data processing task generally represented as a natural language {\em question} $Q$. 
Let $L$ be the Large Language Model (LLM), and $T$ be the input {\em text}, 
i.e., source data to process.  
$T$ consists of a sequence of sentences, 
i.e., $T = \langle s_1, s_2, \ldots, s_n \rangle$, 
where $i$ is the {\em index} of sentence $s_i$. 
Let $A$ be the {\em response} (or answer) 
to question $Q$ in text $T$ using LLM $L$, 
denoted as $L(T, Q) = A$. 
Let $P \sqsubseteq T$ be a subsequence of $T$, i.e., 
$P = \langle s_{i_1},s_{i_2},...,s_{i_m} \rangle$, where $1 \leq i_1 < ... < i_m \leq n$. 

\begin{definition}
    [Verifiable Provenance.]~\label{def:provenance} Given text $T$, question $Q$, and LLM $L$, where $L(T, Q) = A$, we define the subsequence $P\sqsubseteq T$ to be a {\em verifiable provenance} if $L(P, Q) = A'$, where $I(A, A') = \text{True}$. 
\end{definition}
\vspace{-3mm}

\noindent The indicator function $I(A, A')$ returns {\em True} if $A$ and $A'$ are lexically identical or semantically equivalent. 
For settings where $A$ 
is drawn from a closed domain (e.g., true/false or zip code), 
users can opt for $I(A, A')$ as $A$=$A'$. Otherwise, 
we leverage an LLM as a judge~\cite{gu2024survey,zheng2023judging} to evaluate
semantic equivalence \revised{(e.g., ``2012'' and ``year 2012''  judged equivalent)}.  
This approach effectively handles equivalent but non-identical LLM-generated responses. \rmix{To show this, we conducted an evaluation on 2000 text pairs from the GLUE benchmark~\cite{wang2018glue}, where each pair contains two text snippets with a label indicating  equivalence. We evaluated LLM-as-a-judge using five LLMs with different capabilities and four prompts, including a handwritten prompt and an LLM-rewritten variant, each with and without examples. \papertext{Details of prompts and experimental setup can be found in ~\cite{blip}}. \techreport{Details of prompts and experimental setup can be found below.} 
As shown in Table~\ref{tab:llm-as-a-judge}, LLM-as-a-judge is effective with around 94\% accuracy on average across models.} 
\techreport{The LLM used to evaluate $I(A, A')$ does not 
need to be the one used to answer $Q$ on $T$. }


\techreport{\rmix{\topic{\code{response\_equal\_human}}: \textcolor{blue}{[sentence\_1]} and \textcolor{blue}{[sentence\_2]} are placeholders for two test sentences. 
\mypython{
\small 
    llm\_equal\_human: Are the following two sentences semantically equivalent? Please respond with True if they are, and False if they are not. Sentence 1: \textcolor{blue}{[sentence\_1]}. Sentence 2: \textcolor{blue}{[sentence\_2]}. }}}

\techreport{\rmix{\topic{\code{response\_equal\_human\_example}}: \textcolor{blue}{[sentence\_1]} and \textcolor{blue}{[sentence\_2]} are placeholders for two test sentences. 
\mypython{
\small 
    llm\_equal\_human\_example: Are the following two sentences semantically equivalent? Please respond with True if they are, and False if they are not. Examples:
    Example 1:
    Sentence 1: The cat is sleeping on the sofa.
    Sentence 2: A cat is lying on the couch asleep.
    Answer: True
    Example 2:
    Sentence 1: The company reported a profit of 2 million dollars.
    Sentence 2: The company reported a loss of 2 million dollars.
    Answer: False 
Sentence 1: \textcolor{blue}{[sentence\_1]}. Sentence 2: \textcolor{blue}{[sentence\_2]}. }}}

\techreport{\rmix{\topic{\code{response\_equal\_llm}}: \textcolor{blue}{[sentence\_1]} and \textcolor{blue}{[sentence\_2]} are placeholders for two test sentences. 
\mypython{
\small 
    llm\_equal\_2: Determine whether the following two sentences are semantically equivalent. Return True if they are, and False otherwise. Are the following two sentences semantically equivalent? 
Two sentences are semantically equivalent if they express the same meaning, even if the wording is different. 
Minor paraphrasing, reordering, or synonym replacement does NOT change equivalence. 
Differences in factual content, numerical values, negation, or core claims mean they are NOT equivalent. 
Sentence 1: \textcolor{blue}{[sentence\_1]}. Sentence 2: \textcolor{blue}{[sentence\_2]}. } } }

\techreport{\rmix{\topic{\code{response\_equal\_llm\_example}}: \textcolor{blue}{[sentence\_1]} and \textcolor{blue}{[sentence\_2]} are placeholders for two test sentences. 
\mypython{
\small 
    llm\_equal\_2: Determine whether the following two sentences are semantically equivalent. Return True if they are, and False otherwise. Are the following two sentences semantically equivalent? 
Two sentences are semantically equivalent if they express the same meaning, even if the wording is different. 
Minor paraphrasing, reordering, or synonym replacement does NOT change equivalence. 
Differences in factual content, numerical values, negation, or core claims mean they are NOT equivalent. Examples: 

Example 1:
Sentence 1: She did not attend the meeting.
Sentence 2: She skipped the meeting.
Answer: True

Example 2:
Sentence 1: The experiment was conducted in 2020.
Sentence 2: The experiment was conducted in 2021.
Answer: False 
Sentence 1: \textcolor{blue}{[sentence\_1]}. Sentence 2: \textcolor{blue}{[sentence\_2]}. } } }


Intuitively, any subsequence of $T$ that returns (lexically or semantically) equivalent answer to $A$ using the same LLM is considered a verifiable provenance. 
When obvious from  context,  we refer to a
verifiable provenance simply as {\em provenance}. 
In practice, users prefer the provenance to be small (i.e., with fewer sentences) for easier interpretation. Thus, we define the notion of a {\em minimal provenance}. 

\begin{definition}
    [Minimal Provenance.]~\label{def:minimal-provenance} A provenance $P$ is {\em minimal}, if $\forall s\in P, I(L(P\setminus s, Q), A) = False$. 
\end{definition}

\noindent That is, a provenance $P$ is minimal if removing any sentence from $P$ results in an answer that is not equivalent to $A$. 

\begin{definition}
    [Strictly-Minimal Provenance.]~\label{def:strict-minimal-provenance} A provenance $P$ is {\em strictly minimal}, if $\nexists P' \sqsubset P$,  $I(L(P', Q), A) = True$.
\end{definition}

\noindent A provenance $P$ is strictly minimal if there does not exist any subsequence of $P$ that is a provenance. Any strictly minimal provenance is also minimal, by definition. 
In practice, a minimal provenance is often sufficient, as a small-sized, rather than strictly minimal, provenance is typically adequate
for interpreting answers. 
As we will show shortly, under a commonly 
observed assumption for LLMs, 
a minimal provenance is equivalent to a strictly minimal one.

\begin{table}[bt]
\small 
\centering
\scalebox{0.7}{
\begin{tabular}{l|c|c|c|c|c}
\hline
 & \textbf{gpt-4o} & \textbf{gpt-4o-mini} & \textbf{gpt-5} & \textbf{gpt-5-mini} & \textbf{gemini-2-flash} \\ \hline
{\em response\_equal\_human} & 96.3\% & 91.2\% & 94.2\% & 90.4\% & 94.7\% \\ \hline 
{\em response\_equal\_human\_example} & 96.6\% & 90.8\% & 95.1\% & 90.9\% & 94.3\% \\ \hline 
{\em response\_equal\_llm} & 96.8\% & 90.8\% & 94.7\% & 91.1\% & 95.2\% \\ \hline 
{\em response\_equal\_llm\_example} & 96.4\% & 91.5\% & 93.4\% & 91.6\% & 95.7\% \\ \hline 
\end{tabular}}
\caption{\rmix{\small Effectiveness of LLM-as-a-judge.}} 
\label{tab:llm-as-a-judge}
\vspace{-3em}
\end{table}

\vspace{-1mm}
\subsection{Cost Model and KV-Cache} 
\label{subsec:cost-model}
Let $PT = \langle Q, T \rangle$ denote the prompt formed by concatenating the question $Q$ with the text $T$. We use $L(T, Q)$ and $L(PT)$ interchangeably throughout the paper when it does not lead to confusion. Let $c_{PT}$ represent the monetary cost of a single LLM invocation of $L$ with input $PT$, returning a response $L(PT)$. We have:
\begin{equation}
    c_{PT} = |PT|\cdot c_{in} + |L(PT)|\cdot c_{out}
\end{equation}
\noindent Here, $|\cdot|$ denotes the number of tokens, and $c_{in}$, $c_{out}$ refer to the unit cost of model $L$ for input and output tokens, respectively. For example, when $L$ is \code{gpt-4o-mini}, $c_{in}$ is \$$\frac{0.15}{10^6}$ and $c_{out}$ is \$$\frac{0.6}{10^6}$.

If two prompts $PT_i$ and $PT_j$ share a prefix, denoted $Pre_{i,j}$, then $Pre_{i,j}$  is considered cached tokens, employing the {\em KV-cache}~\cite{cai2024pyramidkv,liu2024minicache,adnan2024keyformer}. For example, with \code{gpt-4o-mini} and \code{gemini-2-flash}, cached tokens are 2$\times$ and 4$\times$ cheaper than uncached tokens, respectively.
In particular, if the prompt $PT_i$ has already been computed using $L$, then the cost of  inference on $PT_j$, denoted $c_{PT_j \mid PT_i}$, is
\begin{equation}
    c_{PT_j | PT_i} = \frac{|Pre_{i,j}|\cdot c_{in}}{f_L} + |PT_j \setminus Pre_{i,j}|\cdot c_{in} + |L(PT_j)|\cdot c_{out}
\end{equation}
\noindent Here, $f_L$ is the cost reduction factor for cached tokens in model $L$. For example, $f_{\text{gpt-4o-mini}} = 2$, while $f_{\text{gemini-2-flash}} = 4$.

\begin{table}[bt]
\scriptsize
\centering
\scalebox{0.9}{
\begin{tabular}{c|c||c|c|c|c}
\hline
\textbf{LLMs} & \textbf{Task} & \textbf{Qasper} & \textbf{NL\_DEV} & \textbf{HotpotQA} & \textbf{Average} \\ \hline
\multirow{2}{*}{\textbf{gpt-4o-mini}} 
 & SM & 97\% & 87\% & 88\%  & 90.7\% \\ 
 & WM & 99\% & 96\% & 98\%  & 97.7\% \\ \hline
\multirow{2}{*}{\textbf{gpt-4o}} 
 & SM & 98\% & 93\% & 91\%  & 94.0\% \\ 
 & WM & 100\% & 98\% & 98\%  & 98.7\% \\ \hline
\multirow{2}{*}{\textbf{gemini-2-flash}} 
 & SM & 95\% & 88\% & 81\%  & 88.0\% \\ 
 & WM & 97\% & 96\% & 97\%  & 96.7\% \\ \hline
\end{tabular}}
\caption{\small Percentage of Strongly-Monotonic (SM) and Weakly-Monotonic (WM) Tasks, with per‑row averages.}
\label{tab:tasks}
\vspace{-4em}
\end{table}

\subsection{Problem Definitions}
\label{subsec:problem-def} 

We are now ready to state our problem more formally:
\begin{problem}
    [Provenance Inference.]~\label{def:1-provenance} 
    Given text $T$, question $Q$, LLM $L$ and answer $A=L(T,Q)$, identify a minimal provenance. 
\end{problem}

\noindent The minimal provenance  
in Definition~\ref{def:minimal-provenance} for $\langle Q, A \rangle$ 
within text $T$ may not be unique. 
In our example in Figure~\ref{fig:example}, 
the same officers may be listed multiple times, in the internal affairs report, 
\trs{court hearings report,} as well as incident summary. 
Let $MP = \{mp_1, \dots, mp_n\}$ be the set of distinct minimal provenances within $T$.  
Although in most cases, identifying a single minimal provenance $mp_i \in MP$
is sufficient to provide evidence for human verification,
there are scenarios where multiple minimal provenances
may be helpful, as we will outline in Section~\ref{sec:k-provenance}.
We therefore define the {\em $k$-Provenance Inference} problem below. 
\begin{problem}
    [k-Provenance Inference.]~\label{def:k-provenance} Given text $T$, question $Q$, LLM $L$, and answer $A = L(T,Q)$, identify $k$ distinct minimal provenances,  $\{mp_{i1},mp_{i2},...,mp_{ik}\}$, where $\forall ij, mp_{ij} \in MP$. 
\end{problem}

\vspace{-2mm}
\subsection{Strong and Weak Monotonicity}\label{subsec:strong-weak}
As described previously, our goal, depending on the problem,
is to find one or $k$ minimal provenances.
There is a danger, however, of getting stuck 
at a ``local minimum'', where the provenance $P$ returned is minimal,
but there is a subsequence $P' \sqsubset P$ that is also a provenance---a scenario 
that is addressed by strict minimality.
Here, we introduce two properties and empirically study 
their adherence for our tasks.
We denote a task by the triple $\langle Q,T,L\rangle$.
\begin{definition}
    [Strongly Monotonic Tasks.]~\label{def:strong-mono} A task $\langle Q,T,L\rangle$ with $L(T, Q) = A$, is  {\em strongly monotonic} if for any provenance $P\sqsubset T$, $\forall P'$ where $\forall P\sqsubset P' \sqsubset T$, $P'$ is also a provenance. 
\end{definition}


\noindent A task is {\em strongly monotonic} 
if any superset of any provenance is also a provenance. 
While this may appear restrictive, 
we empirically show in Table~\ref{tab:tasks} 
that strongly monotonic tasks are still common, around 90\% on average
across real workloads with multiple LLMs. (Experimental setup is introduced later.) 
We next introduce the
notion of {\em weakly monotonic} tasks. 

\begin{definition}[Weakly Monotonic Tasks.]
\label{def:weak-mono}
A task $\langle Q,T,L\rangle$ with $L(T,Q)=A$ is {\em weakly monotonic} 
if for any provenance $P \sqsubset T$, where $|P| = i$ and $|T| = n$,
there is a sequence of provenances $P_{i+1}, P_{i+2}, ..., P_{n-1}$
such that $|P_{i+k}| = i+k$, and $P \sqsubset P_{i+1}\sqsubset  P_{i+2}\sqsubset  ...\sqsubset  P_{n-1}\sqsubset T$.
\end{definition}

\begin{figure}[tb]
    \centering
    \includegraphics[width=0.5\linewidth]{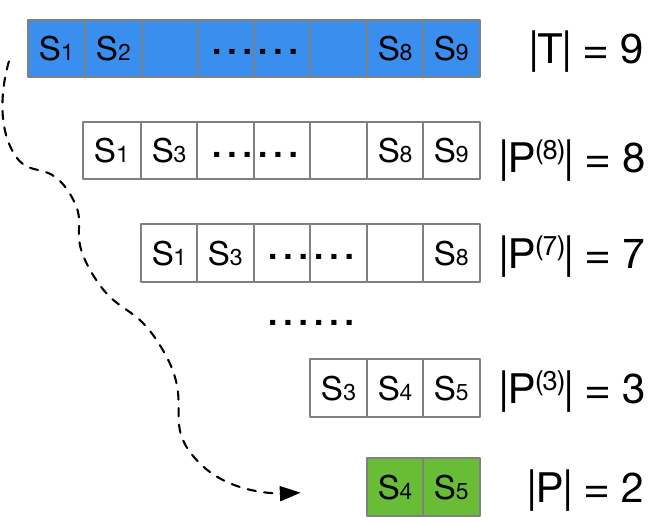}
    \vspace{-1em}
    \caption{\small Weak Monotonicity. }  
    \vspace{-1.8em}
    \label{fig:weak-mono}
\end{figure}
\noindent Instead of requiring that every
superset $P'$ of $P$ is a provenance
as in strong monotonicity, as shown in Figure~\ref{fig:weak-mono}, 
a weakly monotonic task simply requires a {\em chain of provenances}, 
covering every distinct size from $|P|$ to $|T|$, where each larger 
provenance is a superset of smaller ones. 
Weak monotonicity is useful  
because it implies the existence of a {\em path} 
to a minimal 
provenance (e.g., $\langle s_4, s_5\rangle$) 
from  $T$ (e.g., $\langle s_1, \dots, s_9\rangle$). 
We find that {\em weakly monotonic tasks are common,
with a prevalence of {\bf nearly over 95\%} across
workloads and LLMs}. 

We finally connect our notion of weak monotonicity
with the notion of minimal provenance, showing that, for weakly monotonic tasks,
a minimal provenance is also strictly minimal.
\begin{theorem}~\label{theo:minimal-eq}
    For a weakly monotonic task $\langle Q,T,L\rangle$, 
    all minimal provenances $P\sqsubseteq T$ are also strictly minimal. 
\end{theorem}

\noindent \papertext{The proof  
can be found in \vldb{\cite{blip}.}} \techreport{The proof  
can be found below. } 
\rmix{As we will see later, all our solutions for Problem~\ref{def:1-provenance}
are guaranteed to return a minimal provenance for
{\em any} task. } 
Theorem~\ref{theo:minimal-eq} {\bf \em additionally guarantees such minimal provenance is  
also strictly minimal}, 
as long as the task is weakly monotonic--- the dominant
class of tasks in our workload.

\techreport{\begin{proof}
    A strictly minimal provenance is also a minimal provenance by definition. Next, we prove that for a weakly monotonic task $\langle Q, T, L\rangle$, a minimal provenance $P \sqsubseteq T$ is also a strictly minimal provenance. Assume there exists a $P' \sqsubset P$ such that $P'$ is a provenance. We will prove that such a $P'$ does not exist, thereby establishing the correctness of the theorem by contradiction. Let $|P| = n$ and $|P'|=i$. Since $P$ is a minimal provenance, $\forall s \in P$, $P \setminus s$ is not a provenance. Thus, $|P'|$ must be less than $n - 1$. Since the task $\langle Q, T, L\rangle$ is weakly monotonic, there must exist a set of provenances from $P$ to $P'$, i.e., $P' \sqsubset P_{i+1} \sqsubset \dots \sqsubset P_{n-1} \sqsubset P$, where provenance $P_{i+1}$ has size $|i+1|$.  Since $\forall s \in P$, $P \setminus s$ is not a provenance, $\nexists P_{n-1} \sqsubset P$, which contradicts the weak monotonicity property. This completes the proof by contradiction. 
\end{proof}}

\rtwo{So far we have defined strongly and weakly monotonic tasks. We now briefly describe the experimental setup for the results in Table~\ref{tab:tasks}\papertext{, with details in our technical report~\cite{blip}.}\techreport{.} } 

\noindent\rtwo{\textbf{Monotonicity testing protocol.} To compute the percentage of strongly monotonic tasks, for each document-question pair, we begin with a minimal provenance $mp$. We then randomly sample 100 supersets of $mp$, with sizes uniformly distributed between $|mp| + 1$ and $|T| - 1$. We compute the percentage of supersets that reproduce an answer equivalent to $A$, and average this value across all document-question pairs for each dataset. }

\rtwo{To compute the percentage of weakly monotonic tasks, we similarly start from a $mp$. Let $|mp| = i$ and $|T| = n$. Let $\mathcal{P} = \{ P_{i+1}, P_{i+2}, ..., \\ P_{n-1} \}$ be a set of provenances such that $|P_{i+k}| = i+k$, and $P \sqsubset P_{i+1}\sqsubset  P_{i+2}\sqsubset  ...\sqsubset  P_{n-1}\sqsubset T$.  We perform a greedy exhaustive search to construct $\mathcal{P}$. Starting from $mp$, we enumerate candidates for $P_{i+1}$ and stop once one reproduces an answer equivalent to $A$. We then search for $P_{i+2}$, a superset of $P_{i+1}$, and continue this process. If no superset of $P_{i+1}$ with a size of $|i+2|$ yields a valid provenance, we backtrack to explore alternative candidates for $P_{i+1}$.  This continues until $\mathcal{P}$ is discovered\papertext{.} } \techreport{or if not, we treated this as a non-weakly-monotonic task, and we terminate it once the number of LLM-based tests exceeds $m = 500$ to avoid potentially high cost. }

\begin{table}[bt]
\papertext{\vspace{-5pt}}
\scriptsize
\centering
\scalebox{1}{
\begin{tabular}{c|c|c|c|c|c|c}
\hline
& \multicolumn{2}{c|}{\textbf{Qasper}} & \multicolumn{2}{c|}{\textbf{NL\_DEV}} & \multicolumn{2}{c}{\textbf{HotpotQA}}  \\
\hline
& str\_s & llm\_s &str\_s & llm\_s & str\_s & llm\_s \\ \hline 
\textbf{gpt-4o-mini} & 0.94 & \textbf{0.97} & 0.97  & \textbf{0.99}  & 0.99 & \textbf{0.995}  \\ 
\textbf{gpt-4o} & 0.93 & \textbf{0.99}  & 0.97  & \textbf{0.998}  & 0.99  & \textbf{0.998} \\ 
\textbf{gemini-2-flash} & 0.97  & \textbf{0.99}  & 0.96 & \textbf{0.97}  & 0.99 & \textbf{0.99}\\ \hline 
\end{tabular}}
\caption{\small LLM Deterministic Evaluation.}
\label{tab:llm-determininstic}
\vspace{-5em}
\end{table}


\begin{figure*}[tb]
\vspace{-15pt}
\begin{minipage}{0.25\textwidth}
    \centering
    \includegraphics[width=0.9\linewidth]{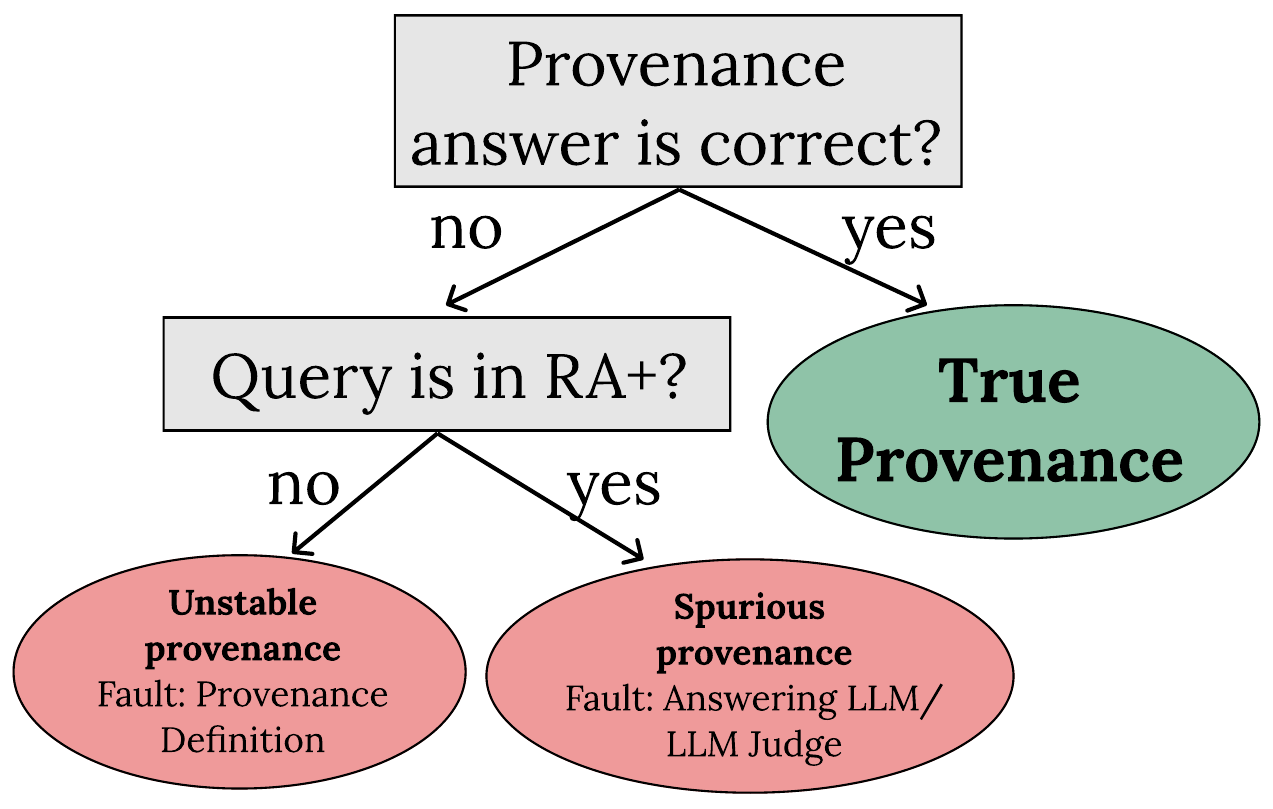}
    \vspace{-1.3em}
    \caption{\rmix{\small Provenance \\Correctness Taxonomy}}  
    \vspace{-1.3em}
    \label{fig:error_taxonomy}
\end{minipage} 
\begin{minipage}{0.7\textwidth}
    \centering
    \includegraphics[width=0.9\linewidth]{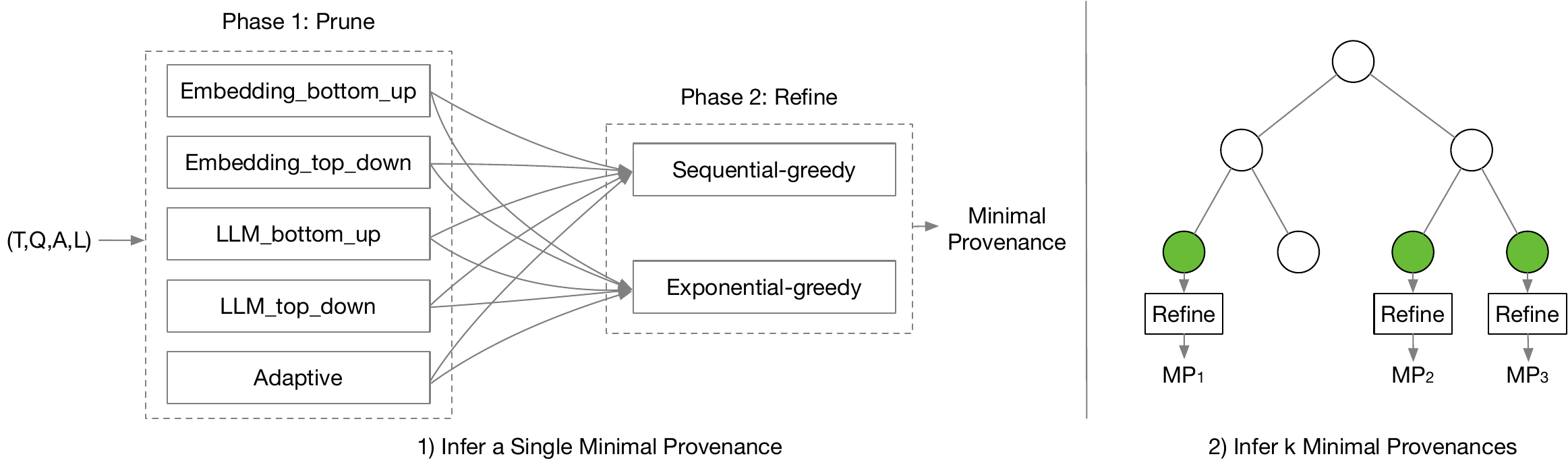}
    \vspace{-1.5em}
    \caption{\small Overview of \sys. $T$, $Q$, $A$, and $L$ denote the input text $T$, natural language question $Q$, answer (as evaluated) $A$, and the LLM model used $L$, respectively. $MP$ stands for minimal provenance.}  
    \vspace{-1.5em}
    \label{fig:overview}
\end{minipage} 
\end{figure*}

\vspace{-2mm}
\subsection{Provenance Correctness}
~\label{subsec:database-prov-relationship} 
\rmix{We just discussed minimality properties of verifiable provenance, but have not yet confronted whether the provenance matches user intent. Recall that a verifiable provenance is a subset of the input that is sufficient for an LLM to produce an answer as it would on the whole input. This definition, however, does not imply that the verifiable provenance is sufficient for a human to arrive at the same answer as the LLM.  We call a provenance that enables a human to arrive at the same answer as the LLM a \textit{true provenance}. 
}

 \rmix{We first discuss two scenarios where a true provenance is undefined and thus our provenance is inapplicable.  First, if the question is ambiguous, then the LLM's answer is not meaningful, making the provenance inapplicable. For example, when asking for ``{\em the revenue of Amazon}'' over its financial report, it is unclear whether the question refers to the monthly, quarterly, or yearly revenue. 
Second, if the LLM's answer (on the entire input) is incorrect, then any provenance is incorrect by definition---no provenance is sufficient for a human to arrive at the same answer as the LLM. }

\rmix{When a true provenance is well-defined, a verifiable provenance is \textit{correct} if it matches the true provenance. We next provide a taxonomy of cases where a verifiable provenance can be incorrect (Figure~\ref{fig:error_taxonomy}). Since the true provenance is unknown, this taxonomy provides a way to understand when to apply our techniques.}

\rmix{\textit{Unstable provenance.} We say that our verifiable provenance is {\em unstable} when it is not meaningful for a given task, so the returned provenance can be incorrect even if the LLM answer is correct. We now discuss set of questions $Q$ for which verifiable provenance is stable.}    As in classical provenance literature~\cite{green2007provenance,buneman2001and,cheney2009provenance,glavic2021data},  $Q$
can represent the computation corresponding to 
any {\em positive relational algebra} (${\mathcal RA}^+$); this class 
encapsulates SPJRU (Select-Project-Join-Rename-Union) queries\trs{, 
along with recursion}. This broad class of queries
encompasses data extraction (e.g., ``extract all
the dates in this text'') or transformation 
(e.g., ``reformat the intake form as a table''). 
$Q$ could also include multi-hop reasoning (e.g., ``what was the origin of
the money that was given to the officer?), analogous to 
a recursive query that follows exchanges of money to its source.
Prior work on {\em why} and {\em where}~\cite{buneman2001and} and {\em how}~\cite{green2007provenance} provenance
also focus on ${\mathcal RA}^+$. 
In fact, Glavic~\cite{glavic2021data} describes
that why provenance is synonymous with ${\mathcal RA}^+$. 
$Q$ could also include as part of its output computation
corresponding to monotonic aggregates, i.e., where
the aggregates trend in a single direction, such as 
\texttt{COUNT(*)} 
or \texttt{SUM} of non-zero positive numbers; e.g., ``count the number of officers involved''
or ``summarize all medical aspects in the document''. 

As in traditional provenance, as soon as $Q$ includes
arbitrary negation or aggregation, or operations on aggregates, 
the usefulness 
of a minimal provenance decreases.
As an extreme example, if $Q$ represents the computation of 
the average of numbers $ \in  \{1, 2, 3, 4, 5\}$,
one minimal provenance could be $\{3\}$, which is not
very useful. 
Likewise, if $Q$ includes a predicate on an aggregate, e.g., ``find
all officers who were mentioned at most three times'',
a minimal provenance might list 
just one mention for each output officer, 
even if they were listed twice or three times, which may not be useful. 

\rmix{\textit{Spurious provenance.} Finally, we say that a verifiable provenance can be {\em spurious} even for ${\mathcal RA}^+$, when LLMs make incorrect judgments. For example, it may incorrectly judge that a true provenance fails to reproduce the answer on the entire input. }

\if 
There are multiple  reasons why a verifiable provenance can 
While all our solutions for Problem~\ref{def:1-provenance}
are guaranteed to return a minimal provenance,
independent of $Q$,
we now discuss the class of questions $Q$ for which
this provenance is {\em useful},
and connect this to classical provenance literature.

As in classical provenance literature~\cite{green2007provenance,buneman2001and,cheney2009provenance} (see \cite{glavic2021data} for a recent survey), \rone{$Q$
can represent the computation corresponding to 
any {\em positive datalog query}; this class 
encapsulates SPJRU (Select-Project-Join-Rename-Union) queries, 
along with recursion. This broad class of queries
encompasses data extraction (e.g., ``extract all
the dates in this text'') or transformation 
(e.g., ``reformat the intake form as a table''). 
$Q$ could also include multi-hop reasoning (e.g., ``what was the origin of
the money that was given to the officer?), analogous to 
a recursive query that follows exchanges of money to its source.} 
Prior work on {\em why} and {\em where}~\cite{buneman2001and} and {\em how}~\cite{green2007provenance} provenance
also focus on this subclass of queries. 
In fact, Glavic~\cite{glavic2021data} describes
that why provenance is synonymous with 
positive relational algebra ${\mathcal RA}^+$,
a class equivalent to SPJRU queries.
 
The question $Q$ could also include as part of its output computation
corresponding to monotonic aggregates, i.e., where
the aggregates trend in a single direction, such as 
\texttt{COUNT(*)} 
or \texttt{SUM} of non-zero positive numbers; e.g., ``count the number of officers involved''
or ``summarize all medical aspects in the document''.

As in traditional provenance, as soon as $Q$ includes
arbitrary negation or aggregation, or operations on aggregates, 
the usefulness
of a minimal provenance decreases. 
As an extreme example, if $Q$ represents the computation of 
the average of numbers $ \in  \{1, 2, 3, 4, 5\}$,
one minimal provenance could be $\{3\}$, which is not
very useful. 
Likewise, if $Q$ includes a predicate on an aggregate, e.g., ``find
all officers who were mentioned at most three times'',
a minimal provenance might list 
just one mention for each output officer, 
even if they were listed twice or three times, which may not be useful.

Finally, it is worth mentioning that prior work on {\em why} provenance~\cite{buneman2001and,cui2000tracing}
introduces the notion of {\em witnesses}, much like our notion of
provenance, with {\em sufficiency} mirroring our notion of  
{\em verifiability}, while minimality being defined identically.
For ${\mathcal RA}^+$, a property analogous to strong monotonicity holds
for traditional data provenance, where any superset of a provenance is also one;
in our case, we also include the LLM $L$ in our 
definition of the property, since
the behavior can vary by LLM (See Table~\ref{tab:tasks}),
and further relax it to give weak monotonicity, a novel property.
\fi

\vspace{-4mm}
\subsection{Handling Non-Determinism in LLM}~\label{subsec:llm-uncertainty}
\revised{So far, we have defined our problems, the applicable tasks based on monotonicity, and the class of questions for which provenance is useful. We now examine LLM behavior and its impact on provenance.} 
When answering questions, instead of producing a single output,  
LLMs produce  
outputs with a probability distribution, which might introduce uncertainty when the same input is processed across multiple runs. While this uncertainty is already handled by our use of LLM-as-a-judge to accomodate for non-identical, but equivalent answers, 
here we go further to see how LLM non-determinism can impact provenance. 

Overall, the degree of uncertainty in LLM output depends on the decoding strategy. 
\rtwo{When greedy decoding is used (by setting the temperature to 0), the model picks the most probable token at each step~\cite{decoding}, and is thus theoretically deterministic. }
To validate this, we run each of 500 question–text pairs from our three workloads 50 times. 
\revised{We then cluster the 50 outputs per pair using exact string match (resp., LLM-as-a-judge), and report the maximum cluster size divided by 50 as {\em str\_s} (resp., {\em llm\_s}) in Table~\ref{tab:llm-determininstic}. }
Non-deterministic outputs measured by string match happen about 1\% to 7\% of the time across workloads. 
\rtwo{Such a gap occurs when probabilities are tied}, due to factors like \rtwo{non-deterministic GPU kernels, quantization rounding, or floating-point errors~\cite{floating}}. These cases are 
caught by LLM-as-a-judge, resulting in {\bf deterministic outputs $>$99\%} of the time---so the impact of non-determinism is insignificant. 

\techreport{To ensure determinism on all possible inputs, we  further extend the indicator function $I()$ (in Definition~\ref{def:provenance}) that checks if two LLM outputs are equivalent to handle the special case of output probabilities being tied. Here, instead of considering a single LLM output, let $\mathcal{A}$ and $\mathcal{A'}$ be the set of all LLM outputs on the full text $T$ and a subset of the text $P$. 
\revised{Then, to verify whether $P$ is a verifiable provenance, we check if the outputs with the maximum probability in $\mathcal{A}$ and $\mathcal{A'}$ are equivalent. We introduce a small constant ($\epsilon=0.01$ in our implementation) to account for floating-point errors when identifying outputs with the maximum probability. For example, consider $\mathcal{A} = \{a_1, a_2, a_3\}$ (where $a_i$ is a possible output) with corresponding probabilities  $0.43$, $0.425$, and $0.145$. In this case, both $a_1$ and $a_2$ are outputs with maximum probability under $\epsilon$-relaxation. $\mathcal{A}$ is equivalent to $\mathcal{A'} = \{a'_1, a'_2\}$ (where $a'_1$ has a probability of 0.95) if either $I(a_1, a'_1)$ or $I(a_2, a'_1)$ returns True. } In the following, we adopt this implementation of $I()$, which, however, has a negligible impact on our results, occurring in less than 1\% of cases, compared to the original $I()$ implementation. }


\vspace{-3.5mm}
\subsection{Overview of \sys} 
\vspace{-1mm}
\label{subsec:overview}

\rone{Our notion of provenance considers a ``black-box'' LLM that performs an unknown computation on a text data source whose output we observe. The goal of our provenance-finding algorithm is to identify the minimal text that can lead to an equivalent output. Because the underlying computation is unknown, our approach, unlike prior work on provenance in relational  databases~\cite{green2007provenance,buneman2001and,cheney2009provenance}, cannot track provenance throughout query evaluation. Instead, we develop a method that \textit{reverse-engineers} the LLM's answering process after the fact.   
}

We next present an overview of \sys, our
provenance discovery framework, as shown in  Figure~\ref{fig:overview}. 
Given text $T$, question $Q$, LLM $L$, and
answer $A = L(Q,T)$, 
\sys uses a two-phase approach to infer a single minimal provenance (Section~\ref{sec:1-provenance}),
and can be extended to find $k$ distinct minimal provenances  (Section~\ref{sec:k-provenance}).

\topic{Finding a single minimal provenance}
\sys uses a two-phase approach, 
{\em prune-refine}, 
to return a minimal provenance. 
The {\em pruning} phase returns a small-sized provenance by aggressively pruning out sentences 
that likely do not contribute to $A$, 
while the {\em refinement} phase 
takes the provenance 
returned from the pruning phase 
and refines it to obtain a minimal provenance. 

For the pruning phase, 
we propose four strategies 
to reduce provenance size 
by leveraging {\em text relevance to $Q$ and $A$} and {\em different scan orders} (top-down or bottom-up).  
\vspace{-1mm}
\begin{itemize}[leftmargin=*]
    \item (Text Relevance) We divide the text $T$ into equal-sized portions and prioritize examining those with high relevance to the pair $\langle Q,A \rangle$, estimated using embedding similarity~\cite{caspari2024beyond} or an LLM. 
    \item (Scan Orders) To return a small-sized provenance quickly, one approach is to process text portions in decreasing order of relevance to $\langle Q,A \rangle$, incrementally add a portion at a time, and return the top‑$i$ portions as soon as they yield an answer equivalent to $A$.    
    Alternatively, starting from $T$, we repeatedly discard the least relevant half until the remaining portions fail to reproduce $A$. 
\end{itemize}
\vspace{-1mm}
As bottom-up and top-down approaches excel in different scenarios, we provide a theoretical cost analysis to characterize where each shines, based on the size of the returned provenance. We further propose an adaptive method that combines them to robustly infer a small-sized provenance while avoiding their worst-case limitations. 

In the refinement phase, we propose two strategies that take any provenance (e.g., the output from the prune phase) 
and  return a corresponding 
minimal provenance by removing 
all non-relevant sentences 
with respect to $\langle Q, A\rangle$. 
The first strategy, {\em sequential-greedy}, 
iteratively deletes sentences 
from a given provenance 
if the remaining text after each removal 
still produces an answer equivalent to $A$. 
In contrast to {\em sequential-greedy} 
that considers one sentence at a time, the second strategy, 
{\em exponential-greedy}, removes sentences in exponentially increasing batch sizes. Starting with one sentence, if deletion succeeds, it attempts to remove two, then four, and so on, doubling the batch size each time to accelerate pruning. Both strategies enumerate sentences in decreasing order of their indexes (from larger to smaller) to maximize token reuse via the KV cache in each LLM invocation, leading to a significant cost reduction. 


\topic{Find top-$k$ minimal provenances}
The provenance returned by the two-phase approach
is minimal. However, 
in certain scenarios, as outlined in Section~\ref{sec:k-provenance}, additional evidence is needed. 
Thus, as part of \sys, we introduce a 
greedy tree-search approach 
to return the $k$ distinct 
minimal provenances, 
allowing users to explore all provenances and gain greater confidence in the answers. 


\vspace{-1mm}
\section{Finding Top-1 Provenance}
\label{sec:1-provenance}

Given the input text $T$, question $Q$, LLM $L$, and answer $A$, we present a two-phase approach to infer a single minimal provenance. 



\subsection{Phase 1: Provenance Pruning} 
\label{subsec:phase1}
To prune sentences that do not contribute to the answer $A$, we propose four intuitive strategies that leverage heuristics to rank sentences by their relevance to the $\langle Q,A \rangle$ pair (e.g., embedding similarities or LLMs), and  different scan orders (e.g., bottom-up or top-down), along with an adaptive strategy that combines the strengths of these approaches in Section~\ref{subsec:analysis}.

\vspace{1pt}
\topic{Strategy 1: Embedding-bottom-up} In this strategy  (Algorithm~\ref{alg:linear-search}), we divide the input text $T$ into a list of equal-sized blocks, $T = \langle b_1, b_2, \dots, b_m\rangle$, where each block contains a sequence of sentences (Line 1). The value $m$ is chosen to balance the cost of strategies and the size of the returned provenance. (Empirically, setting $m = 20$ strikes a good balance, as shown in Section~\ref{sec:exp})  
  We then sort the blocks in descending order of embedding similarity to $\langle Q,A \rangle$, i.e., 
the concatenation of $Q$ and $A$ (by setting \code{ranker} as \code{embedding} in Line 2).  
As soon as the first $i$ sorted blocks reproduce an answer equivalent to $A$, we stop searching and return them as the provenance (Lines 5–9). When evaluating the top-$i$ blocks using LLM $L$, we always re-order the sentences in the evaluated blocks in increasing order of their indexes in $T$, as implemented by \code{rerank\_index} in Line 7. For example, the top-$i$ blocks $\langle b_7,b_8, b_2, b_3 \rangle$ are reordered as $\langle b_2, b_3, b_7,b_8 \rangle$ before being passed to $L$ for answer generation, to prevent unintended changes in meaning due to reordering. 

\setlength{\textfloatsep}{0pt}
\begin{algorithm}[bt]
    \scriptsize
    \caption{\code{Prune}[\code{ranker}, \code{scan}]}
    \label{alg:linear-search}
 \KwIn{$\langle T,Q,L\rangle$}
$T \leftarrow \langle b_1,b_2,...,b_m\rangle$; $A \leftarrow L(T,Q)$ \\ 
$B \leftarrow \code{rank}(T, \code{ranker})$ \\ 
$P \leftarrow \emptyset$ \\ 
\If{\code{scan} == bottom-up}
{\For{$b_i\in B$, in ascending order of $i$}{
$P \leftarrow P \cup b_i$ \\ 
$P' \leftarrow \code{rerank\_index}(P)$ \\ 
\If{$I(L(P',Q),A)==$True}{
{\bf break}
}
}}
\If{\code{scan} == top-down}{
$l \leftarrow 1, r \leftarrow m$ \\ 
    \While{$l \leq r$}{
$mid \leftarrow \lceil \frac{l+r}{2} \rceil$ \\ 
$P \leftarrow <b_1,b_{2},...,b_{mid}>$ \\ 
$P' \leftarrow \code{rerank\_index}(P)$ \\ 
\eIf{$I(L(P',Q),A) == \text{True}$}{
$r \leftarrow mid-1$ \\ 
}
{$l \leftarrow mid + 1$}
}
}
$P \leftarrow \code{rerank\_index}(P)$ \\ 
\textbf{Return} $P$\\
\end{algorithm}

\topic{Strategy 2: Embedding-top-down} This strategy instead  performs a top-down search over the blocks sorted by embedding similarity, by setting \code{scan} to \code{top-down} in Line 10 of Algorithm~\ref{alg:linear-search}. 
We start with the entire ranked list of blocks $B$ from $T$ as the initial provenance $P$, and perform binary search to identify the smallest top-$k$ prefix that reproduces an answer equivalent to $A$. At each iteration, we consider the midpoint $mid = \lceil \frac{l + r}{2} \rceil$ of the current index range $[l, r]$, and evaluate the top-$mid$ blocks, i.e., $⟨b_1, b_2, \dots, b_{mid}⟩$ to see if they constitute a provenance (Line 11-15). If they reproduce an answer equivalent to $A$, we refine the search window to the left half $[l, mid - 1]$; otherwise, we consider the right half $[mid + 1, r]$ in the next iteration (Line 16-19).  
 The algorithm makes at most $\lceil \log_2 m \rceil$ LLM calls, where $m$ is the number of blocks in $T$.  

\vspace{1pt}
\topic{Strategy 3 \& 4: LLM-bottom-up and LLM-top-down} These strategies are like 1 \& 2 but instead of sorting blocks based on embedding similarities, they invoke an LLM model $L'$ to score each block $b_i$ in $B$ on a  scale of 1 to 10 based on the relevance of $b_i$ to the $\langle Q,A\rangle$ pair, using the \code{LLM-Ranker-Prompt} \papertext{in \vldb{\cite{blip}.}}\techreport{below.} The model $L'$ used to score the blocks does not necessarily need to be the same as the model $L$ used to answer $Q$. 

\techreport{\noindent\underline{\code{LLM-Ranker-Prompt}}: the \textcolor{blue}{[Q]}, \textcolor{blue}{[A]}, and \textcolor{blue}{[Text blocks]} are placeholders corresponding to $Q$, $A$, and $B$ in Algorithm~\ref{alg:linear-search}. }



Intuitively, given text blocks $B$ sorted by embedding similarities or LLMs, when the provenance lies within the first‑$i$ blocks and $i$ is small (i.e., the provenance size is small and the ranker is effective), a bottom-up search is ideal as it can quickly identify a small provenance. On the other hand, if $i$ is large, bottom-up incurs high cost by repeatedly evaluating the first $i$ blocks, resulting in a quadratic number of block accesses. In such a case, the top-down search is more efficient, as it discards blocks exponentially faster. In Section~\ref{subsec:analysis}, we analyze the costs of both strategies \rone{and in Section~\ref{sec:rec_strategy} recommend what approach to use for a given dataset. }   

\techreport{
\mypython{
\small 
    LLM-Ranker-Prompt: Given the following question: \textcolor{blue}{[Q]},  and a list of text blocks, the corresponding answers are \textcolor{blue}{[A]}. Your task is to assign a score (from 1 to 10) to each block based on how likely it is to contain context relevant to answering the question. The text blocks are listed below, each starting with Block i: followed by its content. Return only a comma-separated list of scores corresponding to each block, in the order they are given. Do not include any explanations or additional text. \textcolor{blue}{[Text blocks]}}}

\vspace{-6mm}
\subsection{Phase 2: Provenance Refinement}  
The returned provenance from Phase 1, while smaller than $T$,  may not be minimal. In this phase, we present two strategies that take as input any provenance and compute a minimal provenance.

\begin{algorithm}[bt]
\scriptsize
\caption{\code{Refine}[\code{scan}]}
\label{alg:greedy}
\KwIn{$P,\langle T,Q,L\rangle$}  
\SetKwFunction{SP}{Sequential\_Greedy}
\SetKwFunction{EP}{Exponential\_Greedy}
\SetKwProg{Fn}{Function}{:}{}
\Fn{\SP{$P$}}{
\For{$s_i \in P$, in descending order of $i$}{
$P' \leftarrow P \setminus s_i$ \\ 
\If{$I(L(P',Q),A) == True$}{
$P \leftarrow P'$ \\ 
}
}
\textbf{Return }$P$}
\Fn{\EP{$P$}}{
$j \leftarrow |P|-1, l \leftarrow 0$\\ 
\While{$j \geq 0$}{
$i \leftarrow j - 2^l+1$ \\ 
$P' \leftarrow P \setminus P_{[i,j]}$ \\ 
\If{$I(L(P',Q),A) == True$}{
$P \leftarrow P'$ \\ 
$j \leftarrow i - 1$,$l \leftarrow l + 1$ \\ 
}
\Else{
$l \leftarrow 0$
}
}
}
\If{\code{scan} == sequential}{
$P \leftarrow$ \SP{$P$}; $P_{prev} \leftarrow P$; $A \leftarrow L(T,Q)$ \\ 
\While{$P_{prev} \neq P$}{
    $P_{prev} \leftarrow P$ \\ 
    $P \leftarrow \SP{P}$
}
}
\If{\code{scan} == exponential}{
$P \leftarrow$ \EP{$P$}; $P_{prev} \leftarrow P$; $A \leftarrow L(T,Q)$ \\ 
\While{$P_{prev} \neq P$}{
    $P_{prev} \leftarrow P$ \\ 
    $P \leftarrow \EP{P}$
}
}
\textbf{Return} $P$
\end{algorithm}

\topic{Strategy 1: Sequential-Greedy} Taking any provenance $P$ as input, \code{Sequential\_Greedy} (in Algorithm~\ref{alg:greedy}) considers each sentence $s_i \in P$ in descending order of $i$ to maximize the prefixes shared among LLM inferences. If after removing a sentence $s_i$ from the current provenance $P$, the remaining text $P \setminus s_i$ produces an answer equivalent to $A$, then $s_i$ is removed from $P$; otherwise, $s_i$ is retained (Lines 1–6). We repeatedly invoke \code{Sequential\_Greedy} on the resulting provenance $P$ until $P$ no longer changes (Lines 18–21). This step (and a similar one in Strategy 2) is essential for minimality. 

\begin{example}
    Assume the provenance $P$ returned from the first phase is $P = \langle s_1, s_2, s_3, \dots, s_{20} \rangle$. In the first iteration, $s_{20}$ will be removed if $P \setminus \{s_{20}\}$ returns an equivalent answer to $A$. Next, when  $s_{19}$ is examined, its prompt (consisting of the question and context) shares a prefix with the previous one since both contain the same question followed by the same prefix of context, i.e., $\{s_1, s_2, \dots, s_{18}\}$. As a result, the cost of input tokens in the second iteration is reduced by a factor of $f_L$ = 4 (or 2) when using \code{gemini-2-flash} (or \code{gpt-4o-mini}), respectively, because the entire prompt consists of cached tokens. Similar cost reductions apply to subsequent LLM inferences, resulting in significant cost savings. 
\end{example}

\topic{Strategy 2: Exponential-Greedy} In contrast to considering sentences one by one, \code{Exponential\_Greedy} (Algorithm~\ref{alg:greedy}) removes sentences more aggressively by exponentially increasing the number of sentences considered for removal, as long as the result  continues to produce an answer equivalent to $A$. Let $P_{[i,j]}$ denote a subsequence of $P$ from the $i$-th to the $j$-th sentence within $P$. 

We use a running example to explain the behavior. Suppose the provenance $P$ returned from the first phase is $P = \langle s_1, s_2, s_3, \dots, s_{20} \rangle$. This strategy considers each subsequence $P_{[i,j]}$ within $P$. In the first iteration, it considers removing a subsequence $P_{[20,20]} = \langle s_{20} \rangle$.  If $P \setminus P_{[20,20]}$ produces an equivalent answer to $A$, then $P_{[20,20]}$ is removed from $P$. In the next iteration, instead of only $s_{19}$, we double the length of the candidate subsequence and consider removing $P_{[18,19]} = \langle s_{18}, s_{19}\rangle$. If the remaining text $\langle s_1, s_2, \dots, s_{17} \rangle$ still produces an equivalent answer, $P_{[18,19]}$ is removed, and we proceed to $P_{[14,17]}$. If, in any iteration, the remaining text after removing the current subsequence fails to reproduce the answer, the length of the next candidate subsequence is reset to 1 (Lines 8–16). For example, if removing $P_{[18,19]}$ causes the remaining text to produce a different answer, we instead attempt to remove $\langle s_{19} \rangle$. If this succeeds, we then examine the removal of $P_{[17,18]}$ in the next iteration.  The above process repeats until no further sentences can be removed (Lines 8–16). As with \code{Sequential\_Greedy}, we repeatedly apply \code{Exponential\_Greedy} to the provenance returned the previous iteration until the provenance no longer changes (Lines 23–26). 

During this process, the opportunity to leverage prefix caching for cost optimization is maximized, as the prompt evaluated in the current iteration shares the longest possible prefix with the previous one, except for the sentences that have been examined and retained so far.

\subsection{\rtwo{Quality} and Cost Analysis}
\label{subsec:analysis}

\subsubsection{\rtwo{Quality} Analysis}
\label{subsubsec:correct} 
Let $\mathcal{S}_{prune} = \{\text{\code{Embedding-bottom-up}}, \\ \text{\code{Embedding-top-down}},  \text{\code{LLM-bottom-up}}, \text{\code{LLM-top-down}}\}$ and  $\mathcal{S}_{refine} = \\ \{\text{\code{Sequential-Greedy}},  \text{\code{Exponential-Greedy}}\}$ be the set of strategies in the pruning and refinement phases, respectively.  





\begin{theorem}
\label{theo:optimal}
\edit{
$\forall S \in \mathcal{S}_{refine}$, for all tasks $\langle T,Q,L\rangle$, and provenance $P \sqsubseteq T$, $S(P,\langle T,Q,L\rangle)$ always returns a minimal provenance. } 
\end{theorem}

We show the formal proof of Theorem~\ref{theo:optimal} \papertext{in \vldb{\cite{blip}.}}\techreport{below.} \rmix{To return a minimal provenance, \sys does not make any assumptions about the task (e.g., weak or strong monotonicity). These assumptions are only necessary if a strictly-minimal provenance is needed (Theorem~\ref{theo:minimal-eq}).}   

\techreport{\begin{proof}
    We prove Theorem~\ref{theo:optimal} by contradiction. Given a provenance $P$ for the answer $A$, let the provenance returned by $S$ be $P_S$, where $S \in \mathcal{S}_{\text{refine}}$. If $P_S$ is not minimal, there must exist a sentence $s \in P_S$ such that $P_S \setminus s$ produces an answer equivalent to $A$, i.e., $I(L(P_S \setminus s), A)$ is \textit{True}. If $S = \code{Sequential\_Greedy}$, this contradicts the algorithm’s behavior (Lines 3 and 17–21 in Algorithm~\ref{alg:greedy}), since \code{Sequential\_Greedy} is applied in multiple passes to the provenance returned from the previous iteration, ensuring that removing any sentence from it results in an answer that is not equivalent to $A$. A similar contradiction can be established for $S = \code{Exponential\_Greedy}$. Thus, $P_S$ is a minimal provenance returned by any strategy in $\mathcal{S}_{\text{refine}}$. 
\end{proof}}



\begin{theorem}
\label{theo:all}
\edit{$\forall S_i \in \mathcal{S}_{prune}$ and $S_j \in \mathcal{S}_{refine}$,  for any task $\langle T,Q,L\rangle$, let $P = S_i(\langle T,Q,L\rangle)$ be the provenance returned by $S_i$.  $S_j(P,\langle T,Q,L\rangle)$ always returns a minimal provenance for $\langle T,Q,L\rangle$. }  
\end{theorem}
\vspace{-0.5em}
 
The correctness of Theorem~\ref{theo:all} follows directly from the fact that any strategy $S_i$ in $\mathcal{S}_{prune}$ returns a provenance. \rone{Although heuristics are explored in rankers (e.g., embedding-based or LLM-based) and enumerations (e.g., bottom-up or top-down) to reduce cost and latency, Theorem~\ref{theo:all} guarantees that \sys returns a minimal provenance for any data processing task, including non-monotonic tasks.} \rtwo{We provide below an example of how \sys behaves on a non-monotonic task to illustrate this guarantee.}
\rtwo{Consider a non-monotonic task where $T = \langle s_1,s_2,s_3,s_4 \rangle$, and only $\langle s_1 \rangle$ and $\langle s_1,s_2,s_3 \rangle$ are verifiable provenances among all true subsets of $T$ ($T$ itself is a provenance by definition). In this case, $\langle s_1 \rangle$ is a strictly-minimal provenance, while both  $\langle s_1 \rangle$ and  $\langle s_1,s_2,s_3 \rangle$ are minimal provenances. \sys returns $\langle s_1,s_2,s_3 \rangle$, and is not able to return $\langle s_1 \rangle$ since  any superset of it formed by adding one more sentence (e.g., $\langle s_1,s_2 \rangle$) is not a provenance. } 
\rtwo{For the above task, if $\langle s_1,s_2 \rangle$ is also a provenance, then it becomes a weakly-monotonic task. \sys then returns $\langle s_1 \rangle$ by searching along the path from $\langle s_1,s_2,s_3 \rangle$, to $\langle s_1,s_2 \rangle$, and finally $\langle s_1 \rangle$, using Algorithm~\ref{alg:greedy} (Line 17-26). }

\subsubsection{Cost of Bottom-up versus Top-down} 
\label{subsubsec:bottom-top}
\leavevmode 

Consider the \code{bottom\_up} and \code{top\_down} strategies, with the text $T$ as a list of $m$ equal-sized blocks ranked by  embeddings or LLMs. We focus on embeddings, but the analysis is similar for any ranker.  
Let $C_{BU}$ and $C_{TD}$ denote the total cost of \code{embedding\_bottom\_up} and \code{embedding\_top\_down}, respectively. Assume that the minimal provenance is contained within the top-$k$ out of $m$ total blocks. When there are multiple  minimal provenances, we consider $k$ to be the smallest number among them. 

To determine which strategy incurs a lower cost given $k$, we analyze how the cost ratio $\frac{C_{BU}}{C_{TD}}$ varies with respect to $k$.
To this end, we approximate the cost by counting the total number of block accesses in LLM verification calls, \trs{i.e., invocations of $I(L(P, Q), A)$,} where a block is counted multiple times if it appears in multiple LLM calls. 
\techreport{We also assume that the LLM is accurate in provenance verification. Specifically, if a subsequence $T'$ is a superset of any minimal provenance $mp$ (i.e., $mp \subseteq T'$), then the verification function returns true: $I(L(T', Q), A) = \text{True}$. We will later relax this assumption and analyze its impact on our results. Note that under this assumption, both \code{bottom\_up} and \code{top\_down} return exactly the top-$k$ blocks as the provenance in this phase. }

\techreport{The cost of \code{bottom\_up} is computed as:  $C_{BU} = \sum_{i=1}^{k}i = \frac{k(k+1)}{2}$.  } 

\techreport{We use a recursive expression to compute the cost of \code{top\_down}. 
Let $C(l,r)$ be the total number of blocks used in the iteration where \code{top\_down} is searching the blocks in the interval $[l,r]$ in Algorithm~\ref{alg:linear-search}. In this case, $C_{TD} = C(1,m)$.  }
\techreport{
\[
C(l,r)\;=\;
\begin{cases}
0, &
l = r, \\[-1pt]
\displaystyle
\text{mid}(l,r)\;+\;
C\!\bigl(l,\;\text{mid}(l,r)-1\bigr), &
k \le \text{mid}(l,r), \\[-1pt]
\displaystyle
\text{mid}(l,r)\;+\;
C\!\bigl(\text{mid}(l,r)+1,\;r\bigr), &
k > \text{mid}(l,r).
\end{cases}
\]
}
\techreport{where $
\text{mid}(l,r)\;=\Bigl\lceil (l+r)/2\Bigr\rceil
$. In this expression, the number of blocks considered in each iteration, i.e., $\text{mid}(l, r)$, is accumulated until the binary search terminates at $l = r$. }

We \papertext{leave the details of how to compute $C_{BU}$ and $C_{TD}$ in \cite{blip}, and} present a breakdown of the cost comparison based on different values of $k$ relative to $m$, the total number of blocks.  

\begin{theorem}
    \label{theo:top-down-bottom-up}
   When $k \leq L_m$, where $L_m = \frac{\sqrt{8m-7}-1}{2}$,  we have $C_{BU} \leq C_{TD}$; when $k \geq U_m$, where $U_m = \frac{-1 + \sqrt{1+8 (mlog_2m-m+1)}}{2}$, we have $C_{BU} \geq C_{TD}$. For $L_m<k<U_m$, one must evaluate $C_{BU}$ and $C_{TD}$ explicitly to determine which is larger. 
\end{theorem}

We show the detailed proof \papertext{in \vldb{\cite{blip}.}}\techreport{below.} As implied in Theorem~\ref{theo:top-down-bottom-up}, bottom\_up is preferred when the provenance size is relatively small (i.e., $k$ is smaller than $L_m$). In contrast,  top-down wins when $k$ is large (greater than $U_m$). 

\techreport{\begin{proof}
    [Proof of Theorem~\ref{theo:top-down-bottom-up}] 
During binary search in the top\_down approach, consider the lower and upper bounds of $C_{TD}$. In the best case, in every iteration it moves to the left half interval, leading to a cost $C_{TD}^{L} = \lceil \frac{m}{2} \rceil + \lceil \frac{m}{4} \rceil + \dots + 1 \geq m - 1$. In the worst case, in every iteration it moves to the right half interval. In this case, the first iteration incurs a cost of $\lceil \frac{m}{2} \rceil$, and the second iteration incurs a cost of $\lceil \frac{m}{2} \rceil + \lceil \frac{m}{4} \rceil$. Thus, the $i$-th iteration has a cost of $C_i = \sum_{j=1}^{i} \lceil \frac{m}{2^j} \rceil$. The total cost is $C_{TD}^{U} = \sum_{i=1}^{h} C_i$, where $h = \lceil \log_2 m \rceil$. We have, 

\begin{equation}
    C_{TD}^{U} = mlog_2m-m+1
\end{equation}
Solving $C_{BU} \leq C_{TD}^{L}$ leads to, 

\begin{equation}
    k \leq \frac{\sqrt{8m-7}-1}{2} 
\end{equation}
 
Soving $C_{BU} \geq C_{TD}^{U}$ gives, 

\begin{equation}
    k \geq \frac{-1 + \sqrt{1+8 (mlog_2m-m+1)}}{2}
\end{equation}

We denote by $U_m = \frac{-1 + \sqrt{1+8 (mlog_2m-m+1)}}{2}$ and $L_m = \frac{\sqrt{8m-7}-1}{2}$, respectively. When $k\leq L_m$, we have $C_{BU} \leq C_{TD}$. 
When $k \geq U_m$, we have $C_{BU} \geq C_{TD}$. For $L_m<k<U_m$, one must evaluate the 
two formulas explicitly to determine which is larger.
\end{proof}}

 \begin{figure}[tb]
    \centering
    \includegraphics[width=1\linewidth]{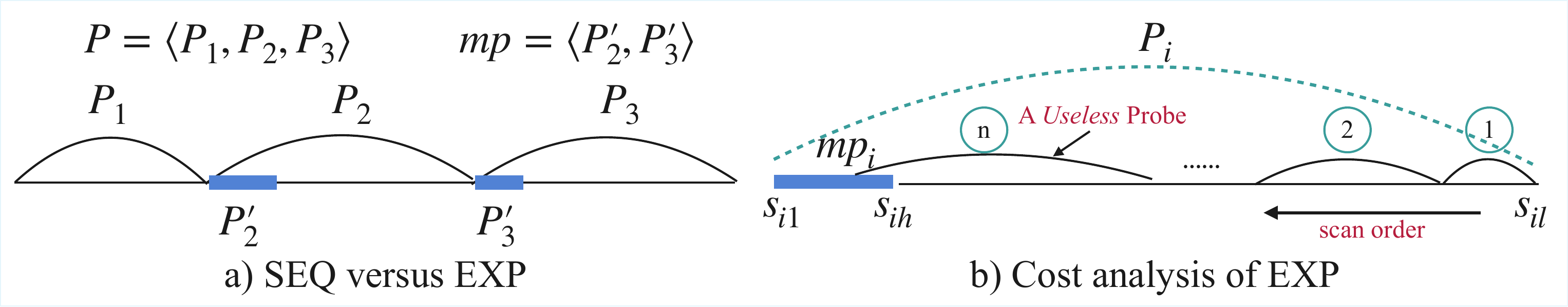}
    \vspace{-1.8em}
    \caption{\small \code{SEQ} versus \code{EXP}. $P$ refers to the provenance returned from the first phase, and $mp \sqsubseteq P$ refers to a minimal provenance.}  
    \label{fig:seq-exp}
\end{figure}

\subsubsection{Cost of Sequential\_Greedy versus Exponential\_Greedy}
\label{subsubsec:seq-exp}
\leavevmode
Taking the provenance, denoted as $P$, returned by the {\em prune} phase as input, 
both Sequential\_Greedy (\code{SEQ} for short) and Exponential\_Greedy (\code{EXP} for short) return a minimal provenance, denoted as $mp \sqsubseteq P$. Let $C_{SEQ}$ and $C_{EXP}$ be the cost of \code{Sequential\_Greedy} and  \code{Exponential\_Greedy}, when applied to $P$. To analyze $C_{SEQ}$ and $C_{EXP}$, we assume  both strategies return the same \trs{minimal provenance} $mp$, as illustrated in Figure~\ref{fig:seq-exp}-a, where the blue intervals represent $mp$. 

Irrespective of the distribution of the minimal provenance $mp$ within $P$, any input sequence $P$ can be decomposed into a list of subsequences (e.g., $\langle P_{1}, P_{2}, P_{3}\rangle$), where the first (possibly empty) subsequence $P_1$ has no contribution to $mp$, followed by $P_2$, $P_3$, ..., each of which has a prefix $P_2'$, $P_3'$ that is part of $mp$.  
This decomposition is useful for comparing \code{EXP} versus \code{SEQ} across different types of sequences. 
For $P_1$, \code{EXP} is clearly preferred over \code{SEQ}, as \code{EXP} eliminates sentences in logarithmic time, whereas \code{SEQ} does so linearly. 

We now focus on the cost comparison for sequences such as $P_2$ or $P_3$. W.L.O.G., 
consider such a sequence $P_i = \langle s_{i1}, s_{i2}, \dots, s_{il} \rangle$, where the leading subsequence $mp_i = \langle s_{i1}, \dots, s_{ih} \rangle$ is part of the minimal provenance $mp$, as illustrated in Figure~\ref{fig:seq-exp}-b. 
 \techreport{Similar to the cost analysis in Section~\ref{subsubsec:bottom-top}, we approximate the cost by counting the total number of block accesses in LLM verification calls and assume the LLM is accurate in provenance verification.} \code{EXP} performs logarithmic deletion over the subsequence $P_i \setminus mp_i$ and linear deletion over $mp_i$, whereas \code{SEQ} always performs linear deletion over the entire $P_i$\papertext{.}\techreport{, both by scanning sentences in decreasing order of their indexes.} However, when an LLM verification fails, \code{EXP} resets the probing window step back to 1, rendering the current LLM call ineffective in deleting any sentences, thus, a {\em useless} probe. To analyze the cost comparison between \code{EXP} and \code{SEQ}, we need to account for both the savings from \code{EXP}'s logarithmic deletion and the additional cost incurred by useless probes. 



\begin{theorem}
\label{theo:seq-exp}
    Given a sequence $P_i = \langle s_{i1}, s_{i2}, \ldots, s_{il} \rangle$ with a leading subsequence $mp_i = \langle s_{i1}, \ldots, s_{ih} \rangle$ that is part of the provenance, let $g = l-h$, we have: 
    \begin{align*}
    &C_{SEQ}=\Theta\!\bigl(gl+h^{2}\bigr),  \\ 
    \Omega\!\bigl(g+h^{2}\bigr)\le
&C_{EXP}\le
O\!\bigl(g\log^2_2 g+h\log^2_2g+h^{2}\bigr)
    \end{align*}
    The cost of $C_{SEQ}$ and $C_{EXP}$ are broken down as follows,  
    \begin{align*}
& g=\omega(1)
       \Longrightarrow\;
         C_{EXP}=o\bigl(C_{SEQ}\bigr), \\
& g=\Theta(1)
       \Longrightarrow\;
         C_{EXP}=\Theta \bigl(C_{SEQ}\bigr).
\end{align*}
\end{theorem}
We show the proof \papertext{in \vldb{\cite{blip}.}}\techreport{below.} As shown in Theorem~\ref{theo:seq-exp}, \( C_{\code{EXP}} \) is asymptotically no worse and often cheaper than \( C_{\code{SEQ}} \) in most cases. 
 In practice, since the provenance size often remains constant relative to the sequence length, i.e., \( r = \omega(1) \), \code{EXP} is cheaper than \code{SEQ} when $l$ and $r$ become larger, implied by \( C_{EXP} = o(C_{SEQ}) \).  

\techreport{\begin{proof}
\newcommand{\SEQ}{\mathrm{SEQ}}
\newcommand{\EXP}{\mathrm{EXP}}

First, consider computing $C_{SEQ}$. Let \( g := l - h \) be the number of sentences that are not part of the provenance. To linearly probe each sentence in \( P \setminus mp \) using \code{SEQ}, the number of data blocks per LLM verification is \( l - 1, l - 2, \dots, l - g \), leading to a total cost of \( gl - \frac{g(g+1)}{2} \). For probing each sentence in \( mp \), since no sentence can be removed, \code{SEQ} incurs an additional cost of \( h(h - 1) \), resulting in the total cost:
\[
C_{SEQ} = gl - \frac{g(g+1)}{2} + h(h - 1).
\] 

We have the tight bound of $C_{SEQ}$ as, 
\[
  \boxed{\,C_{\SEQ}=\Theta\!\bigl(gl+h^{2}\bigr)\,}
\]
since \(gl\ge g\) and \(h^{2}\ge h(h-1)\), while
\(\tfrac12 g(g+1)\le g^{2}\le gl+g^{2}\le2gl\). 

Now consider $C_{EXP}$. To probe and delete sentences that are not part of the provenance, \code{EXP} performs \( k \) successful probes followed by one failure, then resets the window size to 1 and recursively repeats the same probing process. The cost in this stage is broken down into the cost of successful probes and the failure probe. The former is computed as $\sum_{i=0}^{k-1}(h+g - 2^i) = k(h+g) - 2^k + 1$, where $k = \lceil log_2g\rceil$, the number of successful probes.

\[
  \underbrace{k(h+g)-2^k + 1}_{\text{successful probes}}
  \;+\;
  \underbrace{\max\{0,h+g-2^{k+1}\}}_{\text{one failure}}
  \;=\;O\!\bigl((h+g)\log_2 g\bigr).
\]

After a failure probe, at most \(2^{k}\!-\!1<\tfrac12 g\) non-provenance sentences remain,
so the recursion depth is at most \(\lceil\log_{2}g\rceil\). Thus, we have the upper asymptotic bound of \( C_{\code{EXP}} \) as follows, where \( O(h^2) \) accounts for the cost of probing the \( h \) provenance sentences. 

\[
  \boxed{\,C_{\EXP}=O\!\bigl(glog^2_2g+h\log^{2}_2 g+h^{2}\bigr)\,}.
\]

Now we develop the lower asymptotic bound of $C_{EXP}$ as below, 
\[
  \boxed{\,C_{\EXP}=\Omega\!\bigl(g+h^{2}\bigr)\,}.
\]

Here, every non-provenance sentence in \( P \setminus mp \) contributes a cost of at least one block to the overall cost.

Comparing the two strategies, we have the tight bound
\(C_{SEQ}(l,h)=\Theta\!\bigl(gl+h^{2}\bigr)\)
with \(g=l-h\), while \code{EXP} satisfies
\(\Omega\!\bigl(g+h^{2}\bigr)\le
C_{EXP}(l,h)\le
O\!\bigl(g\log^2_2 g+h\log^2_2 g+h^{2}\bigr)\).
Because \(gl\) grows strictly faster than \(g\log g\) once
\(g=\omega(1)\), we obtain \(C_{EXP}=o(C_{SEQ})\) for sequence containing more than a constant number of non-provenance sentences; conversely,
when \(g=\Theta(1)\) both costs are the same order. Overall, we have, 

\[
\boxed{%
\begin{aligned}
& g=\omega(1)
      &&\Longrightarrow\;
         C_{EXP}=o\!\bigl(C_{SEQ}\bigr), \\
& g=\Theta(1)
      &&\Longrightarrow\;
         C_{SEQ}=\Theta\!\bigl(C_{EXP}\bigr).
\end{aligned}}
\]
\end{proof}}

\subsubsection{\rone{Recommended Approach}}\label{sec:rec_strategy}
\rone{We have presented two different pruning  (bottom-up and top-down) and refinement strategies (SEQ and EXP). Theorems~\ref{theo:top-down-bottom-up} and \ref{theo:seq-exp} (as well as our empirical study in Section~\ref{subsec:exp-result}) show that the best pruning and refinement approach depends on the size of the provenance, quality of ranking,  as well as how  provenance is distributed within $T$.
Given absence of knowledge about these factors, we provide the following recommendations to decide which pruning and refinement strategy to use. 
\\\indent
\textit{Recommended Pruning strategy}. We propose an adaptive strategy for pruning. Theorem~\ref{theo:top-down-bottom-up} showed that if provenance is the first $L_m$ blocks, top-down performs best while if it is in the last $U_m$ blocks, then the bottom-up  performs best. 
We use this insight to design the following adaptive strategy. We use $CP = \frac{L_m + U_m}{2}$ to approximate a {\em crossover point}. 
Given a ranking approach, the adaptive strategy initially applies  bottom-up  to probe the ranked blocks. If the top-$CP$ blocks fail to produce the correct answer, the strategy switches to top-down to complete the search. Our empirical results (Section~\ref{subsec:exp-result}) show that this simple strategy often outperforms both methods.  \\\indent
\textit{Recommended Refinement strategy}. SEQ performs best when a short provenance can be found at the beginning of the document (i.e., $h$ is small in Theorem~\ref{theo:seq-exp}) while EXP performs better when documents are longer. This is consistent with our empirical observation: we recommend using SEQ if the number of sentences in the provenance returned from the pruning stage is below a threshold $t$ ($t=10$ reported in Section~\ref{subsec:exp-result}),  but to use EXP otherwise. 
}

\vspace{-2mm}
\subsubsection{\rone{\sys Scope}}\label{sec:usefulness_scope}
\rone{We so far have shown that \sys returns a minimal provenance for any LLM-powered data processing~\cite{zhao2024chat2data, huang2024transform, naeem2024retclean, buss2023generating, zhu2024autotqa} task, as discussed in Section~\ref{subsubsec:correct}; and for any  SPJRU (i.e., ${\mathcal RA}^+$) task, the provenance is well-defined and useful, as discussed in Section~\ref{subsec:database-prov-relationship}.   However, whether \sys can do so efficiently depends on the task itself as well as the abilities of the ranking method used.    \\\indent First, if the provenance for the task is large, then the cost of \sys increases, as discussed in Theorem~\ref{theo:seq-exp}---guaranteeing minimality of a large set is difficult as any sentence can be redundant and our algorithm must check if the removal of any sentence can lead to a smaller provenance. Second, the quality of the ranking method affects the cost---the better the ranker, the more effectively \sys can prune out blocks that are not part of the provenance. The quality of the ranker depends on the capabilities of the embedding model (or the LLM used for ranking), as well as task and data characteristics. Embedding models are commonly used for semantic retrieval tasks both for text~\cite{zhuang2023toolqa} and tabular data~\cite{jin2022survey,nan2022fetaqa}, and as we see in our experiments, can provide accurate ranking for common LLM-powered text and table processing tasks done by LLMs. However, embedding models may not be well-suited for ranking in tasks that require logical operations (e.g., when we have a predicate \texttt{A = 4}, and all values of \texttt{A} are between 3 to 5). Such queries should be performed using SQL, e.g., by using an LLM to rewrite the query $Q$ using text-to-SQL as opposed to the LLM processing data directly.  In this case, existing provenance tracking solutions for databases~\cite{green2007provenance,buneman2001and,cheney2009provenance} are more effective. We leave how our approach can be combined with query answering with SQL to the future work.}

\techreport{Note that adopting \code{Exponential\_Greedy} improves the worst-case performance of the two-phase approach when the provenance returned from the first phase is large, particularly when the provenance distribution is local. This is because it offers a logarithmic-factor reduction in pruning sentences that do not contribute to the answer, thereby making the approach more robust. }

\section{Finding Top-k Provenance} 
\label{sec:k-provenance}

The two-phase approach described in Section~\ref{sec:1-provenance} returns a single minimal provenance. However, it may not be sufficient for users to verify the correctness of an answer. 
\techreport{For example, consider a question asking for the publication year of a scientific paper in Figure~\ref{fig:top-k}. Provenances 1 and 2 represent two distinct provenances that support the correctness of the answer. In contrast, Provenance 3 happens to return the same answer without supporting its correctness, as it refers a different paper citation in the {\em Reference} section. 

A provenance is considered {\em positive} (e.g., Provenances 1 and 2 in Figure~\ref{fig:top-k}) if it supports the correctness of the answer. Otherwise, it is {\em negative} (e.g., Provenance 3). The distinction between positive and negative provenances can ultimately only be determined by a human than an LLM, as both types of provenance pass the LLM's verification process (i.e., they both reproduce the same answer derived from the full text). } 
To address this issue and to enable users to explore multiple distinct provenances and gain trust in the correctness of the answer, we develop an approach, outlined in Algorithm~\ref{alg:top-k}, to return the $k$ minimal provenances. \techreport{During this process, if users identify any one of the provenances as positive, then the algorithm terminates.} 

\begin{algorithm}[tb]
\scriptsize
\caption{\code{top-k-provenance}}
\label{alg:top-k}
\KwIn{$Q,\,T,\,L,\,A,F,I,\rtwo{k}$}  
$T = \langle b_1,\dots ,b_m\rangle$ \\ 
$H \leftarrow \text{empty max heap}$, $P \leftarrow \text{empty list}$ \\ 
$H.append((T_{1,m}, 1)) $ \\  
\While{$H \neq \emptyset$}{
$(T_{l,r}, score_{l,r}) \leftarrow H.pop()$ \\ 
\If{$l < r$}{
$mid = \lceil \frac{l+r}{2} \rceil$ \\ 
    $score_{l,mid} \leftarrow F(L(T_{l,mid},Q),A)$ \\ 
    $score_{mid+1,r} \leftarrow F(L(T_{mid+1,r},Q),A)$ \\ 
     }
$\rtwo{c_1 \leftarrow (l == r) \land (score_{l,r} \geq 0.5)}$ \\ 
\rtwo{$c_2 \leftarrow (l < r) \land (score_{l,r} \geq 0.5)  \land (score_{l,mid} < 0.5) \land   (score_{mid+1,r} < 0.5)$} \\
\If{\rtwo{$(c_1 \lor c_2) == True$}}{
    \If{\rtwo{$c_1 == True$}}{
    $P_{cur} = \code{refine}(T_{l,r})$ \\ }
    \Else{$P_{cur} = \code{refine(prune(}(T_{l,r})$)) \\}
    \techreport{\If{\code{user\_verify}($P_{cur}$) == Positive}{
        {\bf Return} $P_{cur}$, True 
    }}
    $P \leftarrow P \cup P_{cur}$ \\ 
}
\If{$c_2 == False$}{
$H.append((T_{l,mid}, score_{l,mid}))$ \\ 
$H.append((T_{mid+1,r}, score_{mid+1,r}))$ \\
}
\If{$|P| \geq k$}{{\bf Return } $P$}
    }
\end{algorithm}

\topic{Algorithm} 
\rtwo{We illustrate Algorithm~\ref{alg:top-k} for inferring $k$ provenances using a running example shown in Figure~\ref{fig:top-k}. We let the text $T$ be a list of $m$ equal-sized blocks, and we aim to find three minimal provenances ($k=3$).}

\rtwo{Initially, we start from $T$ (or equivalently, $T_{1,m}$, denoting the sequence of blocks from $B_1$ to $B_m$), represented by the root node 1. In the second iteration, we split $T_{1,m}$ into left and right halves, corresponding to its child node 2 $T_{1,\lceil m/2 \rceil}$ and node 3 $T_{\lceil m/2 \rceil + 1,m}$, respectively.} For each node, we compute a score as the likelihood that the corresponding block can reproduce $A$. Such a score is 1 if the answer $A'$ produced by the block is lexically identical to $A$; otherwise, we leverage LLMs to generate the score \papertext{(see prompt in \vldb{\cite{blip}}). }\techreport{using the prompt \code{Top-k-Eval-Prompt} below.}   
\rtwo{For example, the scores of nodes 1 and 2 are 1 and 0.9.} 

\techreport{\noindent{\underline{\code{Top-k-Eval-Prompt}}}: \textcolor{blue}{[Question]}, \textcolor{blue}{[Answer 1]}, and \textcolor{blue}{[Answer 2]} are placeholders for the question $Q$, the answer $L(T_{l,\text{mid}}, Q)$ evaluated on $T_{l,\text{mid}}$, and $A$, respectively, in Line 8-9 in Algorithm~\ref{alg:top-k}.  Note that when \textcolor{blue}{value} is \textit{False}, we use $1 -$~\textcolor{blue}{score} as its likelihood.

\mypython{
\small 
    Top-k-Eval-Prompt: Given the following question, \textcolor{blue}{[Question]}, and two answers, \textcolor{blue}{[Answer 1]} and \textcolor{blue}{[Answer 2]}. Determine whether the two answers are equivalent in meaning. Return the result as a JSON object with the following format: \textcolor{blue}{value}: true if the answers are equivalent, false otherwise. \textcolor{blue}{score}: a real number between 0 and 1 representing the likelihood that your judgment is correct. }}

\techreport{
\begin{figure}[]
    \centering
    \includegraphics[width=0.9\linewidth]{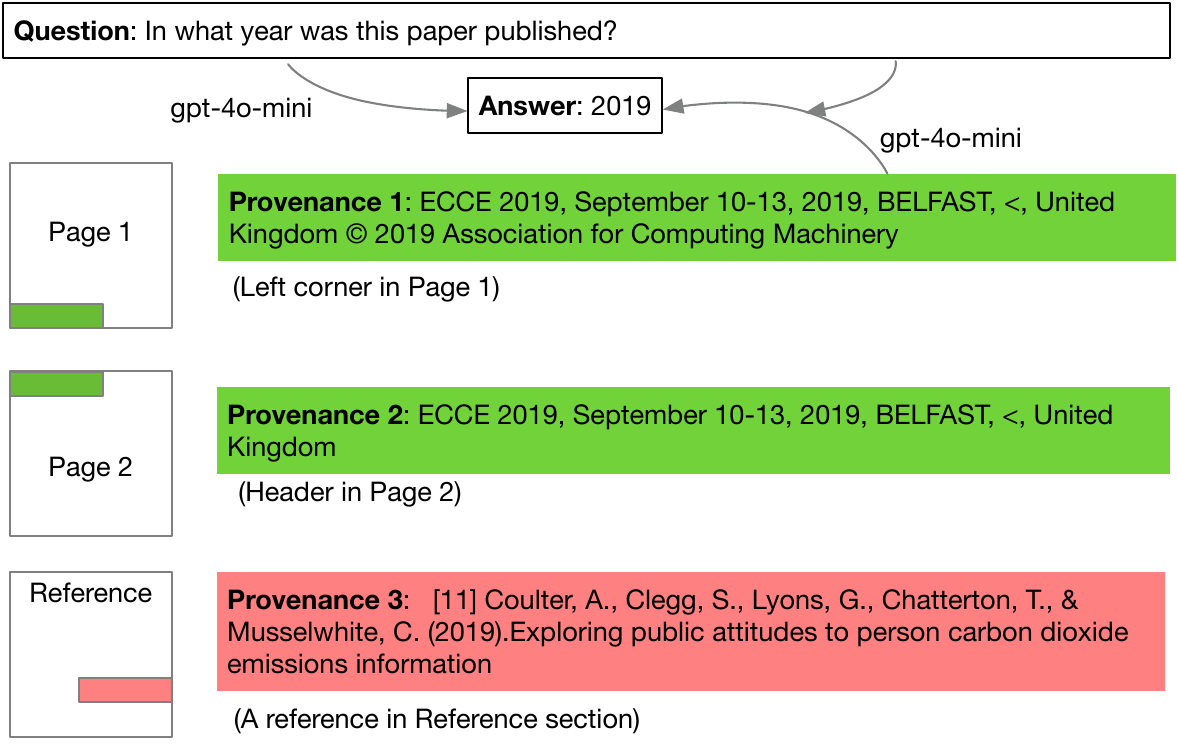}
    \caption{\small Top-k Provenance. }  
    \label{fig:top-k}
\end{figure}
}

We create a max heap $H$ that stores pairs consisting of a text block and its corresponding score. In each iteration, we expand the node with the highest likelihood score from the max heap and remove the node that is already expanded.  
\rtwo{For example, node 1 is removed from $H$ in the second iteration, and node 2 is expanded in the next iteration since it has the highest score.} We then insert its two children (node 4 and 5)—obtained by splitting the current text sequence into left and right halves—into the max heap and rank them based on their likelihood scores.

\rtwo{The above node splitting continues recursively until either of the following two conditions is met.} \rtwo{First, when we reach a node that contains only one block (e.g., node 8) and it passes verification (i.e., reproduces the answer), we return it by invoking \code{refine} to obtain the corresponding minimal provenance.} In this example, node 8 is the first provenance returned by the \code{top-k-provenance} algorithm\papertext{.}\techreport{, assuming it is distinct from the provenance returned by the two-phase approach.}  
\rtwo{Second, when we reach an internal node (e.g., node 6) whose score is no less than 0.5, but both of its child nodes have scores lower than 0.5 (e.g., nodes 10 and 11), node 6 is returned as the third provenance. This provenance is then further refined using the two-phase approach to obtain a minimal provenance.}

\techreport{\topic{Discussion} 
The \code{top-k-provenance} algorithm does not guarantee returning all distinct minimal provenances, but only a subset of them in $T$. We argue that this is often sufficient in many scenarios. If the question over $T$ is correctly answered by LLM $L$, then finding a single positive provenance suffices for user trust, and the algorithm can terminate without exploring all minimal provenances. On the other hand, if the answer returned by $L$ on $T$ is incorrect, ideally  all minimal provenances that are negative are expected to trust that answer is wrong. }

\begin{figure}[tb]
    \centering
    \vspace{-1mm}
    \includegraphics[width=1\linewidth]{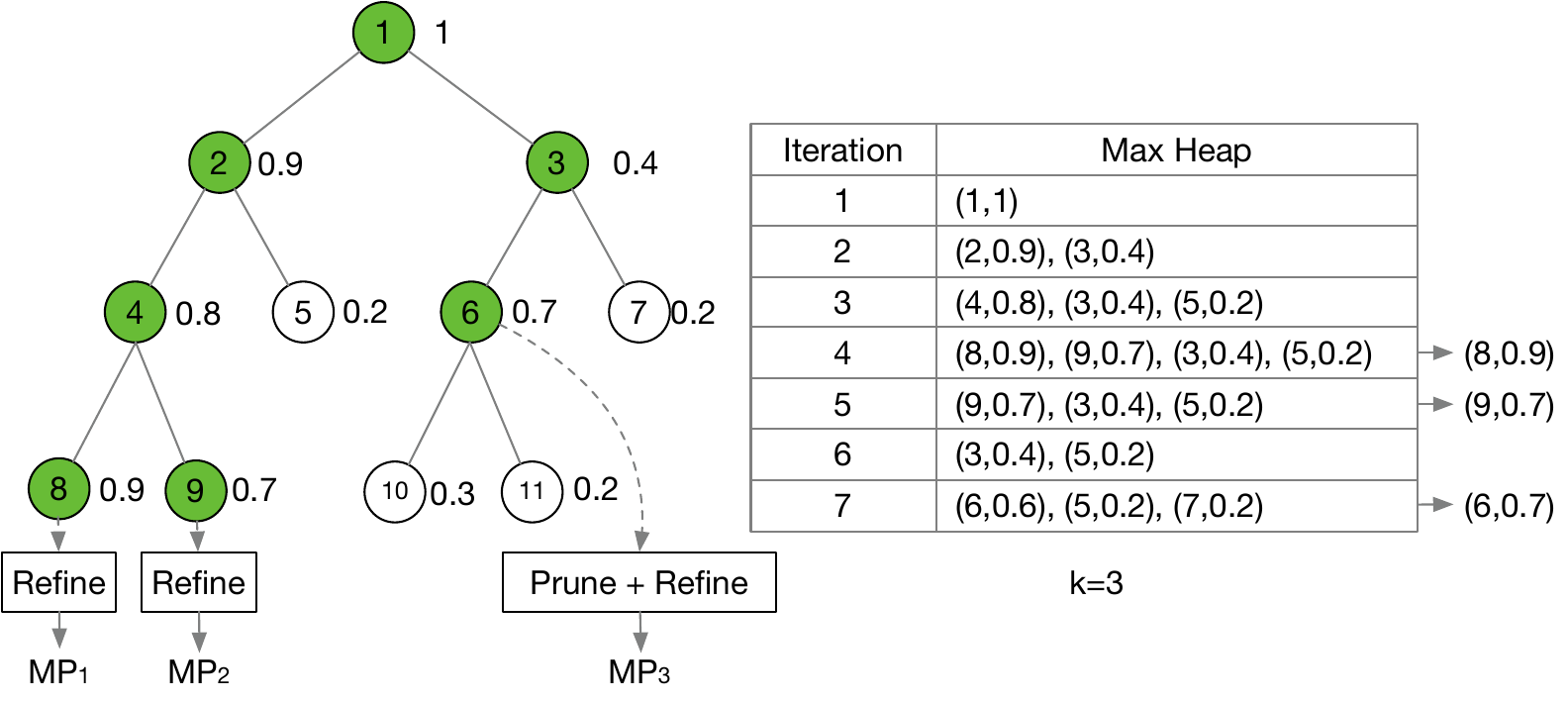}
    \vspace{-8mm}
    \caption{\small Top-$k$ Provenance.} 
    \vspace{-1mm}
    \label{fig:top-k} 
\end{figure}





\begin{table}[bt]
\small
\papertext{\vspace{-1mm}}
\centering
\scalebox{0.7}{
\begin{tabular}{|c|c|c|c|c|c|c|c|}
\hline
& \multicolumn{4}{c|}{\textbf{Question Types}} & \multicolumn{3}{c|}{\textbf{\# of Hops}}  \\ \hline
& \textbf{Look-up} & \textbf{Aggregation} & \textbf{Reasoning} & \textbf{Judge} & \textbf{1-Hop} & \textbf{2-Hops} & $\geq$\textbf{3-Hops} \\ \hline
\textbf{Qasper} & 54.1\% & 4\% & 24.4\% & 17.6\% & 72.3\% & 22.2\% & 5.6\% \\ \hline 
\textbf{NL\_DEV} & 75.2\% & 4\% & 18\% & 2.8\% & 42.3\% & 52.1\% & 5.6\% \\ \hline 
\textbf{HotpotQA} & 68.1\% & 1.8\% & 22.8\% & 7.4\% & 7.6\% & 57.3\% & 35.1\% \\ \hline
\textbf{\rthree{CUAD}} & \rthree{17.1\%} & \rthree{6.3\%} & \rthree{5.2\%} & \rthree{71.4\%} & \rthree{7.3\%} & \rthree{58.5\%} & \rthree{34.1\%} \\ \hline 
\textbf{\rthree{PubMedQA}} & \rthree{6\%} & \rthree{2\%} & \rthree{51\%} & \rthree{41\%} & \rthree{9.6\%} & \rthree{56.2\%} & \rthree{34.2\%} \\ \hline 
\end{tabular}}
\caption{\small Characteristics of Questions across Workloads.}
\papertext{\vspace{-4mm}}
\label{tab:query-workloads}
\end{table}

\begin{table*}[bt]
\small
\centering
\papertext{\vspace{-13pt}}
\scalebox{0.65}{
\begin{tabular}{|c|c|c|c|c|c|c|c|c|c|c|}
\hline
& & \multicolumn{3}{c|}{\textbf{Qasper} (8972 Tokens)} & \multicolumn{3}{c|}{\textbf{NL\_DEV} (7957 Tokens)} & \multicolumn{3}{c|}{\textbf{HotpotQA} (1233 Tokens)} \\ \hline 
 \textbf{Type} & \textbf{Strategy} & \textbf{Size Ratio (\%)} & \textbf{Cost Ratio} & \textbf{Latency} & \textbf{Size Ratio (\%)} & \textbf{Cost Ratio} &\textbf{Latency} & \textbf{Size Ratio (\%)} & \textbf{Cost Ratio} &\textbf{Latency} \\ 
\hline
\multirow{6}{*}{\textbf{Pruning Strategies}} & LLM\_bottom\_up & 9.7 & \textbf{0.2}  & \textbf{2} & 13.6 & 0.3 & 2 & 19.8 & \textbf{0.3} & 1.2 \\
&LLM\_top\_down & 12.6 & 0.9  & 3 & 15.3 & 1.2 & 2.9 & 22.3 & 1.3 &1.9\\
&LLM\_adaptive & 9.8 & \textbf{0.2}  & 2.1 & 13.5 & 0.3 & 2 & 20 & \textbf{0.3} & \textbf{1.1} \\ \cline{2-11}
&Embedding\_bottom\_up & 10.5 & 0.36 & 2.1 & 12.2 & \textbf{0.2}& \textbf{1.9} &  20.2 & \textbf{0.3} & 1.3\\
&Embedding\_top\_down & 12.4 & 1.1 & 3.1 & 14.8 & 1.1& 2.6 & 23.2 & 1.2 & 2\\
&Embedding\_adaptive & 10.6 & 0.35 & 2.1 & 12.1 & \textbf{0.2}& \textbf{1.9} &  20.3 & \textbf{0.3} & 1.2\\
\hline  \hline
\multirow{12}{*}{\textbf{\shortstack{Two-Phase Strategy: \\Pruning Strategies\\ + Refinement Strategies}}} 
&LLM\_bottom\_up\_EXP & 3.8 & \textbf{1.2} & 8.4  & 4.8 & \textbf{1.1}& \textbf{8.5} & 7.7 & \textbf{1.2} & \textbf{3.5} \\
&LLM\_bottom\_up\_SEQ & 3.6 & \textbf{1.3} & 9.7 & 4.9 & 1.3& 12.2 & 7.8 & 1.5 & 3.8\\
&LLM\_top\_down\_EXP& 3.7 & 1.7 & 11.2 &  5.1 & 1.9& 10.5 & 8.1 & 2.2 & 4.1\\
&LLM\_top\_down\_SEQ& 4.1 & 2.1 & 13.1 & 4.9 & 2.4& 14.2 & 8 & 2.9 & 4.7\\
&LLM\_adaptive\_EXP & 3.7 & \textbf{1.2} & \textbf{8.1}  & 4.8 & \textbf{1.1}& \textbf{8.5} & 7.7 & \textbf{1.1} & \textbf{3.3}  \\ 
&LLM\_adaptive\_SEQ & 3.6 & \textbf{1.3} & 9.3 & 4.8 & \textbf{1.2}& 11.4 & 7.8 & 1.4 & \textbf{3.7}\\ \cline{2-11}
&Embedding\_bottom\_up\_EXP& 3.4 & \textbf{1.3} & \textbf{8.2}    & 4.6 & \textbf{1.1}& \textbf{8.1} & 7.9 & \textbf{1.2} & 4.1 \\
&Embedding\_bottom\_up\_SEQ& 3.4 & 1.4 & 11.3 & 4.7 & \textbf{1.2}& 9.9 & 7.8 & 1.4 & 4.5\\
&Embedding\_top\_down\_EXP& 4.3 & 1.9 & 12.9 & 5.2 & 2.2& 10.8&8.2 & 2.4 & 5.8\\
&Embedding\_top\_down\_SEQ& 4.1 & 2.4 & 15.2 &  4.8 & 2.7& 13.5& 7.8 & 2.9 & 6.4\\
&Embedding\_adaptive\_EXP& 3.3 & \textbf{1.3} & \textbf{8}    & 4.6 & \textbf{1.1}& \textbf{8.3} & 7.9 & \textbf{1.1} & 4 \\
&Embedding\_adaptive\_SEQ& 3.4 & 1.4 & 10.8 & 4.7 & \textbf{1.2}& 9.6 & 7.8 & \textbf{1.3} & 4.2\\
\hline
\end{tabular}}
\caption{\small Size Ratio, Cost Ratio, and Latency of Strategies in \code{Qasper}, \code{NL\_DEV}, and \code{HotpotQA} Workloads. The top-1 pruning strategies, and the top-3 two-phase strategies are highlighted in bold for both cost ratio and latency. \code{EXP} and \code{SEQ} stand for \code{Exponential-Greedy} and \code{Sequential-Greedy}. }
\papertext{\vspace{-2.5em}}
\techreport{\vspace{-2em}}
\label{tab:top-1}
\end{table*}


\section{Evaluation} 
\label{sec:exp}

\begin{table}[bt]
\small
\centering
\scalebox{0.75}{
\begin{tabular}{l|l|c|c|c|c|c}
\hline
\textbf{Method} & \textbf{Configuration} & \textbf{Qasper} & \textbf{NL\_DEV} & \textbf{HotpotQA} & \rthree{\textbf{CUAD}} & \rthree{\textbf{PubMedQA}} \\ \hline
\multirow{3}{*}{\textbf{LLM}} & \textbf{gpt-4o-mini} & 0.4 & 0.37 & 0.46  & \rthree{0.42} & \rthree{0.39}\\ 
&\textbf{gpt-4o} & 0.36 & 0.26 & 0.3 & \rthree{0.31} & \rthree{0.34} \\ 
&\textbf{gemini-2-flash} & 0.6 & 0.65 & 0.47 & \rthree{0.47} & \rthree{0.42}\\ \hline 
\multirow{3}{*}{\textbf{Retrieval}} & \textbf{RAG-1\%} & 0.82 & 0.84  & 0.89  & \rthree{0.78} & \rthree{0.92}\\ 
&\textbf{RAG-5\%} & 0.59 & 0.64 & 0.7  & \rthree{0.66} & \rthree{0.81}\\ 
&\textbf{RAG-10\%} & 0.32 & 0.41 & 0.54 & \rthree{0.48} & \rthree{0.61}\\ \hline 
\textbf{BLIP}& {\bf E\_adaptive\_EXP} & \textbf{0.033} & \textbf{0.046} & \textbf{0.079} & \rthree{\textbf{0.03}} & \rthree{\textbf{0.12}}\\ \hline 
\end{tabular}}
\caption{\small \revised{Human Review Effort. (E stands for Embedding in BLIP.)}} 
\label{tab:human_review_time} 
\papertext{\vspace{-9mm}}
\techreport{\vspace{-7mm}}
\end{table}

We evaluate \trs{the performance of} \sys 
through experiments conducted on \rthree{five real-world}, two synthetic workloads, \rmix{as well as a user study}. 


\begin{table}[bt]
\scriptsize
\centering
\scalebox{1}{
\begin{tabular}{c|c|c|c}
\hline
\textbf{Dataset}  & \textbf{Size Ratio(\%)} & \textbf{Cost Ratio} & \textbf{Latency}  \\ \hline
CUAD &  2.8 &    1.4 &   9.8 \\ \hline
PubMedQA & 11.3 & 1.2 & 4.2 \\ \hline 
\end{tabular}}
\caption{\small \rthree{Evaluation Results of \sys (the overall recommended strategy, Embedding\_adaptive\_EXP) on CUAD and PubMedQA.}} 
\label{tab:cuad} 
\papertext{\vspace{-3mm}}
\end{table}

\begin{figure*}[tb]
\centering
\subfigure[Ranker = Embedding\label{fig:seq-exp-embedding}]{\includegraphics[width=0.4\textwidth]{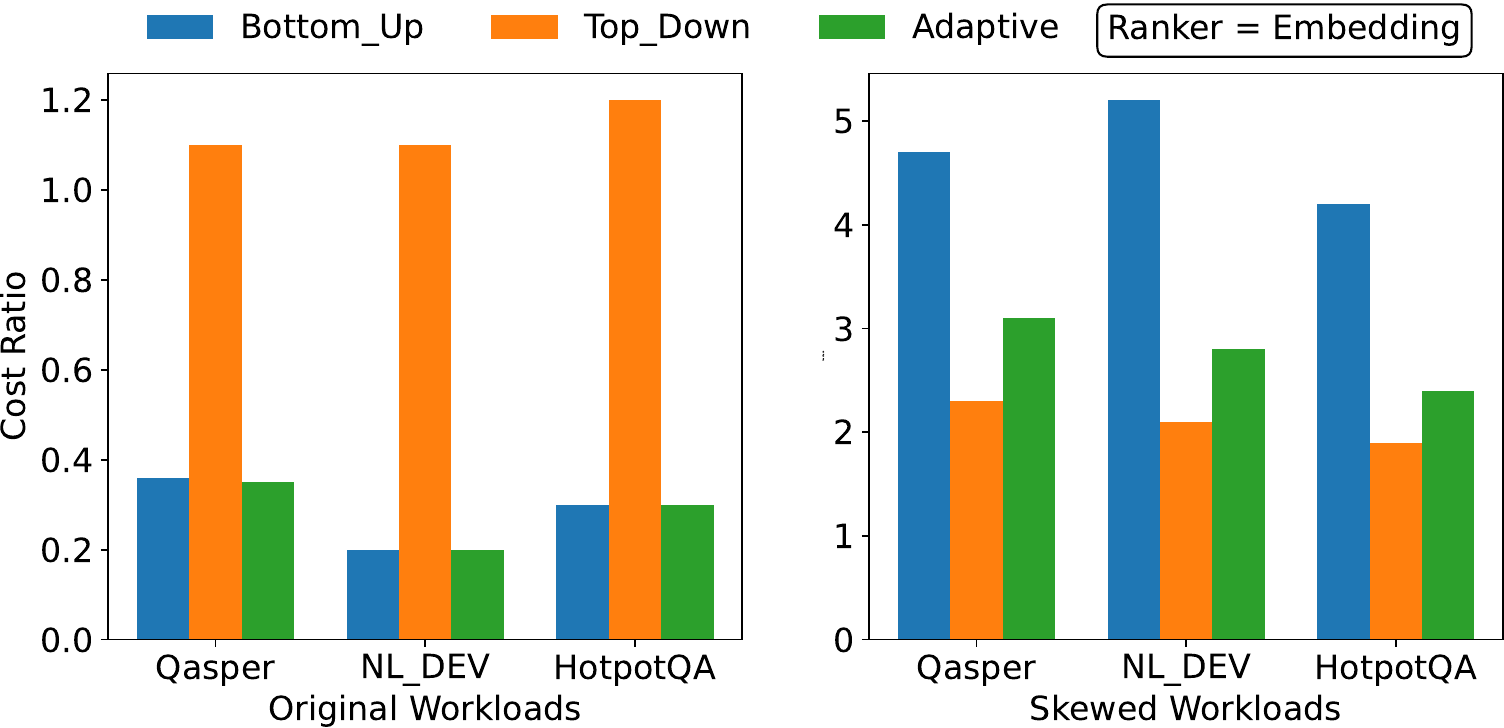}}
\subfigure[Ranker = LLM\label{fig:seq-exp-llm}]{\includegraphics[width=0.4\textwidth]{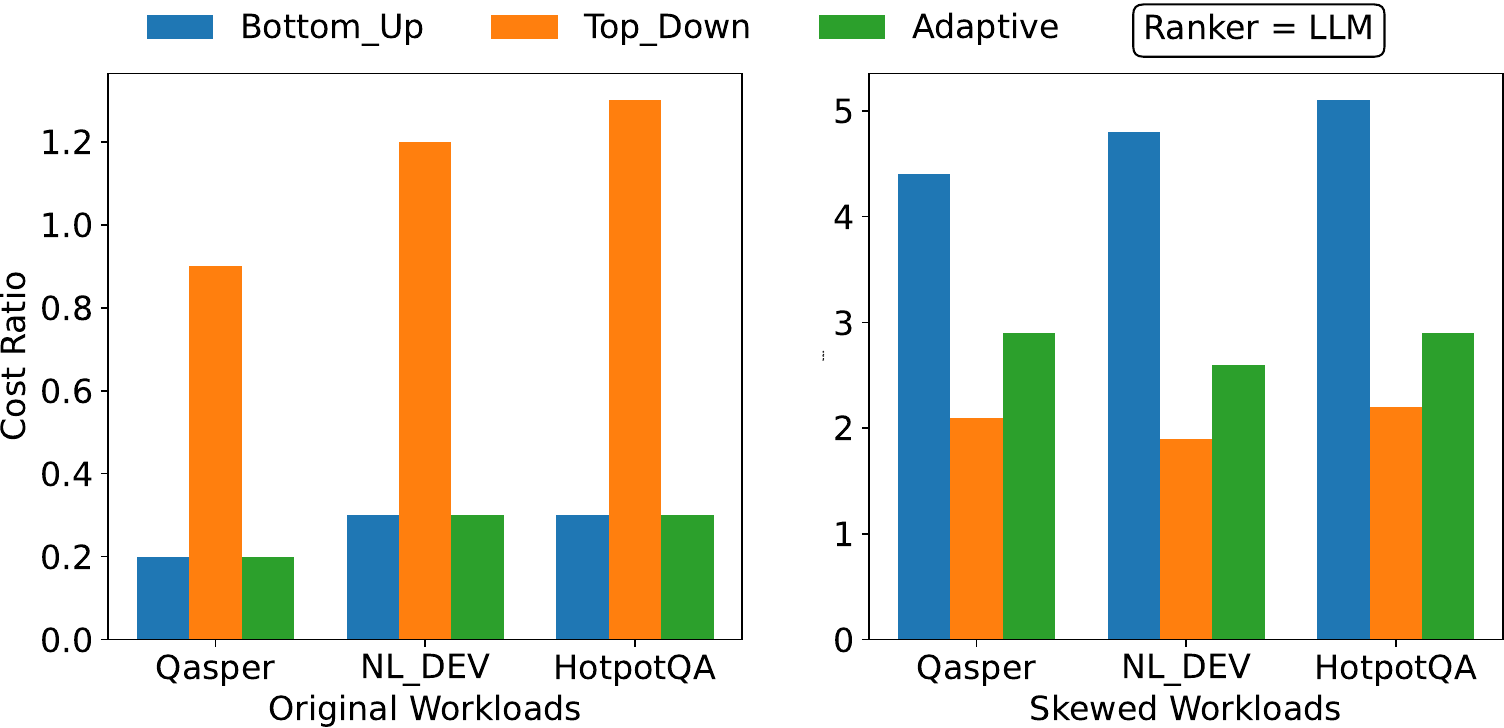}}
\vspace{-1.7em}
\caption{\small Cost ratios of pruning strategies on three real-world workloads and their skewed versions. }
\label{fig:cost-prune}
\vspace{-1.5em}
\end{figure*}

\subsection{Evaluation Setup}
\label{subsec:eval-setup}

\topic{Datasets \& Workloads} \rthree{We conduct experiments on five  commonly used question-answering workloads: {\em Qasper~\cite{dasigi2021dataset}}, {\em NL\_DEV}~\cite{kwiatkowski2019natural}, {\em HotpotQA~\cite{yang2018hotpotqa}}, {\em CUAD}~\cite{hendrycks2021cuad}, and {\em PubMedQA}~\cite{jin2019pubmedqa}, covering \techreport{the domains of} scientific papers, Wikipedia, legal contracts, and medical records.} For each workload, we sampled 500 distinct question-document pairs, ensuring that each question and document in a pair is unique. \rthree{The average tokens per document is 8,972, 7,957, 1,233, 9,964, and 891 for Qasper, NL\_DEV, HotpotQA, CUAD, and PubMedQA, respectively.} \revised{Questions with diverse types and complexities are evaluated: about \rthree{72\%} require at least two reasoning hops, as shown in Table~\ref{tab:query-workloads}; \papertext{see~\cite{blip} for details. }}\techreport{Here, Look-up and Aggregation denote factual questions, Judge corresponds to boolean questions (i.e., yes/no), and Reasoning questions ask why or how, typically requiring longer explanatory answers. } 
We additionally include two synthetic workloads, {\em Movie} and {\em Restaurant}, from TQA-bench~\cite{wolff2025well}, which focus on question answering over tables (TableQA), each containing  64,000 tokens and 100 questions.   

True provenance is available for TableQA workloads but not for others. Since the latter five workloads are on question answering over plain text, true provenance is unknown, and even human annotators may disagree on the ground truth; there may also be multiple true provenances. 
For the TableQA workloads, Movie and Restaurant, true provenance (i.e., the set of tuples) is provided for question-answer pairs. We focus on multi-hop reasoning questions in TableQA workloads, which are common analytical queries over tables. Out of 100 questions per TableQA workload, half are single-hop reasoning questions, while the others involve 2 to 3 hops.  



\topic{Compared Strategies} We first compare the five strategies in the pruning phase, along with ten combined two-phase strategies. 
\code{Sequential-Greedy} and \code{Exponential-Greedy} are abbreviated as \code{SEQ} and \code{EXP}, respectively.  So \code{LLM\_bottom\_up\_EXP} represents a strategy that uses \code{LLM\_bottom\_up} and \code{Exponential-Greedy} in the pruning and refinement phase, respectively.  \revised{We also compare \sys against a retrieval-based approach, RAG, using various context sizes (1\%, 5\%, and 10\% of the full text), and against LLM-generated provenance as mentioned in Section~\ref{sec:introduction}.}

\topic{Implementation} When performing \code{prune} in Algorithm~\ref{alg:linear-search}, we divide the text into $m=20$ equal-sized blocks. The embedding- and LLM-based rankers used to sort the blocks are implemented with the open-source embedding model \code{all-mpnet-base-v2} and LLM \code{Mistral 7B}. Both  embeddings and relevance scores are computed offline\techreport{for each document-question pair}. For question answering and provenance inference, we use \code{gpt-4o-mini} and \code{gemini-2-flash}. Results reported here are based on \code{gpt-4o-mini}. \code{Gemini-2-flash} has similar findings, \papertext{see \vldb{\cite{blip}}.}\techreport{in Table~\ref{tab:top-1-gemini}.}

\topic{Metric} \revised{We first report the {\em accuracy}, defined in Section~\ref{sec:introduction}, as the fraction of document–question pairs for which the inferred provenance can reproduce an answer equivalent to the one from the full text.}  We further report the {\em size ratio}, defined as the number of tokens in the generated provenance divided by that in the full text. We additionally report the {\em latency} and the {\em cost ratio} that is defined as the cost of finding the provenance divided by the cost of question answering. 
For strategies using an embedding-based ranker, we exclude the latency for building embeddings and report only the latency for provenance inference, as embeddings can be computed offline. 
\revised{Finally, observing that an incorrect provenance requires humans to read the entire document, we report {\em human review effort}, measured as the fraction of the document corresponding to the provenance (i.e., size ratio) if accurate (i.e., the provenance is verifiable), and as one plus the size ratio otherwise.} 

We also measure whether methods can recover the true provenance $P_{true}$ for  TableQA workloads  by introducing two metrics: recovery ($R$) and exact-recovery ($eR$). 
\techreport{Provenance is represented as a set of tuples in the tables.} $R$ is defined as the percentage of questions where a predicted provenance $P_{pre}$ contains $P_{true}$ \techreport{over all questions}, i.e., $R = \frac{|P_{true} \subseteq P_{pre}|}{|QA|}$, while $eR$ requires $P_{pre}$ to be an exact match with $P_{true}$, i.e., $eR = \frac{|P_{true} = P_{pre}|}{|QA|}$. 
\rtwo{We note that the {\em accuracy} defined measures how often a provenance inferred by \sys can reproduce the answer (i.e., is verifiable), while {\em recovery} measures whether the inferred provenance matches the true provenance.} 


\vspace{-1em}
\subsection{Experiments}  
\label{subsec:exp-result}

\topic{Experiment 1: Cost, Latency, and Provenance Size of \sys}  
We present the cost ratio, latency, and the size ratio of the provenance generated by \sys in Table~\ref{tab:top-1}. We make several observations. First, strategies in the pruning phase effectively eliminate text portions unrelated to answers, with adaptive strategies  reducing average provenance size ratios to 10.2\%, 12.8\%, and 20.2\% for Qasper, NL\_DEV, and HotpotQA,  respectively. These strategies incur only 0.28$\times$, 0.25$\times$, and 0.3$\times$ the cost of answering questions over the full text, with latencies averaging 2.1, 1.9, and 1.1 seconds. Among them, \code{LLM\_adaptive} {\bf \em reduces provenance size to around 9.8\% of the full text, achieving a 0.2$\times$ cost ratio and an average latency of 2.1 seconds} across the workloads. 

Strategies in $\mathcal{S}_{refine}$ further reduce the provenance size ratio to about {\bf \em 3.7\%, 4.7\%, and 7.7\% on Qasper, NL\_DEV, and HotpotQA}, respectively. Here, \code{Exponential\_Greedy} (i.e., strategies ending with \code{EXP}) have up to a {\bf \em  1.1$\times$ cost ratio and an average latency of 8.1 seconds} across workloads, demonstrating that \delete{the two-phase strategies are  practical solutions for producing small-size provenances that reproduce the answer, while incurring only a slightly higher cost than answering the question over the full text.  Latency can be further improved via parallelization, as part of our future work}\cradd{two-phase strategies produce minimal provenances at only slightly higher cost than full-text QA. Latency can be further reduced via parallelization}. \rthree{We have consistent observations on the CUAD and PubMedQA  workloads, as shown in Table~\ref{tab:cuad}, where \sys incurs a slightly higher cost (with a cost ratio of 1.3) than question answering over the entire input to infer a minimal provenance efficiently. }

\vspace{1pt}
\topic{\revised{Experiment 2: Accuracy and Human Review Effort}}   \rthree{We have already presented the accuracy for retrieval-based (RAG) and LLM-generated provenance in Table~\ref{tab:agreement}, where \sys strategies guarantee {\em \bf an accuracy of 1, over 30\% higher than the best baseline} with a comparable provenance size (i.e., LLM-generated provenance using \texttt{gpt-4o}).} 
\rthree{\delete{Note that non-monotonic tasks account for around 5\%  across datasets (based on Table~\ref{tab:tasks}), and accuracy being 1 demonstrates that \sys is guaranteed to infer verifiable provenance. This is consistent with our theoretical result in Theorem~\ref{theo:optimal}.}\cradd{Accuracy of 1 on non-monotonic tasks ($\sim$5\% of queries; Table~\ref{tab:tasks}) is consistent with Theorem~\ref{theo:optimal}.} }

We report the human review effort in Table~\ref{tab:human_review_time}. \rthree{Compared with RAG-based and LLM-generated provenance, \sys requires substantially less human review effort---only {\bf \em 3\% to 12\%} across five workloads. } 
This is because \sys always returns a minimal verifiable provenance, both effective and efficient for human verification.

\techreport{
\begin{table}[]
\centering
\small
\scalebox{0.8}{
\begin{tabular}{|c|c|c|c|c|c||c|}
\hline
& & & \multicolumn{3}{c|}{{\bf Bottom-up vs. Top-down}} & {\bf EXP vs. SEQ} \\ \hline 
 \textbf{Dataset} & \textbf{Ranker}& {\bf $\mathcal{S}_{Prune}$} & \textbf{$CP_1$} & \textbf{Q1} & \textbf{Q3} &  {\bf $CP_2$ (Pct)} \\
\hline
\multirow{4}{*}{\textbf{Qasper}} &  \multirow{2}{*}{LLM} & BU & 8.2   & 1  &2  & 8.4 (35.2\%) \\
& & TD & 8.2  & 1 & 2  & 9.3 (32.7\%) \\ \cline{2-7}
                    & \multirow{2}{*}{Embedding} & BU & 7.4  & 1 &3  & 8.9 (31.9\%) \\
& & TD & 7.4  & 1 & 3 & 9.5 (28.6\%) \\ \cline{1-7}
\multirow{4}{*}{\textbf{NL\_DEV}} & \multirow{2}{*}{LLM} & BU &6.9  &1 & 3 & 8.5 (38.4\%) \\
& & TD & 6.9  & 1 & 3 & 9.1  (35.2\%)\\ \cline{2-7}  
                    & \multirow{2}{*}{Embedding} & BU & 6.7   &1 &  2  & 8.2 (40.1\%) \\
& & TD & 6.7 & 1 & 2 &  8.9 (38.5\%) \\ \cline{1-7}
\multirow{4}{*}{\textbf{HotpotQA}} & \multirow{2}{*}{LLM}&BU & 7.2  &1 & 3 & 7.5 (40.2\%)\\
& & TD & 7.2 &1 &3 & 8.1 (38.4\%) \\ \cline{2-7}
                    & \multirow{2}{*}{Embedding} &BU &  7.4  &1 & 2  & 7.2 (38.1\%) \\
& &TD & 7.4  & 1 & 2 & 8.9 (37.5\%) \\ 
\hline
\end{tabular}}
\caption{\small Crossover Points ($CP_1$ and $CP_2$). BU and TD stand for bottom\_up and top\_down, respectively. 
$Q1$ and $Q3$ are the first and third quartiles of the returned data blocks when comparing bottom-up and top-down. $Pct$ stands for the percentile of cases where the number of input sentences is less than $CP_2$.  }
\papertext{\vspace{-1.5em}}
\techreport{\vspace{-1em}}
\label{tab:distribution}
\end{table}
}


\techreport{\begin{table*}[bt]
\small
\centering
\scalebox{0.8}{
\begin{tabular}{|c|c|c|c|c|c|c|c|c|c|c|}
\hline
& & \multicolumn{3}{c|}{\textbf{Qasper} (8972 Tokens)} & \multicolumn{3}{c|}{\textbf{NL\_DEV} (7957 Tokens)} & \multicolumn{3}{c|}{\textbf{HotpotQA} (1233 Tokens)} \\ \hline 
 \textbf{Type} & \textbf{Strategy} & \textbf{Size Ratio (\%)} & \textbf{Cost Ratio} & \textbf{Latency} & \textbf{Size Ratio (\%)} & \textbf{Cost Ratio} &\textbf{Latency} & \textbf{Size Ratio (\%)} & \textbf{Cost Ratio} &\textbf{Latency} \\ 
\hline
\multirow{6}{*}{\textbf{Pruning Strategies}} & LLM\_bottom\_up & 10.3 & \textbf{0.3}  & 2 & 14.1 & 0.4 & 2.2 & 17.3 & 0.3 & 1.7 \\
&LLM\_top\_down & 13.4 & 0.8  & 2.7 & 14.7 & 1.4 & 3.2 & 19.5 & 1.1 &2.2\\
&LLM\_adaptive & 10.2 & \textbf{0.3}  & \textbf{1.8} & 13.9 & 0.3 & \textbf{2.1} & 17.7 & \textbf{0.2} & \textbf{1.6} \\ \cline{2-11}
&Embedding\_bottom\_up & 9.4 & 0.5 & 1.9 & 13.1 & \textbf{0.3}& 2.4 &  18.7 & 0.4 & 1.8\\
&Embedding\_top\_down & 10.8 & 1.2 & 2.8 & 14.5 & 1.2& 3.1 & 20.5 & 1.2 & 2.5\\
&Embedding\_adaptive & 9.7 & 0.4 & 2 & 13.1 & \textbf{0.3}& \textbf{2.2} &  18.9 & 0.4 & 1.7\\
\hline  \hline
\multirow{12}{*}{\textbf{\shortstack{Two-Phase Strategy: \\Pruning Strategies\\ + Refinement Strategies}}} 
&LLM\_bottom\_up\_EXP & 4.1 & \textbf{0.9} & \textbf{7.5}  & 5.2 & \textbf{0.9}& \textbf{9.1} & 7.3 & \textbf{0.8} & \textbf{3.7} \\
&LLM\_bottom\_up\_SEQ & 3.9 & \textbf{1.1} & 9.4 & 5 & \textbf{1.1}& 11.9 & 7.2 & 1.1 & 3.9\\
&LLM\_top\_down\_EXP& 4 & 1.4 & 10.4 &  5.1 & 1.6& 10.9 & 7.6 & 1.9 & 4\\
&LLM\_top\_down\_SEQ& 4.3 & 1.8 & 14.2 & 5.4 & 2.1& 13.7 & 7.8 & 2.3 & 4.3\\
&LLM\_adaptive\_EXP & 4 & \textbf{0.9} & 8.6  & 5.2 & \textbf{1.1}& \textbf{8.9} & 7.4 & \textbf{0.8} & \textbf{3.5}  \\ 
&LLM\_adaptive\_SEQ & 3.9 & \textbf{0.9} & 8.8 & 4.9 & \textbf{1}& 12.9 & 7.4 & \textbf{0.9} & \textbf{3.7}\\ \cline{2-11}
&Embedding\_bottom\_up\_EXP& 3.6 & 1.3 & \textbf{7.6}    & 5 & \textbf{1}& \textbf{8.4} & 7.3 & \textbf{1} & 3.9 \\
&Embedding\_bottom\_up\_SEQ& 3.7 & \textbf{1.1} & 10.6 & 4.9 & \textbf{1.1}& 9.6 & 7.6 & 1.4 & 4.2\\
&Embedding\_top\_down\_EXP& 3.9 & 1.7 & 11.3 & 5.2 & 1.7& 12.1&7.9 & 1.8 & 4.6\\
&Embedding\_top\_down\_SEQ& 4 & 2 & 13.9 &  5.3 & 2.2& 14.7& 7.8 & 2.3 & 5.7\\
&Embedding\_adaptive\_EXP& 3.6 & \textbf{1.1} & \textbf{8.2}    & 4.9 & \textbf{0.9}& \textbf{8.2} & 7.5 & \textbf{0.9} & \textbf{3.6} \\
&Embedding\_adaptive\_SEQ& 3.7 & \textbf{1.2} & 10.3 & 4.8 & \textbf{1}& \textbf{9.1} & 7.6 & \textbf{1} & 3.8\\
\hline
\end{tabular}}
\caption{\small Gemini-2-flash: Size Ratio, Cost Ratio, and Latency of Strategies in \code{Qasper}, \code{NL\_DEV}, and \code{HotpotQA} Workloads. The top-1 pruning strategies, and the top-3 two-phase strategies are highlighted in bold for both cost ratio and latency. \code{EXP} and \code{SEQ} stand for \code{Exponential-Greedy} and \code{Sequential-Greedy}. }
\vspace{-2em}
\label{tab:top-1-gemini}
\end{table*}}

\vspace{1pt}
\topic{Experiment 3: Bottom\_up, Top\_down, and Adaptive  strategies} 
Although bottom\_up and adaptive strategies outperform their top\_down counterparts in Table~\ref{tab:top-1}, we provide a breakdown to understand why. Section~\ref{subsubsec:bottom-top} shows that the cost of these strategies depends on the number of returned data blocks, with a suggested crossover point of $CP \approx 8.4$ when the text is divided into 20 blocks. We create a skewed dataset by upsampling question-document pairs so that 80\% of the pairs have more returned data blocks than $CP$. Figure~\ref{fig:cost-prune} reports the cost ratios on both the original and skewed workloads. \delete{It shows that top-down strategies are preferred when the number of returned data blocks is large—being over 2$\times$ cheaper than their bottom-up counterparts on the skewed workloads, with adaptive strategies providing a middle ground throughout.}\cradd{Top-down is over 2$\times$ cheaper on the skewed workloads, while adaptive strategies provide a middle ground in both settings.}
\techreport{To understand the trend of cost relative to the number of returned data blocks, we further plot the cost ratio of  \code{LLM\_bottom\_up},  \code{LLM\_top\_down}, and \code{LLM\_adaptive} for various numbers of returned data blocks in Qasper in Figure~\ref{fig:adaptive}. 
The cost of the bottom\_up strategy grows much faster than that of top\_down—it is lower when fewer than 8.2 data blocks are returned but higher beyond that. This is because top\_down examines blocks exponentially, while bottom\_up proceeds linearly, resulting in a roughly quadratic cost ratio relative to the number of returned data blocks.}  
\techreport{We also report the distribution of returned data blocks in Table~\ref{tab:distribution}, where the actual crossover points, denoted by $CP_1$, are presented when comparing bottom\_up and top\_down. The results show that most question-document pairs have a small number of returned blocks, making bottom\_up more cost-effective on average, whereas top\_down is preferred in the skewed workloads. }

\delete{Overall, for pruning, if an estimate of the number of returned data blocks is available from historical data, one can choose a strategy by comparing this estimate with $CP$. Otherwise, the adaptive strategy wins, as it combines the strengths of both, being efficient for small outputs while avoiding quadratic cost for large ones.}\cradd{Overall, we recommend the adaptive strategy when provenance size is unknown, as it combines both strategies' strengths while avoiding worst-case costs.}

\techreport{
\begin{figure}[tb]
    \centering
    \includegraphics[width=0.9\linewidth]{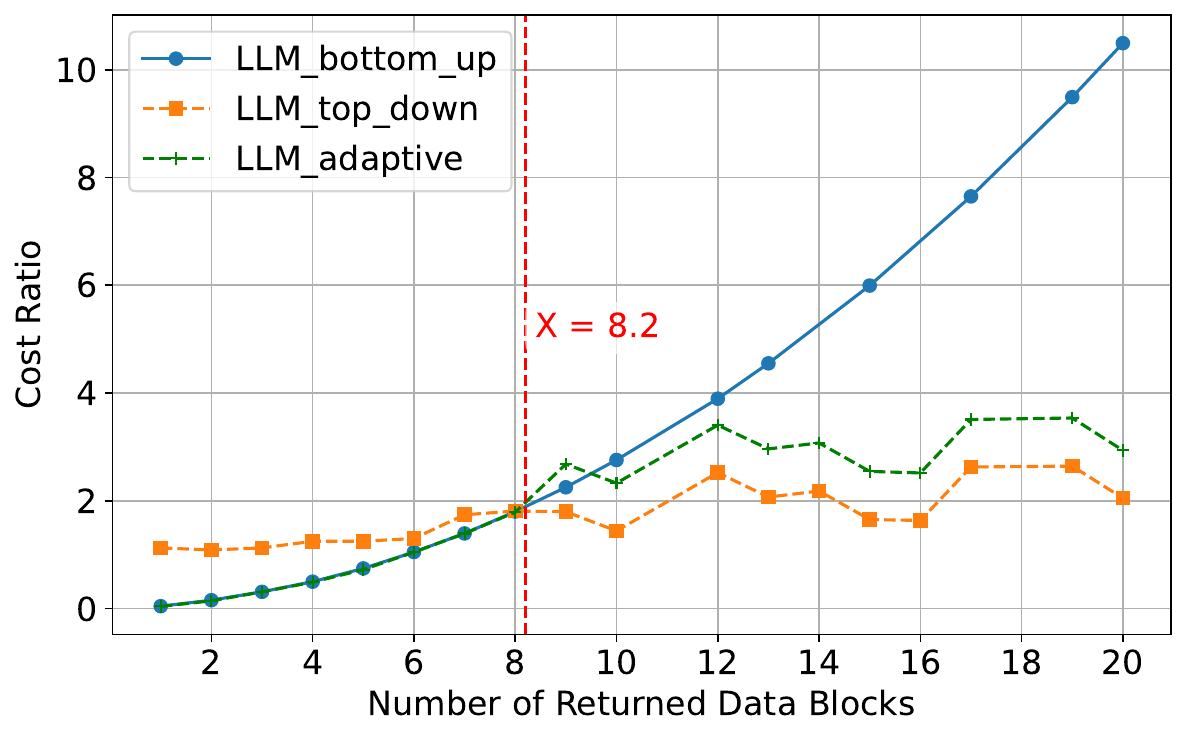}
    \vspace{-1em}
    \caption{\small Cost Ratio of top\_down and bottom\_up on Qasper.} 
    \papertext{\vspace{-0.7em}}
    \label{fig:adaptive}
\end{figure}
}

\begin{figure*}[tb]
    \vspace{-2mm}
\centering

\begin{minipage}[b]{0.24\textwidth}
    \centering
    \vspace{-2mm}
    \includegraphics[width=\linewidth]{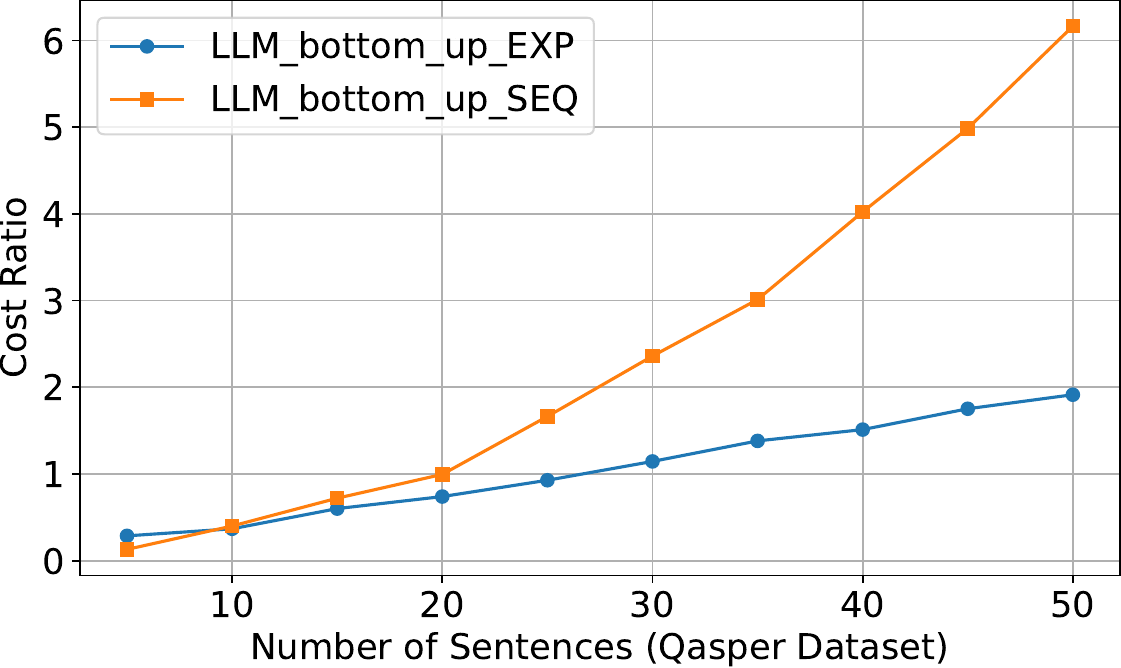}
    \vspace{-8.3mm}
    \captionof{figure}{\small Cost Ratio. (\code{EXP} vs. \code{SEQ})}
    \label{fig:cost-exp-seq}
\end{minipage}
\hfill
\begin{minipage}[b]{0.75\textwidth}
    \centering
    \vspace{-2mm}
    \includegraphics[width=0.32\linewidth]{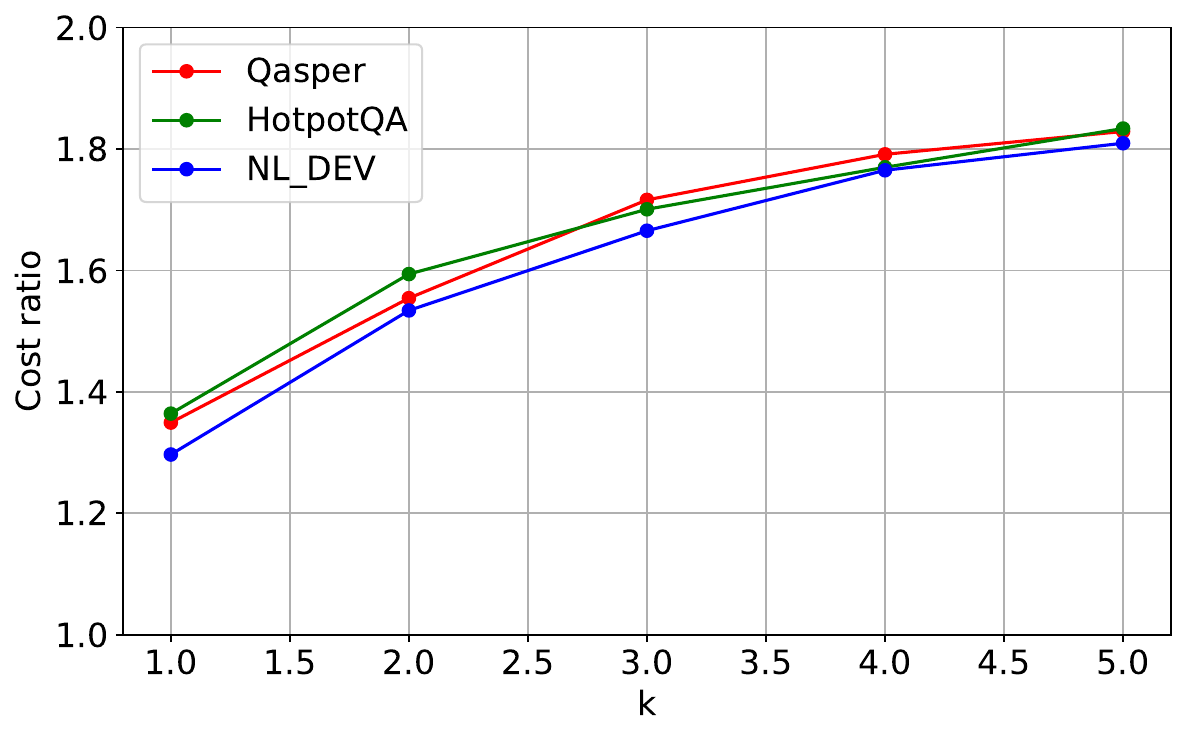}\hfill
    \includegraphics[width=0.32\linewidth]{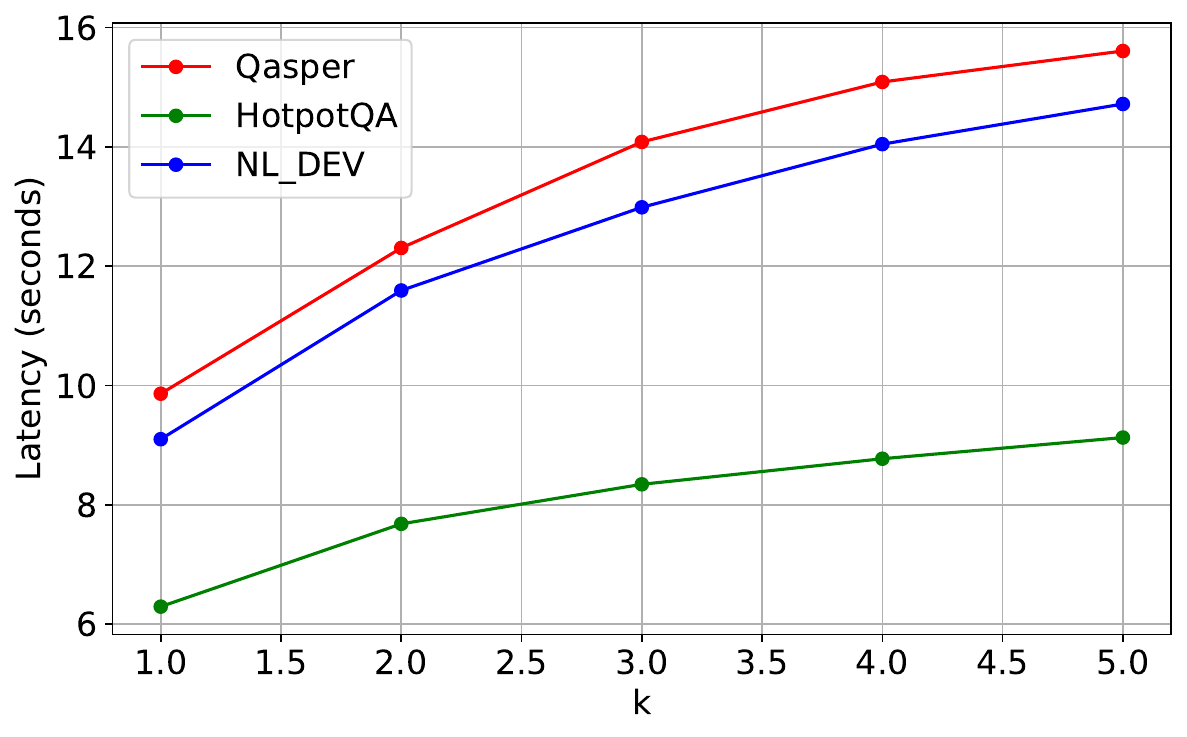}\hfill
    \includegraphics[width=0.32\linewidth]{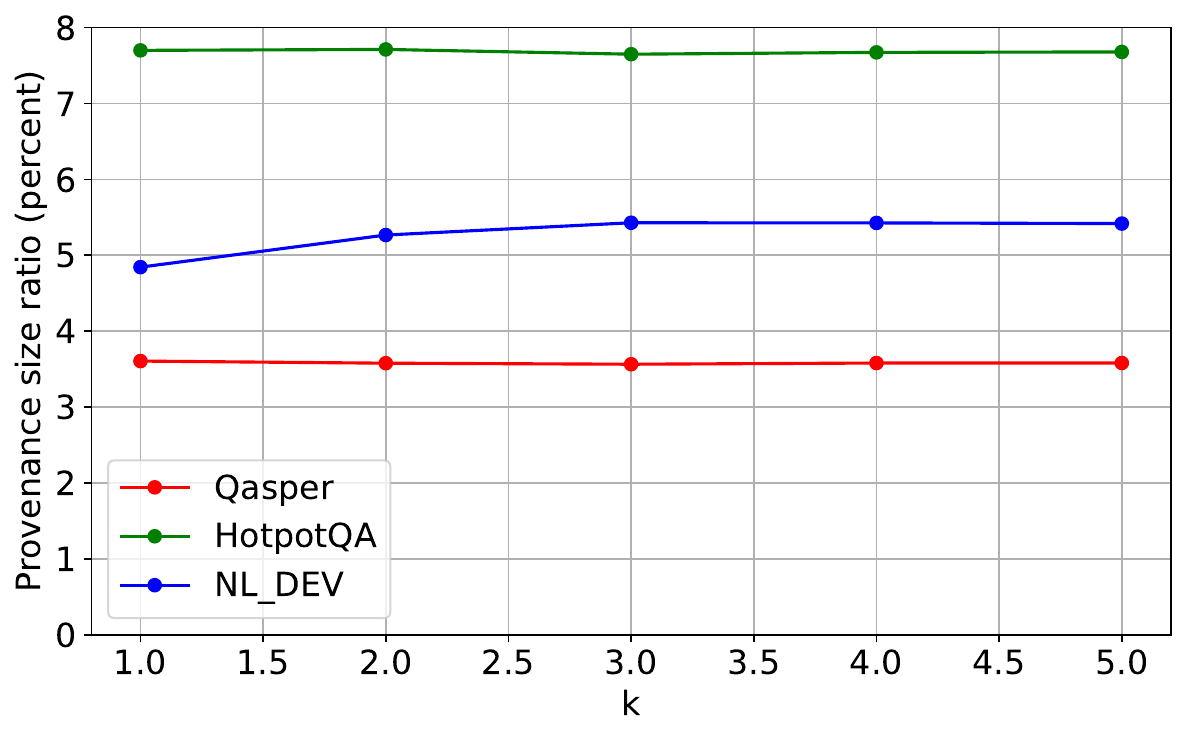}
    \vspace{-5mm}
    \captionof{figure}{\small Cost Ratio, Latency, and Provenance Size Ratio of Top-$k$ Provenance.}
    \label{fig:top-k-eval}
\end{minipage}
\papertext{\vspace{-5mm}}
\end{figure*}

\begin{table}[bt]
\small
\vspace{-10pt}
\centering
\scalebox{0.63}{
\begin{tabular}{|c|c|c|c|c|c|c|c|c|c|c|c|c|c|}
\hline
\textbf{Methods} & \textbf{Models} & \multicolumn{6}{c|}{\textbf{Movie (64k Tokens)}} & \multicolumn{6}{c|}{\textbf{Restaurant (64k Tokens)}}  \\ \hline & & \multicolumn{3}{c|}{\textbf{$Q_{simple}$}} & \multicolumn{3}{c|}{\textbf{$Q_{complex}$}} & \multicolumn{3}{c|}{\textbf{$Q_{simple}$}} & \multicolumn{3}{c|}{\textbf{$Q_{complex}$}} \\ \hline 
& & \textbf{$eR$} & \textbf{$R$} & \rtwo{\textbf{$A$}} &  \textbf{$eR$} & \textbf{$R$} & \rtwo{\textbf{$A$}} & \textbf{$eR$} & \textbf{$R$} & \rtwo{\textbf{$A$}} & \textbf{$eR$} & \textbf{$R$} & \rtwo{\textbf{$A$}} \\ \hline
\textbf{\sys} & \textbf{E\_adaptive}  & \textbf{1} & \textbf{1} & \rtwo{\textbf{1}} & \textbf{0.98} & \textbf{1} & \rtwo{\textbf{1}} & \textbf{0.98} & \textbf{1}  & \rtwo{\textbf{1}} & \textbf{0.96} & \textbf{1} & \rtwo{\textbf{1}} \\ \hline
\multirow{3}{*}{\textbf{LLM}} & \textbf{gpt-4o-mini} & 0.62 & 0.76 & \rtwo{0.74} & 0.58 & 0.75 & \rtwo{0.73} & 0.72 & 0.84 & \rtwo{0.84} & 0.73 & 0.81 & \rtwo{0.79} \\  
&  \textbf{gpt-4o} & 0.64 & 0.89 & \rtwo{0.86} & 0.62 & 0.83 & \rtwo{0.82} & 0.74 & 0.9 & \rtwo{0.88} & 0.79 & 0.86 & \rtwo{0.86} \\ 
&  \textbf{gemini-2-flash} & 0.83 & 0.91 & \rtwo{0.9} & 0.67 & 0.8 & \rtwo{0.76} & 0.7 & 0.86 & \rtwo{0.83} & 0.73 & 0.79 & \rtwo{0.78} \\ \hline 
\multirow{4}{*}{\textbf{Retrieval}} & \textbf{RAG-EXACT} & 0.44 & 0.44 & \rtwo{0.44} & 0.32 & 0.32 & \rtwo{0.32} & 0.68 & 0.68 & \rtwo{0.68}  & 0.62 & 0.62 & \rtwo{0.62} \\ 
& \textbf{RAG-1\%} & 0 & 0.48 & \rtwo{0.48} &  0 & 0.42& \rtwo{0.41} & 0 & 0.82 & \rtwo{0.81}& 0 & 0.74& \rtwo{0.74} \\
& \textbf{RAG-5\%} & 0 & 0.72& \rtwo{0.7} & 0 & 0.58& \rtwo{0.56} & 0 & 0.92& \rtwo{0.91} & 0 & 0.86& \rtwo{0.86} \\ 
& \textbf{RAG-10\%} & 0 & 0.8& \rtwo{0.76} & 0 & 0.74& \rtwo{0.71} & 0 & 0.96& \rtwo{0.93} & 0 & 0.94& \rtwo{0.93} \\ \hline
\end{tabular}}
\caption{\small \revised{Recovery} and \rtwo{Accuracy} in TableQA Workloads. (E stands for Embedding.)}
\label{tab:accuracy}
\vspace{-1em}
\end{table}


\vspace{1pt}
\topic{Experiment 4: Exponential-Greedy (EXP) versus Sequential-Greedy (SEQ) Strategies} Consider \code{SEQ} and \code{EXP} \trs{in the refine phase} in Table~\ref{tab:top-1}.  \code{EXP} consistently outperforms \code{SEQ}, 11.6\%, 12\%, and 18.6\% cheaper, 
 and 1.3$\times$, 1.34$\times$, and 1.1$\times$ faster for Qasper, NL\_DEV, and HotpotQA, respectively. When the size of the input provenance is larger, \code{EXP} shows greater improvement over \code{SEQ}, up to 24\% cheaper and 1.49$\times$ faster. \delete{To provide a breakdown of their performance, we plot their cost ratios conditioned on the number of sentences in the provenance returned from the pruning phase in Figure~\ref{fig:cost-exp-seq} on Qasper. Both strategies perform similarly on small inputs, with \code{SEQ} slightly preferred when the number of sentences is less than 10\trs{ in Figure~\ref{fig:cost-exp-seq}}, while \code{EXP} scales better to larger inputs.}\cradd{Figure~\ref{fig:cost-exp-seq} plots their cost ratios vs.\ input size on Qasper: \code{SEQ} is preferred below 10 sentences, while \code{EXP} scales better beyond that.}   \techreport{Similar observations hold for other strategies and workloads in Table~\ref{tab:distribution}, where $CP_2$ denotes the crossover point based on the number of input sentences. Specifically, when the input contains fewer than $CP_2$ sentences, \code{SEQ} has a lower cost for more than half of the question-document pairs. An average percentile (Pct for $CP_2$) below 40\%, implying that only less than 40\% of inputs favor \code{SEQ}, explains why \code{EXP} has an overall lower cost. } 
Overall, when the number of sentences is small (10 is an empirical crossover point), \code{SEQ} is recommended for its simplicity and lower cost. Otherwise, \code{EXP} is clearly preferred.


\vspace{1pt}
\topic{Experiment 5: \revised{Recovery} of True Provenance in TableQA Workloads} 
To measure how well methods can recover the true provenance, we report both recovery ($R$) and exact recovery ($eR$) \rtwo{along with accuracy ($A$) for \sys} against four variants of RAG- and LLM-generated provenances in Table~\ref{tab:accuracy}.  \sys uses an adaptive strategy using embedding model as the ranker, where \code{SEQ} or \code{EXP} is chosen based on the provenance size (fewer than 10 sentences or not) returned in the previous phase. 
\techreport{All RAGs rank tuples based on embedding similarity to the question–answer pair, using the same embedding model as \sys.}  The top-$k$ tuples are used in RAG-EXACT, where $k$ matches the size of the true provenance. In the RAG-$p$\% strategy, the top-$p$\% of tuples are returned.

\sys can always recover the superset of true provenance as indicated by $R=1$, while it recovers the exact true tuples for over 98\% and 96\% of questions in Movie and Restaurant workloads, respectively. 
In contrast, RAG-EXACT achieves significantly lower $eR$, recovering the  true provenance in only 32\% to 68\% of cases across the two workloads.  Increasing the size of the returned provenance to 10\% (RAG-10\%) improves recovery but introduces many false positives (over 50 times larger than the ground-truth provenance), making it less interpretable. \techreport{For example, in Restaurant, RAG-10\% returns 106 tuples, while the true provenance contains only 1 tuple, undermining its usefulness.} 
LLM-based baselines recover the exact ground-truth provenance in 50--70\% of cases across different LLMs and workloads, and fail to capture the true in over 20\% of complex workloads.  \rtwo{Finally, the provenance inferred by \sys can always reproduce the answer, as indicated by $A=1$, while other approaches have no guarantee of producing verifiable provenance. }


\vspace{1pt}
\topic{Experiment 6: Top-k Provenance}  
We examine the performance of \code{top-k-provenance} with $k$ ranging from 1 to 5 in Figure~\ref{fig:top-k-eval}. 
 The algorithm stops until {\em up to} $k$ minimal provenances are found. 
As $k$ increases, the cost ratio (and similarly, latency) shows a modest increase, rising from approximately 1.3 to 1.8 across all workloads. This is because during the tree-based search in \code{top-k-provenance}, intermediate results computed for earlier provenances can often be reused when identifying subsequent ones. Consequently, later provenances frequently begin from smaller initial text sequences, leading to lower computation cost and latency. \delete{For example, identifying the provenance for node 9 in Figure~\ref{fig:top-k} starts from the text corresponding to node 9 rather than the full text, making the computation faster when returning the second provenance.}\cradd{For example, the second provenance starts from node 9's subtree in Figure~\ref{fig:top-k} rather than the full text, reducing cost.}  
\techreport{Morever, the average size of the returned minimal provenances remains stable, which is consistent with our findings in Table~\ref{tab:top-1}. }


\begin{table}[bt]
\papertext{\vspace{-10pt}}
\scriptsize
\centering
\scalebox{0.9}{
\begin{tabular}{|c|c|c|c|c|c|c|c|c|}
\hline
& \multicolumn{2}{c|}{\textbf{police\_reports}} & \multicolumn{2}{c|}{\textbf{legal}}  & \multicolumn{2}{c|}{\textbf{paper}} & \multicolumn{2}{c|}{\textbf{finance}} \\ \hline
& $H_q$ & $H_e$ & $H_q$ & $H_e$ & $H_q$ & $H_e$ & $H_q$ & $H_e$ \\ \hline 
\textbf{\sys} & \textbf{4.1}  & \textbf{75\%} &\textbf{5} & \textbf{100\%} &\textbf{4.9} & \textbf{88\%} &\textbf{5} & \textbf{63\%} \\ \hline 
\textbf{RAG} & 3.1 & 38\% &3.1 & 25\% &3.3 & 20\% &2.8 & 0\% \\ \hline 
\end{tabular}}
\caption{\rtwo{\small User Study. ($H_q$ ranges from 1 to 5 and measures the usefulness of the provenance, while $H_e$ is the percentage of cases in which a provenance contains the right amount of information.)}}
\papertext{\vspace{-3mm}}
\label{tab:user-study}
\end{table}

\vspace{-3mm}
\subsection{User Study}
\label{subsec:user-study}

\rtwo{Our experimental evaluation thus far
has focused on the accuracy and recoverability of provenance,
as opposed to its usefulness.
We conducted a user study to assess 
whether humans
find the provenance returned by \sys to be useful 
 when validating answers from LLMs.
We selected four documents from 
varied domains: 
{\em police\_reports}, 
{\em legal\_docs},  {\em scientific papers}, and {\em finance\_docs}, drawn from our legal partners~\cite{clean} and open-source benchmarks~\cite{dasigi2021dataset,finance}. 
We selected three questions per document, 
testing three scenarios: 
1) provenance generated by \sys; 
2) RAG-generated provenance, where up to 5\% (or 500 tokens, whichever is smaller) for ease of human interpretation; and 3) no provenance. 
Our goal was to compare the effectiveness of
\sys's provenance with that returned by RAG.
The last scenario was targeted at 
understanding human review effort to
determine correctness  
when no provenance is provided. 
We had eight participants in our user study,
who considered each scenario for each document, amounting
to 12 scenarios in total. 
\techreport{We randomized the order of settings across participants.}
\papertext{The details of other setups  
can be found in~\cite{blip}.}} 

\noindent \rtwo{{\bf Metrics.} 
Participants rated the {\em quality}
of a returned provenance from 1 (low)--5 (high),
assessing its usefulness in supporting
the answer returned by the LLM;
we let $H_q$ be the average quality across
participants for each dataset.
We also asked whether the provenance was both necessary and sufficient (i.e., containing no extra information and missing no critical information); let $H_e$ denote the percentage of such cases, averaged across participants per dataset. }

\rtwo{Table~\ref{tab:user-study} 
reports the average $H_q$ and $H_e$ over eight participants.
Overall, participants found the provenance 
inferred by \sys to be useful for judging answer correctness, 
as indicated by the overall $H_q$ of 4.75 (out of 5)\papertext{.}  
\techreport{across all documents.} 
One participant (P3) commented that ``{\em this seems to contain 
the right amount of information, 
including proper context to figure out the answer.}'' 
\delete{Cases where participants were not satisfied were
primarily due to semantic ambiguity\papertext{,}
\techreport{in the concepts contained in the provenance, }
where the LLM may carry different interpretations from humans.
For example, to identify officers using force with an answer ``John Smith'', when interpreting the provenance ``{\em Use of force details. Involved officer: reporting officer -- John Smith}\footnote{Name anonymized for privacy.}'',
the term ``reporting officer'' causes
ambiguity as to whether John
actually used force, indicating more context
would have been helpful for assessing validity.}\cradd{Dissatisfaction arose mainly from semantic ambiguity between LLM and human interpretations. For example, the provenance ``{\em Use of force details. Involved officer: reporting officer -- John Smith}\footnote{Name anonymized for privacy.}'' left it unclear whether John actually used force, suggesting more context could help.}}

\rtwo{RAG-based provenance was deemed to be 
less effective, with an overall $H_q$ of 3.1. One participant (P2) complained: “{\em This looks more like a search result than provenance, and provenance should explain why this answer. }” 
Participants also found \sys's 
provenance to contain the right amount 
of information 82\% of the time, 61\% higher 
than that for RAG-based provenance. 
\delete{When participants felt that \sys
returned extra information, this was typically because
some of this information was not necessary given
the participants' domain knowledge. For example, to identify the number of employees at Amazon from its 10-K financial report with the answer ``29,239'', users can interpret the provenance ``{\em we employed 29,239 people...}'' using their knowledge that ``we'' refers to Amazon. In contrast, the additional provenance returned by \sys that includes the earlier mention of ``Amazon'' (prior to ``we'') is often deemed unnecessary for users, though it is required for LLMs to interpret the answer.}\cradd{When \sys returned extra information, it was typically unnecessary given users' domain knowledge. For example, users inferred that ``{\em we employed 29,239 people...}'' referred to Amazon without the earlier ``Amazon'' mention \sys includes---context needed by LLMs but not humans.}}


\rtwo{In the no-provenance setting, or when the provenance was insufficient (e.g., for RAG), participants often reverted to keyword search. As P7 noted:
“{\em Evidence containing keywords that match the question and answer feels more trustworthy. But I'm not sure if I selected the right keywords due to false positives.}” This keyword search requires multiple passes, typically around 2 distinct keyword search queries, skimming over 18 matches on average
to locate relevant content. 
This inefficiency worsens when keywords 
are frequently present in the document or when documents are large. }

\rtwo{\delete{We found two potential areas of improvement.
First, when the returned provenance was large,
even if minimal,
it was difficult
to interpret.
Providing mechanisms to summarize this provenance
\techreport{(as in relational provenance summarization~\cite{lee2020approximate,alomeir2023summarizing})} may be helpful, such as
highlighting keywords that match the question or answer.}\cradd{We found two areas of improvement. First, large provenance, even if minimal, was hard to interpret; highlighting keywords matching the question or answer could help.} “{\em Highlighting sentences within a long provenance and indicating their correspondence to each part of the answer helps a lot}”,  as noted by P4. 
Second, providing mechanisms to ``tighten'' or ``expand''
the provenance (e.g., dropping the least relevant portions,
or including the most relevant excluded portions), 
depending on the degree of familiarity
with the domain, may also be helpful.  
}



\vspace{-2mm}
\section{Related Work}
\label{sec:relatedwork}
\vspace{-1mm}

Our work is related to provenance in relational databases, scientific workflows, and LLM inference. 

\topic{Database Provenance} We compared against database provenance in Section~\ref{subsec:database-prov-relationship}. Here, we provide a more general background. 
Database provenance~\cite{cheney2009provenance,buneman2007provenance,glavic2021data,karvounarakis2010querying} focuses on tracking data lineage to ensure reproducibility. 
Given a SQL query and inputs, provenance investigates: 
(1) the input tuples that produce the output (why-provenance)~\cite{cui2000tracing,buneman2001and,buneman2002propagation}, (2) the specific locations (e.g., table, row, or column) in the input from which elements in the output originate (where-provenance)~\cite{buneman2001and,bhagwat2005annotation}, and (3) the transformations applied to the input to derive the output (how-provenance)~\cite{green2007update,green2007provenance,green2007orchestra,hernandez2021computing}. 
Our setting poses additional challenges:
namely the ambiguity of queries and data,
the black-box nature of LLMs,
and dealing with LLM costs in inferring provenance.

\topic{Provenance in Data Science Pipelines} 
Provenance has been widely studied in scientific workflow systems, which use a dataflow model.
VisTrails~\cite{scheidegger2008querying} enhances provenance query efficiency by enabling users to reuse common sub-workflows. 
Subzero~\cite{wu2013subzero} leverages data locality, where multiple operators in a workflow often share similar input cells. 
DfAnalyzer~\cite{silva2018dfanalyzer} tracks provenance at runtime without waiting for workflow completion.
Adriane et al.~\cite{chapman2020capturing} focuses on machine learning workflows, computing provenance to debug model performance. 
\sys instead focuses on computing verifiable provenance within a single LLM invocation, and extending \sys to workflows would be interesting future work. 



\topic{Provenance in LLM Inference}
Overall, LLM provenance has only been studied in limited settings. Prior work, including
LLM Attributor~\cite{lee2025llm},  ContextCite~\cite{cohen2024contextcite}, and Hithesh et al~\cite{sankararaman2024provenance} all identify top-k relevant tokens as provenance, using embedding similarities, next-token probabilities, or relevance scores predicted by pretrained models. However, such heuristic methods lack verifiability. 
Notably, these approaches can be integrated into \sys as candidate ranking models in our pruning phase. \trs{Note that LLM Attributor~\cite{lee2025llm} requires access to LLM training data, which is often unavailable.} 
\vspace{-1mm}
\section{Conclusion}
\label{sec:conclusion}
Ensuring trust in LLM-powered data processing tasks is critical for maintaining usability. We presented \sys, a framework for inferring verifiable minimal provenance—subsets of the input text that reproduce an equivalent answer to the one obtained from the full input. \sys introduces a set of strategies, each guaranteed to return a minimal verifiable provenance, including an adaptive strategy that combines the strengths of multiple strategies, to infer provenance more efficiently. \techreport{In scenarios where a single provenance is insufficient to establish confidence in the answer's correctness, \sys offers a strategy to identify $k$ distinct minimal provenances, enabling users to gain stronger trust.} Extensive experiments across five datasets demonstrate that \sys is guaranteed to return a provenance that can reproduce the answer at a low cost, achieving over 30\% higher accuracy than the best-performing baseline. 


\begin{acks}
We acknowledge support from grants DGE-2243822, IIS-2129008, IIS-1940759, and IIS-1940757 awarded by the National Science Foundation, funds from the State of California, funds from the Alfred P. Sloan Foundation, as well as EPIC lab sponsors: Adobe, Google, G-Research, Jane Street, Microsoft, PromptQL, Sigma Computing, Snowflake, and Bridgewater. Compute credits were provided by Azure, Modal, NSF (via NAIRR), and OpenAI. 
\end{acks}
\balance
\bibliographystyle{ACM-Reference-Format}
\bibliography{refs}

@String{Computing = "Computing" }

@String{Springer = "Springer-Verlag" }

@inproceedings{lee2025llm,
  title={Llm attributor: Interactive visual attribution for llm generation},
  author={Lee, Seongmin and Wang, Zijie J and Chakravarthy, Aishwarya and Helbling, Alec and Peng, ShengYun and Phute, Mansi and Chau, Duen Horng Polo and Kahng, Minsuk},
  booktitle={Proceedings of the AAAI Conference on Artificial Intelligence},
  volume={39},
  number={28},
  pages={29655--29657},
  year={2025}
}

@article{gao2023enabling,
  title={Enabling large language models to generate text with citations},
  author={Gao, Tianyu and Yen, Howard and Yu, Jiatong and Chen, Danqi},
  journal={arXiv preprint arXiv:2305.14627},
  year={2023}
}

@article{sankararaman2024provenance,
  title={Provenance: A Light-weight Fact-checker for Retrieval Augmented LLM Generation Output},
  author={Sankararaman, Hithesh and Yasin, Mohammed Nasheed and Sorensen, Tanner and Di Bari, Alessandro and Stolcke, Andreas},
  journal={arXiv preprint arXiv:2411.01022},
  year={2024}
}

@manual{decoding,
  title = {\url{https://huggingface.co/blog/mlabonne/decoding-strategies}},
  year = {2025}
    }

@manual{floating,
  title = {\url{https://www.taivo.ai/__are-llms-deterministic/?utm_source=chatgpt.com}},
  year = {2025}
    }

@manual{finance,
  title = {\url{https://www.sec.gov/data-research/sec-markets-data/financial-statement-data-sets}},
  year = {2026}
    }

@article{lee2020approximate,
  title={Approximate summaries for why and why-not provenance (extended version)},
  author={Lee, Seokki and Lud{\"a}scher, Bertram and Glavic, Boris},
  journal={arXiv preprint arXiv:2002.00084},
  year={2020}
}

@article{alomeir2023summarizing,
  title={Summarizing provenance of aggregate query results in relational databases},
  author={AlOmeir, Omar and Lai, Eugenie Y and Milani, Mostafa and Pottinger, Rachel},
  journal={IEEE Transactions on Knowledge and Data Engineering},
  volume={35},
  number={10},
  pages={10695--10709},
  year={2023},
  publisher={IEEE}
}

@manual{clean,
  title = {\url{https://journalism.berkeley.edu/berkeley-and-stanford-police-database/}},
  year = {2025}
    }

@manual{blip,
  title = {Technical Report of Bolt-on, Verifiable Provenance for LLM-Powered Data Processing; {\url{https://drive.google.com/file/d/16thcosiFUVvwiMr3pgzZRHNejOV52lNp/view?usp=sharing}}} ,
  year = {2025}
    }

@article{wolff2025well,
  title={How well do LLMs reason over tabular data, really?},
  author={Wolff, Cornelius and Hulsebos, Madelon},
  journal={arXiv preprint arXiv:2505.07453},
  year={2025}
}

@article{kwiatkowski2019natural,
  title={Natural questions: a benchmark for question answering research},
  author={Kwiatkowski, Tom and Palomaki, Jennimaria and Redfield, Olivia and Collins, Michael and Parikh, Ankur and Alberti, Chris and Epstein, Danielle and Polosukhin, Illia and Devlin, Jacob and Lee, Kenton and others},
  journal={Transactions of the Association for Computational Linguistics},
  volume={7},
  pages={453--466},
  year={2019},
  publisher={MIT Press One Rogers Street, Cambridge, MA 02142-1209, USA journals-info~…}
}

@article{cheney2009provenance,
  title={Provenance in databases: Why, how, and where},
  author={Cheney, James and Chiticariu, Laura and Tan, Wang-Chiew and others},
  journal={Foundations and Trends{\textregistered} in Databases},
  volume={1},
  number={4},
  pages={379--474},
  year={2009},
  publisher={Now Publishers, Inc.}
}

@inproceedings{buneman2007provenance,
  title={Provenance in databases},
  author={Buneman, Peter and Tan, Wang-Chiew},
  booktitle={Proceedings of the 2007 ACM SIGMOD international conference on Management of data},
  pages={1171--1173},
  year={2007}
}

@article{islam2023financebench,
  title={Financebench: A new benchmark for financial question answering},
  author={Islam, Pranab and Kannappan, Anand and Kiela, Douwe and Qian, Rebecca and Scherrer, Nino and Vidgen, Bertie},
  journal={arXiv preprint arXiv:2311.11944},
  year={2023}
}

@article{arora2025healthbench,
  title={Healthbench: Evaluating large language models towards improved human health},
  author={Arora, Rahul K and Wei, Jason and Hicks, Rebecca Soskin and Bowman, Preston and Qui{\~n}onero-Candela, Joaquin and Tsimpourlas, Foivos and Sharman, Michael and Shah, Meghan and Vallone, Andrea and Beutel, Alex and others},
  journal={arXiv preprint arXiv:2505.08775},
  year={2025}
}

@article{
liang2023holistic,
title={Holistic Evaluation of Language Models},
author={Percy Liang and Rishi Bommasani and Tony Lee and Dimitris Tsipras and Dilara Soylu and Michihiro Yasunaga and Yian Zhang and Deepak Narayanan and Yuhuai Wu and Ananya Kumar and Benjamin Newman and Binhang Yuan and Bobby Yan and Ce Zhang and Christian Alexander Cosgrove and Christopher D Manning and Christopher Re and Diana Acosta-Navas and Drew Arad Hudson and Eric Zelikman and Esin Durmus and Faisal Ladhak and Frieda Rong and Hongyu Ren and Huaxiu Yao and Jue WANG and Keshav Santhanam and Laurel Orr and Lucia Zheng and Mert Yuksekgonul and Mirac Suzgun and Nathan Kim and Neel Guha and Niladri S. Chatterji and Omar Khattab and Peter Henderson and Qian Huang and Ryan Andrew Chi and Sang Michael Xie and Shibani Santurkar and Surya Ganguli and Tatsunori Hashimoto and Thomas Icard and Tianyi Zhang and Vishrav Chaudhary and William Wang and Xuechen Li and Yifan Mai and Yuhui Zhang and Yuta Koreeda},
journal={Transactions on Machine Learning Research},
issn={2835-8856},
year={2023},
url={https://openreview.net/forum?id=iO4LZibEqW},
note={Upto date leaderboard at \url{https://crfm.stanford.edu/helm/lite/latest/}}
}

@article{liu2024large,
  title={Large language models in the clinic: a comprehensive benchmark},
  author={Liu, Fenglin and Li, Zheng and Zhou, Hongjian and Yin, Qingyu and Yang, Jingfeng and Tang, Xianfeng and Luo, Chen and Zeng, Ming and Jiang, Haoming and Gao, Yifan and others},
  journal={arXiv preprint arXiv:2405.00716},
  year={2024}
}

@article{guha2023legalbench,
  title={Legalbench: A collaboratively built benchmark for measuring legal reasoning in large language models},
  author={Guha, Neel and Nyarko, Julian and Ho, Daniel and R{\'e}, Christopher and Chilton, Adam and Chohlas-Wood, Alex and Peters, Austin and Waldon, Brandon and Rockmore, Daniel and Zambrano, Diego and others},
  journal={Advances in neural information processing systems},
  volume={36},
  pages={44123--44279},
  year={2023}
}

@article{glavic2021data,
  title={Data provenance},
  author={Glavic, Boris and others},
  journal={Foundations and Trends{\textregistered} in Databases},
  volume={9},
  number={3-4},
  pages={209--441},
  year={2021},
  publisher={Now Publishers, Inc.}
}

@inproceedings{karvounarakis2010querying,
  title={Querying data provenance},
  author={Karvounarakis, Grigoris and Ives, Zachary G and Tannen, Val},
  booktitle={Proceedings of the 2010 ACM SIGMOD International Conference on Management of data},
  pages={951--962},
  year={2010}
}

@article{cui2000tracing,
  title={Tracing the lineage of view data in a warehousing environment},
  author={Cui, Yingwei and Widom, Jennifer and Wiener, Janet L},
  journal={ACM Transactions on Database Systems (TODS)},
  volume={25},
  number={2},
  pages={179--227},
  year={2000},
  publisher={ACM New York, NY, USA}
}

@inproceedings{buneman2001and,
  title={Why and where: A characterization of data provenance},
  author={Buneman, Peter and Khanna, Sanjeev and Wang-Chiew, Tan},
  booktitle={Database Theory—ICDT 2001: 8th International Conference London, UK, January 4--6, 2001 Proceedings 8},
  pages={316--330},
  year={2001},
  organization={Springer}
}

@inproceedings{buneman2002propagation,
  title={On propagation of deletions and annotations through views},
  author={Buneman, Peter and Khanna, Sanjeev and Tan, Wang-Chiew},
  booktitle={Proceedings of the twenty-first ACM SIGMOD-SIGACT-SIGART symposium on Principles of database systems},
  pages={150--158},
  year={2002}
}

@inproceedings{green2007update,
  title={Update exchange with mappings and provenance},
  author={Green, Todd J and Karvounarakis, Grigoris and Ives, Zachary G and Tannen, Val},
  year={2007},
  organization={VLDB}
}

@inproceedings{green2007provenance,
  title={Provenance semirings},
  author={Green, Todd J and Karvounarakis, Grigoris and Tannen, Val},
  booktitle={Proceedings of the twenty-sixth ACM SIGMOD-SIGACT-SIGART symposium on Principles of database systems},
  pages={31--40},
  year={2007}
}

@inproceedings{green2007orchestra,
  title={ORCHESTRA: Facilitating collaborative data sharing},
  author={Green, Todd J and Karvounarakis, Grigoris and Taylor, Nicholas E and Biton, Olivier and Ives, Zachary G and Tannen, Val},
  booktitle={Proceedings of the 2007 ACM SIGMOD international conference on Management of data},
  pages={1131--1133},
  year={2007}
}

@article{bhagwat2005annotation,
  title={An annotation management system for relational databases},
  author={Bhagwat, Deepavali and Chiticariu, Laura and Tan, Wang-Chiew and Vijayvargiya, Gaurav},
  journal={The VLDB Journal},
  volume={14},
  pages={373--396},
  year={2005},
  publisher={Springer}
}

@article{hernandez2021computing,
  title={Computing how-provenance for SPARQL queries via query rewriting},
  author={Hern{\'a}ndez, Daniel and Gal{\'a}rraga, Luis and Hose, Katja},
  journal={Proceedings of the VLDB Endowment},
  volume={14},
  number={13},
  pages={3389--3401},
  year={2021},
  publisher={VLDB Endowment}
}

@article{chapman2020capturing,
  title={Capturing and querying fine-grained provenance of preprocessing pipelines in data science},
  author={Chapman, Adriane and Missier, Paolo and Simonelli, Giulia and Torlone, Riccardo},
  journal={Proceedings of the VLDB Endowment},
  volume={14},
  number={4},
  pages={507--520},
  year={2020},
  publisher={VLDB Endowment}
}

@article{silva2018dfanalyzer,
  title={DfAnalyzer: runtime dataflow analysis of scientific applications using provenance},
  author={Silva, V{\'\i}tor and de Oliveira, Daniel and Valduriez, Patrick and Mattoso, Marta},
  journal={Proceedings of the VLDB Endowment},
  volume={11},
  number={12},
  pages={2082--2085},
  year={2018},
  publisher={VLDB Endowment}
}

@inproceedings{scheidegger2008querying,
  title={Querying and re-using workflows with VsTrails},
  author={Scheidegger, Carlos E and Vo, Huy T and Koop, David and Freire, Juliana and Silva, Claudio T},
  booktitle={Proceedings of the 2008 ACM SIGMOD international conference on Management of data},
  pages={1251--1254},
  year={2008}
}

@inproceedings{jin2019pubmedqa,
  title={PubMedQA: A Dataset for Biomedical Research Question Answering},
  author={Jin, Qiao and Dhingra, Bhuwan and Liu, Zhengping and Cohen, William and Lu, Xinghua},
  booktitle={Proceedings of the 2019 Conference on Empirical Methods in Natural Language Processing and the 9th International Joint Conference on Natural Language Processing (EMNLP-IJCNLP)},
  pages={2567--2577},
  year={2019}
}

@article{russo2025abacus,
  title={Abacus: A Cost-Based Optimizer for Semantic Operator Systems},
  author={Russo, Matthew and Sudhir, Sivaprasad and Vitagliano, Gerardo and Liu, Chunwei and Kraska, Tim and Madden, Samuel and Cafarella, Michael},
  journal={arXiv preprint arXiv:2505.14661},
  year={2025}
}

@article{patel2024semantic,
  title={Semantic operators: a declarative model for rich, ai-based data processing},
  author={Patel, Liana and Jha, Siddharth and Pan, Melissa and Gupta, Harshit and Asawa, Parth and Guestrin, Carlos and Zaharia, Matei},
  journal={arXiv preprint arXiv:2407.11418},
  year={2024}
}

@article{sun2025quest,
  title={Quest: Query optimization in unstructured document analysis},
  author={Sun, Zhaoze and Deng, Qiyan and Chai, Chengliang and Jin, Kaisen and Guo, Xinyu and Han, Han and Yuan, Ye and Wang, Guoren and Cao, Lei},
  journal={arXiv preprint arXiv:2507.06515},
  year={2025}
}

@article{satriani2025logical,
  title={Logical and physical optimizations for sql query execution over large language models},
  author={Satriani, Dario and Veltri, Enzo and Santoro, Donatello and Rosato, Sara and Varriale, Simone and Papotti, Paolo},
  journal={Proceedings of the ACM on Management of Data},
  volume={3},
  number={3},
  pages={1--28},
  year={2025},
  publisher={ACM New York, NY, USA}
}

@article{jo2024thalamusdb,
  title={Thalamusdb: Approximate query processing on multi-modal data},
  author={Jo, Saehan and Trummer, Immanuel},
  journal={Proceedings of the ACM on Management of Data},
  volume={2},
  number={3},
  pages={1--26},
  year={2024},
  publisher={ACM New York, NY, USA}
}

@article{wang2025unify,
  title={Unify: A System For Unstructured Data Analytics},
  author={Wang, Jiayi and Li, Yuan and Wu, Jianming and Xu, Shihui and Li, Guoliang},
  journal={Proceedings of the VLDB Endowment},
  volume={18},
  number={12},
  pages={5287--5290},
  year={2025},
  publisher={VLDB Endowment}
}

@article{hendrycks2021cuad,
  title={Cuad: An expert-annotated nlp dataset for legal contract review},
  author={Hendrycks, Dan and Burns, Collin and Chen, Anya and Ball, Spencer},
  journal={arXiv preprint arXiv:2103.06268},
  year={2021}
}

@inproceedings{wang2018glue,
  title={GLUE: A multi-task benchmark and analysis platform for natural language understanding},
  author={Wang, Alex and Singh, Amanpreet and Michael, Julian and Hill, Felix and Levy, Omer and Bowman, Samuel},
  booktitle={Proceedings of the 2018 EMNLP workshop BlackboxNLP: Analyzing and interpreting neural networks for NLP},
  pages={353--355},
  year={2018}
}

@inproceedings{wu2013subzero,
  title={Subzero: a fine-grained lineage system for scientific databases},
  author={Wu, Eugene and Madden, Samuel and Stonebraker, Michael},
  booktitle={2013 IEEE 29th International Conference on Data Engineering (ICDE)},
  pages={865--876},
  year={2013},
  organization={IEEE}
}

@article{cohen2024contextcite,
  title={Contextcite: Attributing model generation to context},
  author={Cohen-Wang, Benjamin and Shah, Harshay and Georgiev, Kristian and Madry, Aleksander},
  journal={Advances in Neural Information Processing Systems},
  volume={37},
  pages={95764--95807},
  year={2024}
}

@article{dasigi2021dataset,
  title={A dataset of information-seeking questions and answers anchored in research papers},
  author={Dasigi, Pradeep and Lo, Kyle and Beltagy, Iz and Cohan, Arman and Smith, Noah A and Gardner, Matt},
  journal={arXiv preprint arXiv:2105.03011},
  year={2021}
}

@article{yang2018hotpotqa,
  title={HotpotQA: A dataset for diverse, explainable multi-hop question answering},
  author={Yang, Zhilin and Qi, Peng and Zhang, Saizheng and Bengio, Yoshua and Cohen, William W and Salakhutdinov, Ruslan and Manning, Christopher D},
  journal={arXiv preprint arXiv:1809.09600},
  year={2018}
}

@article{lewis2020retrieval,
  title={Retrieval-augmented generation for knowledge-intensive nlp tasks},
  author={Lewis, Patrick and Perez, Ethan and Piktus, Aleksandra and Petroni, Fabio and Karpukhin, Vladimir and Goyal, Naman and K{\"u}ttler, Heinrich and Lewis, Mike and Yih, Wen-tau and Rockt{\"a}schel, Tim and others},
  journal={Advances in neural information processing systems},
  volume={33},
  pages={9459--9474},
  year={2020}
}

@inproceedings{jin2022survey,
  title={A survey on table question answering: recent advances},
  author={Jin, Nengzheng and Siebert, Joanna and Li, Dongfang and Chen, Qingcai},
  booktitle={China Conference on Knowledge Graph and Semantic Computing},
  pages={174--186},
  year={2022},
  organization={Springer}
}

@article{nan2022fetaqa,
  title={FeTaQA: Free-form table question answering},
  author={Nan, Linyong and Hsieh, Chiachun and Mao, Ziming and Lin, Xi Victoria and Verma, Neha and Zhang, Rui and Kry{\'s}ci{\'n}ski, Wojciech and Schoelkopf, Hailey and Kong, Riley and Tang, Xiangru and others},
  journal={Transactions of the Association for Computational Linguistics},
  volume={10},
  pages={35--49},
  year={2022},
  publisher={MIT Press One Broadway, 12th Floor, Cambridge, Massachusetts 02142, USA~…}
}

@article{zhao2024chat2data,
  title={Chat2data: An interactive data analysis system with rag, vector databases and llms},
  author={Zhao, Xinyang and Zhou, Xuanhe and Li, Guoliang},
  journal={Proceedings of the VLDB Endowment},
  volume={17},
  number={12},
  pages={4481--4484},
  year={2024},
  publisher={VLDB Endowment}
}

@article{naeem2024retclean,
  title={RetClean: Retrieval-Based Data Cleaning Using LLMs and Data Lakes},
  author={Naeem, Zan Ahmad and Ahmad, Mohammad Shahmeer and Eltabakh, Mohamed and Ouzzani, Mourad and Tang, Nan},
  journal={Proceedings of the VLDB Endowment},
  volume={17},
  number={12},
  pages={4421--4424},
  year={2024},
  publisher={VLDB Endowment}
}

@article{shankar2024docetl,
  title={DocETL: Agentic Query Rewriting and Evaluation for Complex Document Processing},
  author={Shankar, Shreya and Chambers, Tristan and Shah, Tarak and Parameswaran, Aditya G and Wu, Eugene},
  journal={arXiv preprint arXiv:2410.12189},
  year={2024}
}

@article{lin2024towards,
  title={Towards accurate and efficient document analytics with large language models},
  author={Lin, Yiming and Hulsebos, Madelon and Ma, Ruiying and Shankar, Shreya and Zeigham, Sepanta and Parameswaran, Aditya G and Wu, Eugene},
  journal={arXiv preprint arXiv:2405.04674},
  year={2024}
}

@article{patel2024lotus,
  title={Lotus: Enabling semantic queries with llms over tables of unstructured and structured data},
  author={Patel, Liana and Jha, Siddharth and Guestrin, Carlos and Zaharia, Matei},
  journal={arXiv preprint arXiv:2407.11418},
  year={2024}
}

@article{liu2024declarative,
  title={A declarative system for optimizing ai workloads},
  author={Liu, Chunwei and Russo, Matthew and Cafarella, Michael and Cao, Lei and Chen, Peter Baille and Chen, Zui and Franklin, Michael and Kraska, Tim and Madden, Samuel and Vitagliano, Gerardo},
  journal={arXiv preprint arXiv:2405.14696},
  year={2024}
}

@article{liu2024lost,
  title={Lost in the middle: How language models use long contexts},
  author={Liu, Nelson F and Lin, Kevin and Hewitt, John and Paranjape, Ashwin and Bevilacqua, Michele and Petroni, Fabio and Liang, Percy},
  journal={Transactions of the Association for Computational Linguistics},
  volume={12},
  pages={157--173},
  year={2024},
  publisher={MIT Press One Broadway, 12th Floor, Cambridge, Massachusetts 02142, USA~…}
}

@article{bai2023longbench,
  title={Longbench: A bilingual, multitask benchmark for long context understanding},
  author={Bai, Yushi and Lv, Xin and Zhang, Jiajie and Lyu, Hongchang and Tang, Jiankai and Huang, Zhidian and Du, Zhengxiao and Liu, Xiao and Zeng, Aohan and Hou, Lei and others},
  journal={arXiv preprint arXiv:2308.14508},
  year={2023}
}

@article{zhang2024benchmarking,
  title={Benchmarking the text-to-sql capability of large language models: A comprehensive evaluation},
  author={Zhang, Bin and Ye, Yuxiao and Du, Guoqing and Hu, Xiaoru and Li, Zhishuai and Yang, Sun and Liu, Chi Harold and Zhao, Rui and Li, Ziyue and Mao, Hangyu},
  journal={arXiv preprint arXiv:2403.02951},
  year={2024}
}

@article{adnan2024keyformer,
  title={Keyformer: Kv cache reduction through key tokens selection for efficient generative inference},
  author={Adnan, Muhammad and Arunkumar, Akhil and Jain, Gaurav and Nair, Prashant J and Soloveychik, Ilya and Kamath, Purushotham},
  journal={Proceedings of Machine Learning and Systems},
  volume={6},
  pages={114--127},
  year={2024}
}

@article{cai2024pyramidkv,
  title={Pyramidkv: Dynamic kv cache compression based on pyramidal information funneling},
  author={Cai, Zefan and Zhang, Yichi and Gao, Bofei and Liu, Yuliang and Li, Yucheng and Liu, Tianyu and Lu, Keming and Xiong, Wayne and Dong, Yue and Hu, Junjie and others},
  journal={arXiv preprint arXiv:2406.02069},
  year={2024}
}

@article{liu2024minicache,
  title={Minicache: Kv cache compression in depth dimension for large language models},
  author={Liu, Akide and Liu, Jing and Pan, Zizheng and He, Yefei and Haffari, Reza and Zhuang, Bohan},
  journal={Advances in Neural Information Processing Systems},
  volume={37},
  pages={139997--140031},
  year={2024}
}

@article{gu2024survey,
  title={A survey on llm-as-a-judge},
  author={Gu, Jiawei and Jiang, Xuhui and Shi, Zhichao and Tan, Hexiang and Zhai, Xuehao and Xu, Chengjin and Li, Wei and Shen, Yinghan and Ma, Shengjie and Liu, Honghao and others},
  journal={arXiv preprint arXiv:2411.15594},
  year={2024}
}

@article{zheng2023judging,
  title={Judging llm-as-a-judge with mt-bench and chatbot arena},
  author={Zheng, Lianmin and Chiang, Wei-Lin and Sheng, Ying and Zhuang, Siyuan and Wu, Zhanghao and Zhuang, Yonghao and Lin, Zi and Li, Zhuohan and Li, Dacheng and Xing, Eric and others},
  journal={Advances in Neural Information Processing Systems},
  volume={36},
  pages={46595--46623},
  year={2023}
}

@article{huang2024transform,
  title={Transform Table to Database Using Large Language Models},
  author={Huang, Zezhou and Guo, Jia and Wu, Eugene},
  journal={Proceedings of the VLDB Endowment. ISSN},
  volume={2150},
  pages={8097},
  year={2024}
}

@inproceedings{buss2023generating,
  title={Generating Data Augmentation Queries Using Large Language Models.},
  author={Buss, Christopher and Mosavi, Jasmin and Tokarev, Mikhail and Termehchy, Arash and Maier, David and Lee, Stefan},
  booktitle={VLDB Workshops},
  year={2023}
}

@article{zhuang2023toolqa,
  title={Toolqa: A dataset for llm question answering with external tools},
  author={Zhuang, Yuchen and Yu, Yue and Wang, Kuan and Sun, Haotian and Zhang, Chao},
  journal={Advances in Neural Information Processing Systems},
  volume={36},
  pages={50117--50143},
  year={2023}
}

@article{zhu2024autotqa,
  title={Autotqa: Towards autonomous tabular question answering through multi-agent large language models},
  author={Zhu, Jun-Peng and Cai, Peng and Xu, Kai and Li, Li and Sun, Yishen and Zhou, Shuai and Su, Haihuang and Tang, Liu and Liu, Qi},
  journal={Proceedings of the VLDB Endowment},
  volume={17},
  number={12},
  pages={3920--3933},
  year={2024},
  publisher={VLDB Endowment}
}

@article{caspari2024beyond,
  title={Beyond benchmarks: Evaluating embedding model similarity for retrieval augmented generation systems},
  author={Caspari, Laura and Dastidar, Kanishka Ghosh and Zerhoudi, Saber and Mitrovic, Jelena and Granitzer, Michael},
  journal={arXiv preprint arXiv:2407.08275},
  year={2024}
}

@misc{databricks-llm,
  author       = {Patrick Wendell and Eric Peter and Nicolas Pelaez and Jianwei Xie
                  and Vinny Vijeyakumaar and Linhong Liu and Shitao Li},
  title        = {{Introducing AI Functions: Integrating Large Language Models with Databricks SQL}},
  year         = {2023},
  month        = apr,
  day          = {17},
  howpublished = {\url{https://www.databricks.com/blog/2023/04/18/introducing-ai-functions-integrating-large-language-models-databricks-sql.html}},
  note         = {Accessed: 2025-06-22},
}

@misc{snowflake-llm,
  author       = {Arun Agarwal and Renee Huang},
  title        = {{Introducing Cortex AISQL: Reimagining SQL into AI Query Language for Multimodal Data}},
  year         = {2025},
  month        = June,
  day          = {3},
  howpublished = {\url{https://www.snowflake.com/en/blog/ai-sql-query-language/
}},
  note         = {Accessed: 2025-06-22},
}

@misc{duckdb-llm,
  author       = {Till Döhmen},
  title        = {{Introducing the prompt() Function: Use the Power of LLMs with SQL!}},
  year         = {2024},
  month        = {oct},
  day          = {17},
  howpublished = {\url{https://motherduck.com/blog/sql-llm-prompt-function-gpt-models/}},
  note         = {Accessed: 2025-06-22},
}

@misc{alloydb-llm,
  title        = {Perform intelligent SQL queries using AlloyDB AI query engine},
  author       = {{Google Cloud}},
  howpublished = {\url{https://cloud.google.com/alloydb/docs/ai/evaluate-semantic-queries-ai-operators}},
  note         = {Accessed: 2025-06-22; Last updated: 2025-06-11},
  year         = {2025},
}

\sigmod{\appendixtext{\newpage\input{appendix}}}

\end{document}